\documentclass[pdflatex,sn-mathphys]{sn-jnl}
\newcommand{\g}{{\mathsf{g}} }
\newtheorem{thm}{Theorem}[section]
\newtheorem{lem}[thm]{Lemma}

\newtheorem{remark}[thm]{Remark}
\newcounter{IssueCounter}
\newtheorem{Issue}[IssueCounter]{Issue}

\def\be {\begin{equation}}
\def\ee {\end{equation}}
\def\ba {\begin{eqnarray}}
\def\ea {\end{eqnarray}}
\def\bas {\begin{eqnarray*}}
\def\eas {\end{eqnarray*}}

\jyear{2021}%
\begin{document}

\title[Bifurcations of Riemann Ellipsoids]{Bifurcations of Riemann Ellipsoids}

\author[1]{\fnm{Mokhtari,} \sur{Fahimeh}}\email{f.mokhtari@vu.nl}

\author[2]{\fnm{Palaci\'an,} \sur{Jes\'us F.}}\email{palacian@unavarra.es}

\author[2]{\fnm{Yanguas,} \sur{Patricia}}\email{yanguas@unavarra.es}

\affil[1]{\orgdiv{Department of Mathematics}, \orgname{Vrije Universiteit Amsterdam}, \orgaddress{\street{Nieuw Universiteitsgebouw, De Boelelaan 1111}, \city{Amsterdam}, \postcode{1081 HV}, \country{The Netherlands}}}

\affil[2]{\orgdiv{Departamento de Estad{\'\i}stica, Inform\'atica y Matem\'aticas and Institute for Advanced Materials and Mathematics (INAMAT$^2$)}, \orgname{Universidad P\'ublica de Navarra}, \orgaddress{\street{Campus de Arrosadia s/n}, \city{Pamplona}, \postcode{31006}, \state{Navarra}, \country{Spain}}}

\abstract{We give an account of the various changes in the stability character in the five types of Riemann ellipsoids by establishing the occurrence of different quasi-periodic Hamiltonian bifurcations. Suitable symplectic changes of coordinates, that is, linear and non-linear normal form transformations are performed, leading to the characterisation of the bifurcations responsible of the stability changes. Specifically we find three types of bifurcations, namely, Hamiltonian pitchfork, saddle-centre and Hamiltonian-Hopf in the four-degree-of-freedom Hamiltonian system resulting after reducing out the symmetries of the problem. The approach is mainly analytical up to a point where non-degeneracy conditions have to be checked numerically. We also deal with the regimes in the parametric plane where Liapunov stability of the ellipsoids is accomplished. This strong stability behaviour occurs only in two of the five types of ellipsoids, at least deductible only from a linear analysis.}


\maketitle  
\section{Introduction}

One of the relevant problems Newton addresses in his Principia \cite{Principia} is the determination of the shape of the Earth when rotating around its axis. He assumes the Earth to be a homogeneous axisymmetric self-gravitating fluid slowly rotating around its axis and shows that rotation makes the body oblate. This investigation constitutes the first attempt at determining how rotation affects the shape of a body whose surface is not rigid and is the beginning of a fruitful line of research that notable scientists, such as MacLaurin or Jacobi follow. The models addressed by these authors assume that the fluid is stationary in a frame rotating with the body.
\\

In 1860, Dirichlet \cite{Dirichlet} starts a new line of research when considering non-rigid movements. He contemplates configurations whose motion in an inertial frame is a linear function of the coordinates. In particular, he considers
\begin{equation} 
x( t, y ) = F( t ) \, y, \qquad F( t ) \in SL( 3 ), 
\label{Dirichlet}
\end{equation}
where $x( t, y )$ is the position of a particle at time $t$ in a fixed reference system, $y$ are the coordinates of the particle in the reference system centred at the body and $SL( 3 )$ denotes the special linear group of degree $3$ over ${\mathbb R}$. Dirichlet addresses the problem of determining the conditions for these configurations to have an ellipsoidal figure at any moment \cite{Dirichlet,chandrasekhar1969ellipsoidal,fasso2001stability}. Assuming that the reference configuration of the fluid mass is a radius \(\rho\)-ball, the free surface of the fluid mass determined by \eqref{Dirichlet} is an ellipsoid with semi-axes 
\( \rho a_1, \rho a_2, \rho a_3 \), where \( a_i \) are the singular values of matrix \( F \), i.e. the square root of the eigenvalues of $F F^T$ or $F^T F$.
\\

It is Euler's equations that govern the dynamics of adiabatic and inviscid flows. They are a set of quasi-linear partial differential equations that correspond to the Navier-Stokes equations with no viscosity and no thermal conductivity. Dirichlet's philosophy is using Lagrangian formulation to reduce Euler's equations to a system of ordinary differential equations such that the position of the particle in the ellipsoid at any time is a linear homogeneous function of its initial position. He proves that (\ref{Dirichlet}) forms an invariant subsystem of Euler's equations of fluid dynamics. It is his student, Dedekind, who publishes Dirichlet's work posthumously and completes some results.
\\

Riemann \cite{Riemann} subsequently continues Dirichlet's work and reformulates the equations of motion (\ref{Dirichlet}) in a convenient way to study steady asymmetric configurations. Riemann determines and classifies all possible relative equilibrium conditions and analyses the stability of the corresponding equilibria. They are characterised by different relations among the angular velocities and the figures' semi-axes. There are five kinds of these states called {\em Riemann ellipsoids}. They are denoted by $S_2$, $S_3$, I, II, III and are motions of the type (\ref{Dirichlet}) such that 
\be 
F(t) = \exp( t \, \Omega_l ) \, A \exp( -t \, \Omega_r ), 
 \label{Riemann_svd}
\ee
where  
\[ 
A = \mbox{diag} \left( a_1, a_2, a_3 \right), \,  \Omega_l, \Omega_r \,\, \mbox{being constant antisymmetric $( 3 \times 3 )$-matrices.} 
 \]
The equilibrium form does not perform a rigid motion, since it is a composition of an internal rotation together with a stretch along the principal axes and a spatial rotation such that the free surface retains a rotating ellipsoidal shape. Thence, Riemann ellipsoids are steady states of an ideal incompressible homogeneous self-gravitating fluid mass that has an ellipsoidal shape. The fluid particles describe either periodic or quasi-periodic rosette-shaped motions. In this latter case they depend on the two angular frequencies, $\omega_l$ and $\omega_r$, respectively associated to matrices $\Omega_l$ and $\Omega_r$. For ellipsoids $S_2$ and $S_3$ vectors $\omega_l$ and $\omega_r$ are parallel to the same principal axis of the ellipsoid; this axis is either the shortest for $S_3$ or the middle one for $S_2$. For ellipsoids I, II and III, both $\omega_l$ and $\omega_r$ lie in one of the two principal planes containing the longest ellipsoid's principal axis. Given an $S$-type Riemann ellipsoid we say it is {\em co-parallel} when the dot product of $\omega_l$ and $\omega_r$ is positive, otherwise we call it {\em counter-parallel}. All $S_3$-ellipsoids are counter-parallel, while the $S_2$-ellipsoids may be co-parallel or counter-parallel.
\\

Ensuing contributions by Liapunov, Poincar\'e and Cartan enhance knowledge of this problem. Chandrasekhar \cite{chandrasekhar1965, chandrasekhar1966, chandrasekhar1969ellipsoidal} enlarges and completes the work initiated by the previous authors, studying the linear (its spectral version) stability of the ellipsoids by making use of the Virial Theorem of Mechanics. In particular, he presents in a unified way Dirichlet and Dedekind's approaches, MacLaurin spheroids, Jacobi and Dedekind ellipsoid and Riemann ellipsoids. In \cite{lebovitz1996}, Lebovitz explains why some of Riemann's conclusions on stability are incorrect, in contrast to Chandrasekhar's. Rosensteel \cite{Rosensteel1, Rosensteel2, Rosensteel3} reformulates the problem from a symplectic geometry point of view and also adapts it to nuclear physics models. Moreover, applying symplectic geometry Lewis \cite{Lewis} gives an account of the stability of MacLaurin spheroids, already accomplished by Riemann and Chandrasekhar. Paper \cite{roberts1999symmetries} continues with the Hamiltonian geometric approach to Riemann's classification of equilibrium ellipsoids. Related amendments to Riemann's conclusions are included in \cite{Marshalek}. The Hamiltonian formulation is also discussed in \cite{MoLeBi}. A recent differential geometric approach can be seen in \cite{OlmosSousa}, where the authors deal with the non-linear stability of some particular Riemann ellipsoids that can be formulated as three-degree-of-freedom Hamiltonian systems. Indeed, there is  a wide bibliography on the theme. For a historical account see, for instance \cite{chandrasekhar1969ellipsoidal}. A good review paper with some new results on the dynamics of self-gravitating liquid and gas ellipsoids is \cite{BorisovKilinMamaev}. 
\\

In \cite{fasso2001stability, fasso2014erratum} Fass\`o and Lewis perform a thorough analysis of the stability of Riemann ellipsoids, improving and completing previous studies appearing in the literature, in particular, they amend some of Chandrasekhar's findings. Especially, they notice that the regions of known instability of the ellipsoids of types II and III are substantially smaller than those sought by Chandrasekhar. As a first step, Fass\`o and Lewis perform a linear (the so-called spectral) stability analysis, as they focus on the eigenvalues of the linearisation matrix. In a second step they deal with the non-linear stability of the ellipsoids, applying Nekhoroshev theory on exponentially long-time stability of solutions. The approach followed by Fass\`o and Lewis can be interpreted as semi-numerical: when possible, the calculations are carried out symbolically, but the determination of the bifurcation curves is done numerically. In this paper we continue their work and deepen the analysis of the dynamics and stability of Riemann ellipsoids. Our notations and calculations are based in theirs. We have reproduced the material from \cite{fasso2001stability} needed for the understanding of the present manuscript. 
\\

We present a systematic study of the bifurcations arising for Riemann ellipsoids. This problem is far from trivial, as the computations we perform imply manipulating large formulae. The Hamiltonian function accounts for a system of four degrees of freedom and it depends on an incomplete elliptic integral that is handled in closed form, that is, without resorting to numerical approximations. Our study is mostly analytical, said in other words, in closed form: the linear and non-linear normal forms and related transformations are all analytical, but some final checks proving non-degeneracy (normally that an expression does not vanish on a bifurcation line) should be made numerically, as we shall mention adequately. Moreover, all calculations have been carried out with {\sc Mathematica}, using integer arithmetic. Notice that the numerical testing does not reduce the rigour of our analysis, and we can state the occurrence of different bifurcations by means of theorems. Regarding previous approaches dealing with bifurcations analysis for the four-degree-of-freedom Riemann ellipsoids we only know of a recent reference \cite{benavides}, which is mainly numerical. On our side we prove that there are up to three types of quasi-periodic bifurcations, the most abundant being the Hamiltonian-Hopf bifurcations, that arise for the ellipsoids of types I, II and III. However, the ellipsoids $S_2$ experience a Hamiltonian-pitchfork bifurcation, whereas type II-ellipsoids undergo a saddle-centre bifurcation. All these bifurcations take place in the parametric plane determined by the two essential parameters, the same plane as the one considered in \cite{chandrasekhar1969ellipsoidal} and \cite{fasso2001stability}. 
\\

Quasi-periodic bifurcations occurring in the Hamiltonian context have been extensively studied by Broer, Han{\ss}man and co-workers, and we follow the monograph \cite{hanssmann2006local} to establish our results on the bifurcations of Riemann ellipsoids.
\\

In future we will provide the non-linear stability analysis of the different Riemann ellipsoids. By this we mean stability of formal type, the so-called Lie stability, which in particular generalises Nekhoroshev-type stability for equilibrium points of elliptic character, see \cite{Carcamo2021}. 
\\

The paper is structured as follows. Section \ref{Formulation} presents the Hamiltonian formulation of the system. Section \ref{Equilibria} provides the equilibria of the problem and their regions of existence. The analysis of the bifurcations can be found in the subsequent sections. The stability of $S_2$-ellipsoids is dealt with in Section \ref{SectionS2}. The main finding concerning these ellipsoids is a Hamiltonian pitchfork bifurcation, which is studied through a non-linear approach. Ellipsoids of type $S_3$ are Liapunov stable, as shown in Section \ref{SectionS3}. This result is due to Riemann, we simply recover it for completeness. In Section \ref{TypeI} we deal with type-I ellipsoids, analysing the particular case of irrotational ellipsoids (one of the two angular frequencies vanishes and the Hamiltonian system can be reduced by one degree of freedom). Concretely we study its stability and make the observation that the transition from stability to instability is done by means of two Hamiltonian-Hopf bifurcation points. There is a saddle-centre bifurcation related to type-II ellipsoids that is described in Section \ref{TypeII}. For type-III ellipsoids there is a Hamiltonian-Hopf bifurcation that is explained in Section \ref{TypeIII}. Finally, there is a global bifurcation involving the $S_2$ and type-III ellipsoids that is described in Section \ref{global}. This bifurcation corresponds to a global viewpoint of the pitchfork bifurcation tackled in Section \ref{SectionS2}. The main achievements and some remarks regarding possible future approaches are outlined in Section \ref{conclusions}. Appendix \ref{C1C2} provides explicit expressions of two improper integrals in terms of two incomplete elliptic integrals that are required in our approach. In Appendix \ref{CoefficientsL} we place the essential formulae related to the regime of the $S_2$-ellipsoids where the pitchfork bifurcation arises. Finally, Appendix \ref{Linearization} is devoted to the description of Markeev's procedure to compute the linear normal form of a Hamiltonian system corresponding to an elliptic equilibrium. We also collect the entries of the transformation matrix used to deal with the pitchfork bifurcation analysis.  
\\

The calculations presented in Appendix \ref{C1C2} are crucial for the achievements obtained on the Riemann ellipsoids. Actually, the determination of these two functions allows us to explicitly write the coordinates of all ellipsoids, as well as the sets in the parametric plane where the ellipsoids are properly defined. Moreover, excepting the curves associated to Hamiltonian-Hopf bifurcations that, much as determined analytically, are approximated by applying numerical techniques, making the approach practical, the rest of lines and points in the parametric plane corresponding to changes in stability have been obtained in closed form. This is in part due to the improper integrals provided in Appendix \ref{C1C2}. We shall give details on this feature when dealing with the study performed in the five ellipsoids.  
\\

Our analysis has not pursued the heavy task of seeking all stability regions and bifurcation curves in the parametric plane (some portions of it certainly being very subtle) corresponding to ellipsoids of types I, II and III. Apart from the analysis carried out in Sections \ref{TypeI}, \ref{TypeII} and \ref{TypeIII} on the bifurcations, we have checked the linear stability in the regions encountered in \cite{fasso2001stability}, by simply picking samples in different regions of the parametric plane. Our results agree with the ones obtained by Fass\`o and Lewis. Moreover, according to our appraisals, the bifurcations of these three types of ellipsoids not considered in our study seem to be of Hamiltonian-Hopf type, although we have not performed a further study about this. 
\\

The bifurcations accounted for in Sections \ref{SectionS2}, \ref{TypeI}, \ref{TypeII} and \ref{TypeIII} have to be understood as the dynamical behaviour of a single ellipsoid. Noticing that a specific point in the parametric plane represents a Riemann ellipsoid with its type of stability, for such a Riemann ellipsoid, the occurring bifurcations of invariant (KAM) tori of various dimensions have to be thought as the typical bifurcations expected to take place in a Hamiltonian system of four degrees of freedom. The KAM tori change their stability depending on the bifurcation they experience. From this viewpoint the richness in the dynamic behaviour of the Riemann ellipsoids is evident, a fact already seen by Chandrasekhar \cite{chandrasekhar1969ellipsoidal} and Fass\`o and Lewis \cite{fasso2001stability}. The case of the bifurcation described in Section \ref{global} is different because it involves two types of ellipsoids, namely, $S_2$ and type III.
\\

In general the computations in the work are lengthy; that is why in the text we have written down the most abridged ones, while the rest is comprised in a {\sc Mathematica} 13.2 file attached to this manuscript. In this file we have included the derivation of all formulae providing detailed explanations. The calculations performed in the file are usually quite involved and they often need careful simplification rules towards getting compact expressions. In this respect there is a clear distinction between the treatment of $S$-ellipsoids where the formulae are long but manageable and the treatment of types I, II and III where the computations become enormous, although they are affordable to extract useful information regarding the bifurcations of the problem. As well, we have checked our findings with care, both analytically and numerically. The {\sc Mathematica} program runs on medium-scale computers, such as laptops with 2,9 GHz Intel Core i7 processor and 16Gb of memory.
\\  

\section{Formulation of the problem}
\label{Formulation}

Following the detailed description appearing in \cite{fasso2001stability} we start by summarising the essential steps and notations to state the formulation of the problem. 
\\

Riemann uses the singular value decomposition of matrices to formulate system (\ref{Dirichlet}) with $F$ in  (\ref{Riemann_svd}). Given $F\in SL( 3 )$, in any singular value decomposition there exist matrices \( U_l \) and \( U_r \) such that $F = U_l A U_r^{T}$, where \( A \) is a diagonal matrix (the singular matrix) whose diagonal elements are the eigenvalues of \( F \). The ordering fixed for the elements in matrix \( A \) is $a_1 \ge a_2 \ge a_3 > 0$. In like manner, $\Omega_l = U_l^T \dot{U}_l$ and $\Omega_r = U_r^T \dot{U}_r$.
\\

At this point we introduce the potential function
\[\mathcal{V}( A ) = -2 \pi \g \int_{0}^{\infty} \left( ( s + a_1^2 ) ( s + a_2^2 ) ( s + a_3^2 ) \right)^{-1/2} ds, \]
where $\g$ denotes the gravitational constant.
\\

Dirichlet shows that (\ref{Dirichlet}) is a solution of the hydrodynamical equation for an ideal incompressible homogeneous self-gravitating fluid with constant pressure at the boundary when
\be
\mathbb{P}_{F} \left[ \ddot{F} + U_l \mathcal{V}'(A) U_r^T \right] = 0,
\label{Dirichlet_solution}
\ee
where $\mathcal{V}' = \mbox{diag} \left( 
\frac{\partial \mathcal{V}}{\partial a_1}, 
\frac{\partial \mathcal{V}}{\partial a_2}, 
\frac{\partial \mathcal{V}}{\partial a_3} \right)$ and $\mathbb{P}_F( G ) = G - \frac{1}{3} \langle G, F \rangle F^{-T}$, for any $G \in L( 3 )$, with $L( 3 )$ denoting the group of motions of the three-dimensional Euclidean space and $\langle \,, \rangle$ denoting the standard inner product in $L( 3 )$.
\\

After Riemann's reformulation, the previous condition is translated into
\be
\mathbb{P}_{A}\left[\ddot{A}
+ 2( \Omega_l {\dot A} - {\dot A} \Omega_r )
+ {\dot\Omega_l} A - A {\dot \Omega_r}
+ \Omega_l^2 A - 2 \Omega_l A \Omega_r + A \Omega_r^2
+ {\mathcal V}'( A )
\right] = 0.
\label{Riemann_solution}
\ee
This equation determines a second-order differential system on the manifold $\mathcal{A}\times SO( 3 ) \times SO( 3 )$, where
\[
{\mathcal A} = \left\{ \mbox{diag} \left( a_1, a_2, ( a_1 a_2 )^{-1} \right) \, : \, a_1 > a_2 > a_1^{-1/2} \right\}.
\]
Riemann's condition (\ref{Riemann_solution}) is equivalent to the restriction of Dirichlet's condition (\ref{Dirichlet_solution}) to the submanifold
\[ 
Q = \left\{ F\in SL( 3 ) \, : \, a_1 > a_2 > a_1^{-1/2} \right\}.
\]
As shown in \cite{fasso2001stability}, the two conditions are related by a four-to-one covering.
\\

In the following we present Riemann's equation (\ref{Riemann_solution}) in Hamiltonian form on the cotangent bundle of $\mathcal{A} \times SO( 3 ) \times SO( 3 )$. As a first step, a diffeomorphism is established between $\mathcal{A}$ and 
\ba
\label{eq:B}
\mathcal{B} = \left\{ b = ( b_1, b_2 ) \in \mathbb{R}^2 \, : \, b_1 > b_2 > \frac{1}{\sqrt{b_1}} > 0 \right\},
\ea
where $( b_1, b_2 )$ are the first two singular values of $F$ and $b_3 = ( b_1 b_2 )^{-1}$. After due identifications (see \cite{fasso2001stability}) it is possible to pass to the sixteen-dimensional manifold ${\mathcal M} = {\mathcal B} \times {\mathbb R}^2 \times ( SO( 3 ))^2 \times ( {\mathbb R}^3 )^2$, which is diffeomorphic to the cotangent bundle of $\mathcal{A} \times SO( 3 ) \times SO( 3 )$. 
\\

Proposition 2 in \cite{fasso2001stability} establishes that Riemann's equation (\ref{Riemann_solution}) on $\mathcal{A} \times SO( 3 ) \times SO( 3 )$ is equivalent to the following Hamiltonian defined on ${\mathcal M}$
\bas
H( b, c, U, m ) = \mbox{$\frac{1}{2}$} c \cdot
\mathcal{K}( b ) \cdot c + \mbox{$\frac{1}{2}$} m \cdot \mathcal{J}( b ) \cdot m + \mathcal{V}( b ),
\eas
with 
\bas
	\mathcal{K}( b ) = \frac{1}{b_1^2 b_2^2 + b_1^2 b_3^2 + b_2^2 b_3^2}
	\left( 
 \begin{array}{cc} 
 b_1^2 ( b_2^2 + b_3^2 ) & - b_3 \\ \noalign{\medskip} 
 -b_3 & b_2^2 ( b_1^2 + b_3^2 ) 
 \end{array}
	\right), 
\eas 
\bas 
\mathcal{J}( b ) = 
\left(
	\begin{array}{cccccc}
	   \frac{b_2^2 + b_3^2}{( b_2^2 - b_3^2 )^2} & 0 & 0 & \frac{2\, b_2 b_3}{( b_2^2 - b_3^2 )^2} & 0 & 0 \\
		 0 & \frac{b_1^2 + b_3^2}{( b_1^2 - b_3^2 )^2} & 0 & 0 & \frac{2\, b_1 b_3}{( b_1^2 - b_3^2 )^2} & 0 \\
	   0 & 0 & \frac{b_1^2 + b_2^2}{( b_1^2 - b_2^2 )^2} & 0 & 0 & \frac{2\, b_1 b_2}{( b_1^2 - b_2^2 )^2} \\
	   \frac{2\, b_2 b_3}{( b_2^2 - b_3^2 )^2} & 0 & 0 & \frac{b_2^2 + b_3^2}{( b_2^2 - b_3^2 )^2} & 0 & 0 \\
	   0 & \frac{2\, b_1 b_3}{( b_1^2 - b_3^2 )^2} & 0 & 0 & \frac{b_1^2 + b_3^2}{( b_1^2 - b_3^2 )^2} & 0 \\
	   0 & 0 & \frac{2\, b_1 b_2}{( b_1^2 - b_2^2 )^2} & 0 & 0 & \frac{b_1^2 + b_2^2}{( b_1^2 - b_2^2 )^2} \\
	\end{array}
\right)
\eas
and ${\mathcal V}$ the self-gravitational potential 
\[ {\mathcal V}( b ) = -\frac{4 \pi \g}{\sqrt{b_1^2 - b_3^2}} \, F \left( \arccos \left( \frac{b_3}{b_1} \right) \bigg\vert \, \frac{b_1^2 - b_2^2}{b_1^2 - b_3^2} \right), 
\]
where $F( \phi \mid k )$ stands for the incomplete elliptic integral of the first kind, for \( \phi \in ( 0, \frac{\pi}{2} ) \) and \( k \in ( 0, 1 ) \), i.e.
\bas
F( \phi \mid k ) = \int_0^{\phi} \left( 1 - k^2 \sin^2( \theta ) \right)^{-1/2} d\theta.
\eas
Notice that this is not a standard notation. It corresponds to the way {\sc Mathematica} and also \cite{fasso2001stability} handle it. Hamiltonian $H$ is invariant under a symplectic action of $SO( 3 ) \times SO( 3 )$ on ${\mathcal M}$ and, consequently, the reduced space is the eight-dimensional symplectic manifold given by  
\ba\label{eq:p}
 P_{L,R} = \mathcal{B} \times \mathbb{R}^2 \times
  \left( S^2_{L} \times S^2_{R} \right)
\ea
(see Proposition 3 in \cite{fasso2001stability}), where $S^{2}_\rho$ is the sphere of radius $\rho$, $L = \| \eta_l \|$, $R = \| \eta_r \|$ are fixed,  $m = ( \eta_l, \eta_r) \neq ( 0, 0 )$ with $\eta_l = ( m_1, m_2, m_3 ) \in \mathbb{R}^3$ and $\eta_r = ( m_4, m_5, m_6 ) \in \mathbb{R}^3$. Vectors $\eta_l$ and $\eta_r$ refer to the angular momentum and circulation (or vorticity or angular velocity) vectors, respectively. They correspond to the $\omega_l$ and $\omega_r$ of the introduction. 
The reduced Hamiltonian is 
\[
\begin{array}{ccccc}
H &:& P_{L, R} & \longrightarrow & \mathbb{R} \\
&&(b, c, \eta_l, \eta_r) & \mapsto & H( b, c, \eta_l, \eta_r ),
\end{array}\]
where $c = ( c_1, c_2 )$ are the momenta conjugate to $b = ( b_1, b_2 )$. Besides, $\eta_l, \eta_r$ have the following Poisson structure:
\[
\begin{array}{lcllcllcl}
\{ m_1, m_2 \} &=& m_3, \quad
\{ m_1, m_3 \} &=& -m_2, \quad
\{ m_2, m_3 \} &=& m_1,\\[1ex]
\{ m_4, m_5 \} &=& m_6, \quad
\{ m_4, m_6 \} &=& -m_5, \quad
\{ m_5, m_6 \} &=& m_4,
\end{array}
\]
with the rest of the Poisson brackets equal to zero.
\\

Hence, the simplified form of the Hamiltonian defined on the manifold $P_{L, R}$ is given by
\begin{equation}
\begin{array}{lcl}
	 H( b, c, m ) &=& \displaystyle
\frac{( b_2^2 + b_3^2 ) b_1^2 c_1^2 + ( b_1^2 + b_3^2 ) b_2^2 c_2^2 - 2\, b_3 c_1 c_2}{2\, (b_1^2 b_2^2 + b_1^2 b_3^2 + b_2^2 b_3^2)}
\\[3ex]&& \displaystyle
	+ \, \frac{( b_1^2 + b_2^2 ) ( m_3^2 + m_6^2 ) + 4\, b_1 b_2 m_3 m_6}{2\, ( b_1^2 - b_2^2 )^2}
\\[3ex]&& \displaystyle
	+ \, \frac{( b_1^2 + b_3^2 ) ( m_2^2 + m_5^2 ) + 4\, b_1 b_3 m_2 m_5}{2\, ( b_1^2 - b_3^2 )^2}
\\[3ex]&& \displaystyle
	+ \, \frac{( b_2^2 + b_3^2 ) ( m_1^2 + m_4^2 ) + 4\, b_2 b_3 m_1 m_4}{2\, ( b_2^2 - b_3^2 )^2} 
\\[3ex]&& \displaystyle
    - \, \frac{4 \pi \g}{\sqrt{b_1^2 - b_3^2}} F \left(\arccos \left( \frac{b_3}{b_1} \right) \bigg\vert \, \frac{b_1^2 - b_2^2}{b_1^2 - b_3^2} \right).
	\end{array}
 \label{Hamiltonian}
 \end{equation}
Hamiltonian $H( b, c, m )$ represents the Hamiltonian function in the coordinates $b_1, b_2$ with respective conjugate momenta $c_1, c_2$, and the three-dimensional vectors $\eta_l = ( m_1, m_2, m_3 )$ and $\eta_r = ( m_4, m_5, m_6 )$. Notice that $b_3$ is related to $b_1$, $b_2$ through the constraint $b_1 b_2 b_3 = 1.$ The reduced system has $4$ degrees of freedom and reads as
\begin{equation}
\left\{ \hspace*{-0.1cm}
\begin{array}{rcl}
\displaystyle \frac{d b}{d t} &=& \displaystyle\frac{\partial H}{\partial c}, \\[1.6ex]
\displaystyle \frac{d c}{d t} &=& \displaystyle-\frac{\partial H}{\partial b}, \\[1.6ex]
\displaystyle \frac{d \eta_l}{d t} &=& \displaystyle \eta_l \times \frac{\partial H}{\partial \eta_l}, \\[1.6ex]
\displaystyle \frac{d \eta_r}{d t} &=& \displaystyle\eta_r \times \frac{\partial H}{\partial \eta_r}, \\[1.6ex]
\end{array} \right.
\label{reduced_system}
\end{equation}
except one case, the so-called {\em irrotational} ellipsoid, where either $\eta_l = 0$ or $\eta_r = 0$. In this case the Hamiltonian system has 3 degrees of freedom and the reduced space becomes $P_{L} = \mathcal{B} \times \mathbb{R}^2 \times S^2_{L}$ or $P_{R} = \mathcal{B} \times \mathbb{R}^2 \times S^2_{R}$.

\section{The equilibria: Riemann ellipsoids}
\label{Equilibria}

This section is intended to the introduction of the five types of Riemann ellipsoids, which are the equilibria of the reduced system (\ref{reduced_system}). Most of the formulae we present are provided in Section 3 of \cite{fasso2001stability}; in particular, see Proposition 4 and lemmas 3 and 4 (in Appendix A). The expressions that follow are key to the development of the rest of the paper. Notice that $c^\ast = ( 0, 0 )$ at the equilibrium and we denote $b^\ast = ( b_1^\ast, b_2^\ast, b_3^\ast )$ also at the equilibrium. 
\\

With the aim of establishing the existence of the different types of equilibria in the space ${\mathcal B}$ the following functions are introduced:
\begin{equation}
\begin{array}{ccl}
G( x, y, z ) &=& x^2 ( y^2 - z^2 ) C_1( x, y, z ) + ( y^2- 4\, z^2 ) \big( z^2 C_1( x, y, z ) \\[0.8ex]
&& \hspace*{5.5cm} + \, C_2( x, y, z ) \big), \\[1.5ex]
D( x, y, z ) &=& x^2 ( y^2 - z^2 ) + z^2 ( 4\, z^2 - y^2 ), \\[1.5ex]
G^{S}_{\pm}( x, y, z ) &=& \displaystyle\frac{( x \mp z )^4}{x z}
\Big( \big( x y^2 z \pm( x^2 y^2 - x^2 z^2 + y^2 z^2 ) \big) C_1( x, y, z ) 
\\[0.8ex]
&&\hspace*{1.6cm} + \, ( x z \pm y^2 ) C_2( x, y, z ) \Big),
\\[1.5ex] 
N^S_{\pm}(x, y, z) &=& \frac{1}{2} \left( \sqrt{G^S_+( x, y, z )} \pm \sqrt{G^S_- ( x, y, z )} \right), \\[1.5ex]
G^{R}_{\pm}( x, y, z ) &=&
\displaystyle
( y \mp z )^4
\left( x^2 - ( y \pm 2\, z )^2 \right)
\frac{( x^2 - z^2 ) G( x, y, z )}{( x^2 - y^2 ) D( x, y, z )}, \\[2.4ex]
N^R_{\pm}( x, y, z ) &=& \frac{1}{2} \Big(\sqrt{G^R_+( x, y, z )} \pm \sqrt{G^R_-( x, y, z )} \Big),
\end{array}
\label{eq:Gs}
\end{equation}
where 
\bas
\begin{array}{ccc}
C_{1}( x, y, z ) &=& \displaystyle 2 \pi \g \int_{0}^{\infty} \left( ( s + x^2 )( s + y^2 )( s + z^2 ) \right)^{-3/2} s \, ds,
\\[2.4ex]
C_{2}( x, y, z ) &=& \displaystyle 2 \pi \g \int_{0}^{\infty} \left( ( s + x^2 )( s + y^2 )( s + z^2 ) \right)^{-3/2} s^2 \, ds.
\end{array}
\label{eq:c1}
\eas
These integrals have been treated numerically in \cite{chandrasekhar1969ellipsoidal, fasso2001stability}. In Appendix \ref{C1C2} we provide analytical expressions of them. \\

Using these functions the domains of existence of the five equilibria are given by
\[
\begin{array}{lcl}
{\mathcal B}_{S_2} &=& \Big\{ b \in {\mathcal B} \, : \, G_-^S( b_1^\ast, b_2^\ast, b_3^\ast ) \geq 0 \Big\}, \nonumber\\[1.5ex]
{\mathcal B}_{S_3} &=& \Big\{ b \in {\mathcal B} \, : \, G_+^S( b_1^\ast, b_3^\ast, b_2^\ast ) \geq 0 \Big\}, \nonumber\\[1.5ex]
{\mathcal B}_{\rm I} &=& \Big\{ b \in {\mathcal B} \, : \, b_1^\ast \leq 2 \, b_2^\ast - b_3^\ast \Big\},
\nonumber\\[1.5ex]
{\mathcal B}_{\rm II} &=& \Big\{ b \in {\mathcal B} \, : \, b_1^\ast \geq 2\, b_2^\ast + b_3^\ast, \,\, D( b_1^\ast, b_3^\ast, b_2^\ast) < 0 \Big\},
\nonumber\\[1.5ex]
{\mathcal B}_{\rm III} &=& \Big\{ b\in {\mathcal B} \, : \, b_1^\ast \geq b_2^\ast + 2\, b_3^\ast, \,\, G( b_1^\ast, b_2^\ast, b_3^\ast ) > 0 \Big\}.
\end{array}
\]
They are represented in the subsequent sections. There is an overlapping among the different regions, despite the fact that there exists a small portion of the parametric plane where there are no Riemann ellipsoids. Fig. \ref{fig:mesh1} contains the superposition of the five regions in the parametric plane $b_2^\ast/b_1^\ast$--$b_3^\ast/b_1^\ast$.
\\

\begin{figure}[ht]
    \centering
    \includegraphics[width=0.6\textwidth]{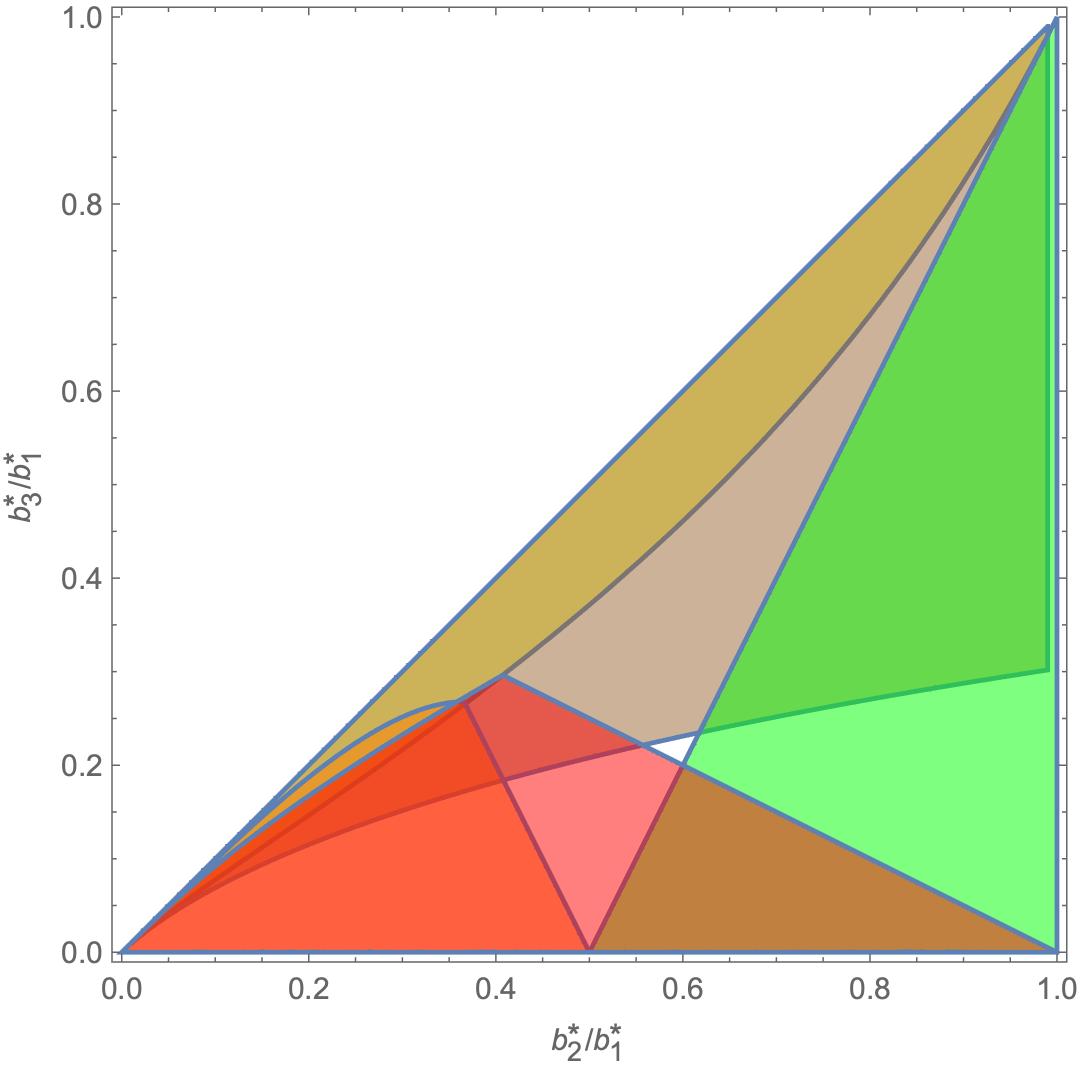}
    \caption{Superposition of the regions of existence of the five Riemann ellipsoids in the plane $b_2^\ast/b_1^\ast$--$b_3^\ast/b_1^\ast$}
    \label{fig:mesh1}
\end{figure}
The five types of Riemann ellipsoids appear in Table \ref{table:nonlin}, together with their regions of existence and their coordinates in the reduced space. The canonical basis of ${\mathbb R}^3$ is denoted as $\{ e_1, e_2, e_3 \}$. 
\\
\begin{table}[ht]
\caption{The first column shows the five types of Riemann ellipsoids. In the second column regions of existence in the plane $b_2^\ast/b_1^\ast$--$b_3^\ast/b_1^\ast$ are indicated. The third column accounts for the vectors $\mu^{\pm}_\alpha( b^{\ast} )$
standing for $\eta_l( b^{\ast} )$, $\eta_r( b^{\ast} )$ such that the coordinates in $S_L^2 \times S_R^2$ become $( \mu^+_\alpha( b^{\ast} ), \mu^-_\alpha( b^{\ast} ))$ or $( \mu^-_\alpha( b^{\ast} ), \mu^+_\alpha( b^{\ast} ))$.  Hence, the coordinates of an ellipsoid of type $\alpha$ in $P_{L, R}$ are $( b^\ast, 0, \mu^{+}_\alpha( b^{\ast} ), \mu^{-}_\alpha( b^{\ast} ))$ or $( b^\ast, 0, \mu^{-}_\alpha( b^{\ast} ), \mu^{+}_\alpha( b^{\ast} ))$} 
\centering 
\begin{tabular}{c c c} 
\hline\hline 
Type ($\alpha$) & Region $( b_2^\ast/b_1^\ast, b_3^\ast/b_1^\ast )$ & $\mu^{\pm}_\alpha( b^{\ast} )$ \\ [0.5ex] 
\hline \\[-2.3ex]
$S_2$ & ${\mathcal B}_{S_2}$ & $N^S_{\pm}( b_1^\ast, b_2^\ast, b_3^\ast ) e_2$ \\[2ex] 
$S_3$ & ${\mathcal B}_{S_3}$ & $N^S_{\pm}( b_1^\ast, b_3^\ast, b_2^\ast ) e_3$ \\[2ex]
I & ${\mathcal B}_{\rm I}$ & $N^R_{\pm}( b_1^\ast, b_3^\ast, b_2^\ast ) e_1 + N^R_{\pm}( b_3^\ast, b_1^\ast, b_2^\ast ) e_3$ \\[2ex]
II & ${\mathcal B}_{\rm II}$ & $N^R_{\pm}( b_1^\ast, b_3^\ast, b_2^\ast ) e_1 + N^R_{\mp}( b_3^\ast, b_1^\ast, b_2^\ast ) e_3$ \\[2ex]
III & ${\mathcal B}_{\rm III}$ & $N^R_{\pm}( b_1^\ast, b_2^\ast, b_3^\ast ) e_1 + N^R_{\mp}( b_2^\ast, b_1^\ast, b_3^\ast ) e_2$ \\[1ex] 
\hline 
\end{tabular}
\label{table:nonlin} 
\end{table}

\begin{remark}
The reduced system is invariant under both a ${\mathbb Z}_2$ and a ${\mathbb Z}_4$ action. A Riemann ellipsoid can be identified with the $({\mathbb Z}_4 \times {\mathbb Z}_2)$-orbit of an equilibrium of the reduced system, and it can consist of eight, four, or two equilibrium points on the reduced phase space, this depending on the number of zeroes vectors $\eta_l$, $\eta_r$ have; see more details in Proposition 4 of \cite{fasso2001stability}. As ${\mathbb Z}_2$, ${\mathbb Z}_4$ are discrete symmetries, their application to further reduce the system would introduce singularities in the reduced space. In view of this, we do not reduce the system further and work with regular reduction techniques. Finally, following Fass\`o and Lewis, we distinguish between relative equilibria whose projections in $S_L^2 \times S_R^2$ have coordinates $( \mu^+_\alpha( b^{\ast} ), \mu^-_\alpha( b^{\ast} ))$ or $( \mu^-_\alpha( b^{\ast} ), \mu^+_\alpha( b^{\ast} ))$, calling them adjoint equilibria.
\label{adjoint}
\end{remark}

In the following sections we describe the bifurcations of the equilibria. As a first step, their linear stability is determined. For that, we calculate the associated symplectic linear normal form and see that the equilibrium's linearisation matrix is diagonalisable \cite{CushmanBurgoyne, LaubMeyer1974}. Here we follow Markeev's procedure \cite{Markeev} to bring the linear Hamiltonian system (i.e. the one corresponding to the quadratic terms of the Hamilton function) to diagonal form. The algorithm is described in Appendix \ref{Linearization} and is designed for elliptic equilibria in Hamiltonian systems. The next step is the analysis of the non-linear stability and the bifurcations. We start by studying the stability of $S_2$-ellipsoids.
\\

\section{Stability of $S_2$-ellipsoids and the quasi-periodic pitchfork bifurcation in $\mathcal{B}_{S_2}$}
\label{SectionS2}

The linear stability analysis of the $S_2$-ellipsoids is performed analytically without particularising for specific values of $( b_2^\ast/b_1^\ast, b_3^\ast/b_1^\ast )$ on a grid of points in the parametric plane, albeit the expressions are quite big. Nevertheless, the computations are much easier for these ellipsoids and for $S_3$ than they are for types I, II and III. The reason stems from the discrete symmetries of the problem and from the fact that the angular frequencies $\eta_l$ and $\eta_r$ are parallel to the same principal axis of the ellipsoid. This implies that the Hessian matrix and the associated linearisation matrix contain several zero blocks. \\

Recall that the region of existence for $S_2$-ellipsoids is defined by
\bas
\mathcal{B}_{S_2} = \Big\{ b\in \mathcal{B} \, : \, G^S_{-}( b_1^\ast, b_2^\ast, b_3^\ast ) \ge 0 \Big\},
\eas
where \( \mathcal{B} \) is given in \eqref{eq:B}. The region is represented in Fig. \ref{fig:S2}. It is enclosed between the lines $b_2^\ast = b_3^\ast$ and $G^S_{-}( b_1^\ast, b_2^\ast, b_3^\ast ) = 0$. The green line is defined through the identity $G^{S}_{+}( b_1^\ast, b_2^\ast, b_3^\ast ) = G^{S}_{-}( b_1^\ast, b_2^\ast, b_3^\ast )$ and corresponds to irrotational ellipsoids. On this line the Hamiltonian equation gets reduced to a system with three degrees of freedom. Above the line the two momenta of the ellipsoids are counter-parallel, whereas they are co-parallel below the curve.  \\

\begin{figure}[htb]
    \centering
    \includegraphics[width=0.6\textwidth]{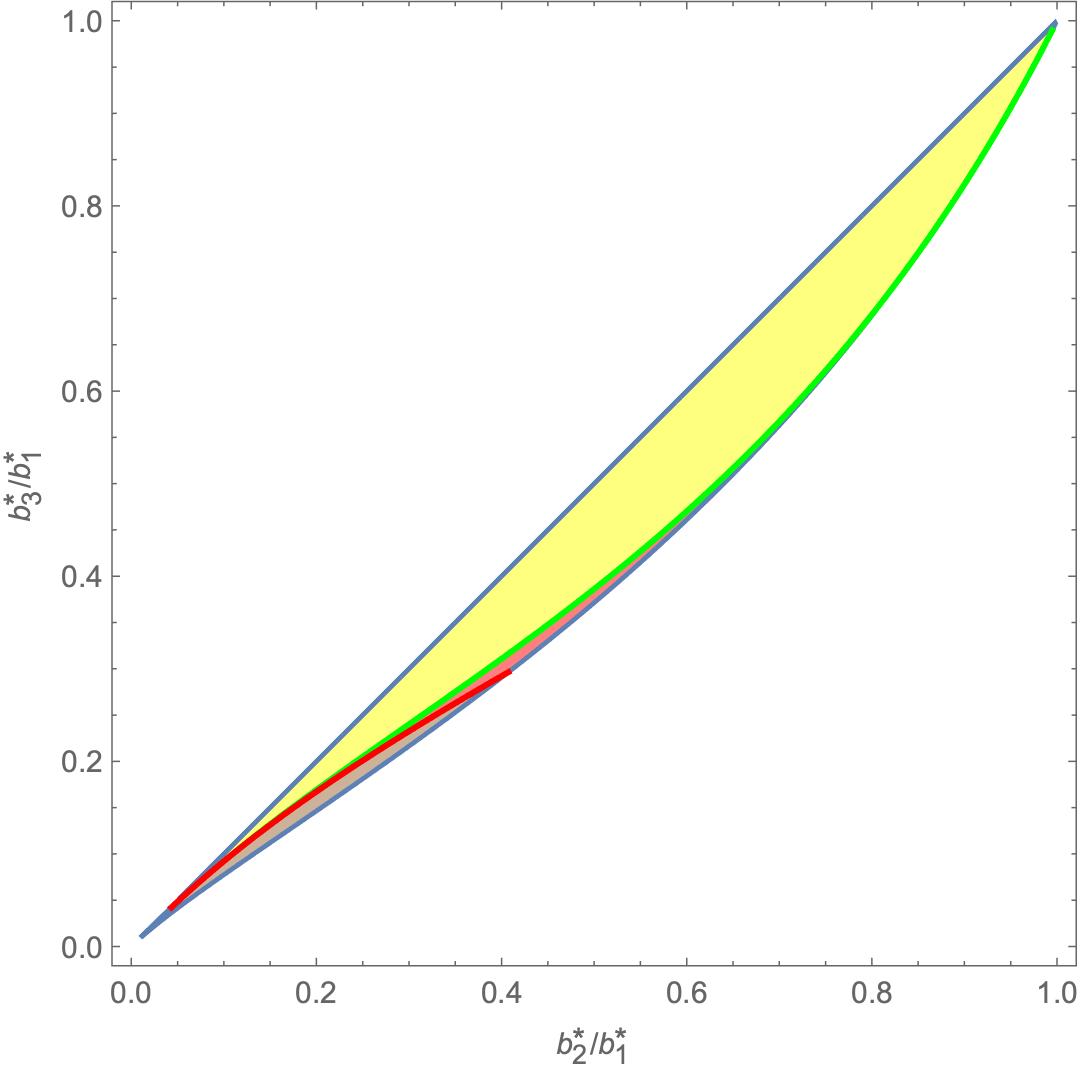}
    \caption{${\mathcal B}_{S_2}$: Region of existence of the $S_2$-Riemann ellipsoids in the parametric plane $b_2^\ast/b_1^\ast$--$b_3^\ast/b_1^\ast$. The green line ( $G^{S}_{+}( b_1^\ast, b_2^\ast, b_3^\ast ) = G^{S}_{-}( b_1^\ast, b_2^\ast, b_3^\ast )$) corresponds to the irrotational ellipsoids, whereas the red-one ($G = 0$) accounts for a quasi-periodic Hamiltonian pitchfork bifurcation of co-parallel $S_2$-ellipsoids. In the yellow subregion the ellipsoids are Liapunov stable, while in the pink one they are linearly stable and in the brown one they are unstable}
    \label{fig:S2}
\end{figure}

The region is divided into three sub-regions with different dynamics. Counter-parallel $S_2$-ellipsoids are Liapunov stable, as Riemann already stated \cite{Riemann}. Co-parallel $S_2$-ellipsoids result to be linearly stable with indefinite quadratic Hamiltonian function, thus their Liapunov stability is not known from the linear analysis and a non-linear investigation is due. 
\\
 
 We prove that co-parallel ellipsoids undergo a Hamiltonian pitchfork bifurcation of quasi-periodic nature. The red curve ($G = 0$) corresponds to a supercritical quasi-periodic pitchfork bifurcation of invariant $3$-tori. An elliptic $3$-torus above the curve becomes parabolic on $G = 0$ and then it turns hyperbolic when crossing the bifurcation line. Additionally, two elliptic $3$-tori are born when the first torus changes its stability. Furthermore, the appearance of the elliptic tori is associated to a global bifurcation involving type-III ellipsoids and it will be described in Section \ref{global}. The proof of the pitchfork bifurcation of invariant tori is partially based on KAM theory. We follow \cite{hanssmann2006local} (Section 4.1), but references \cite{LitvakHinenzonRomKedar2002, LitvakHinenzonRomKedar2002Nonlinearity} are also illustrative. 
 \\
 
The values of $\mu^{\pm}_{S_2}( b^{\ast} )$ introduced in the previous section are
\begin{equation}
    \label{Nes}
N^{S}_{\pm}( b_1^\ast, b_2^\ast, b_3^\ast ) e_2 = \left( 0, \mbox{$\frac{1}{2}$} \Big( \sqrt{G^{S}_{+}( b_1^\ast, b_2^\ast, b_3^\ast )} \pm \sqrt{G^{S}_{-}( b_1^\ast, b_2^\ast, b_3^\ast )} \Big), 0 \right),
\end{equation}
see Table \ref{table:nonlin}. We check that these expressions are well defined. Due to the fact that in the $\mathcal{B}_{S_2}$-region \(G^S_{-}( b_1^\ast, b_2^\ast, b_3^\ast ) \ge 0 \) holds, the only thing that should be checked is ${G^{S}_{+}( b_1^\ast, b_2^\ast, b_3^\ast )} \ge 0.$ From \eqref{eq:Gs} we obtain 
\[
\begin{array}{lcl}
G_{+}^{S}( b_1^\ast, b_2^\ast, b_3^\ast ) &=& 
\displaystyle 
\frac{( b_1^{\ast 2} b_2^\ast - 1 )^4}{b_1^{\ast 6} b_2^{\ast 5}} 
\Big(\big( b_2^{\ast 2} + b_1^{\ast 4} b_2^{\ast 4} + b_1^{\ast 2} ( b_2^{\ast 3} - 1 ) \big) C_1( b_1^\ast, b_2^\ast, b_3^\ast ) \\
&& \hspace*{2.1cm} + \, \big( b_1^{\ast 2} b_2^\ast ( b_2^{\ast 3} + 1 ) \big) C_2( b_1^\ast, b_2^\ast, b_3^\ast) \Big).
  \end{array}
\] 
Note that $C_{i}( b^\ast )$ are non-negative, as they are integrals of positive functions. 
Additionally, using \eqref{eq:B} it is readily deduced that the terms factorising $C_1( b^\ast )$ and $C_2( b^\ast )$ are also positive. Hence, $G_{+}^{S}( b_1^\ast, b_2^\ast, b_3^\ast )$ is non-negative in this region.
\\

The following theorem is the main result regarding the dynamics of $S_2$-ellipsoids. 
		
\begin{thm} 
\label{S2Theorem}
Region $\mathcal{B}_{S_2}$ is divided into three sub-regions with the following features:

\begin{enumerate}

\item [i.] The first subregion is bounded by the lines $b_2^\ast = b_3^\ast$ and $G^{S}_{+}( b_1^\ast, b_2^\ast, b_3^\ast ) = G^{S}_{-}( b_1^\ast, b_2^\ast, b_3^\ast )$ (green curve in Fig. \ref{fig:S2}) and corresponds with counter-parallel $S_2$-ellipsoids. The counter-parallel $S_2$-ellipsoids in the interior of this subregion and the irrotational ones
on the green curve are Liapunov stable.
\\
 
\item[ii.] The second subregion is delimited above by the line $G^{S}_{+}( b_1^\ast, b_2^\ast, b_3^\ast ) = G^{S}_{-}( b_1^\ast, b_2^\ast, b_3^\ast )$ and below by the curves \( G = 0 \) and $G^s_{-}( b_1^\ast, b_2^\ast, b_3^\ast ) = 0$: \\

\begin{itemize}
	    
\item The co-parallel $S_2$-ellipsoids are linearly stable inside this region and on this part of the line $G^s_{-}( b_1^\ast, b_2^\ast, b_3^\ast ) = 0$.

\item The curve \(G = 0\) (red line in Fig. \ref{fig:S2}) corresponds to a Hamiltonian pitchfork bifurcation of quasi-periodic nature. 
	 
\end{itemize}

\item[iii.] The third subregion is bounded from above by the curve $G = 0$ and by $G^s_{-}( b_1^\ast, b_2^\ast, b_3^\ast ) = 0$ from below and corresponds with co-parallel $S_2$-ellipsoids. Inside this region the ellipsoids are unstable.
\end{enumerate}
\end{thm}

\begin{proof}
We choose the ellipsoid with $( S^2_L \times S^2_R )$-coordinates $( \mu^{+}_{S_2}( b^{\ast} ), \mu^{-}_{S_2}( b^{\ast} ))$, the study being the same for its adjoint ellipsoid. Parameters $L$, $R$ satisfy 
\[ 
L = N^{S}_{+}( b^{\ast} ), \quad R = \vert N^{S}_{-}( b^{\ast} ) \vert, 
\]
where $N^S_{\pm}$ has been introduced in (\ref{Nes}). Counter-parallel $S_2$-ellipsoids are represented by the equilibrium with coordinates $( b_1^\ast,b_2^\ast, 0,0, 0,-L,0, 0,R,0 )$.  Thus, the projection onto $S^2_L \times S^2_R$ corresponds to the South-North poles of the two-spheres. The coordinates of co-parallel ellipsoids are $( b_1^\ast,b_2^\ast, 0,0, 0,L,0, 0,R,0 )$ and the projection onto $S^2_L \times S^2_R$ corresponds to the North-North poles of the two-spheres.
\\

We start by determining the linear stability. The first step is shifting the equilibrium to the origin. For that, the following transformation is applied
   \[
  ( \bar b, \bar c, q_1, q_2, p_1, p_2 ) \rightarrow ( b, c, \eta_l, \eta_r ),
   \]
where $\bar b = ( \bar b_1, \bar b_2 )$, $\bar c = ( \bar c_1, \bar c_2 )$ and
 \[
   \begin{array}{lcl}
   b_i &=& b_i^\ast + {\bar b_i}, \quad
   c_i = {\bar c_i}, \\[1ex]
   m_1 &=& \pm q_1 \sqrt{L - \frac{q_1^2 + p_1^2}{4}}, \quad
   m_2 = \mp L \pm\frac{q_1^2 + p_1^2}{2}, \quad
   m_3 = p_1 \sqrt{L - \frac{q_1^2 + p_1^2}{4}}, \\[1ex]
   m_4 &=& -q_2 \sqrt{R - \frac{q_2^2 + p_2^2}{4}}, \quad
   m_5 = R - \frac{q_2^2 + p_2^2}{2}, \quad
   m_6 = p_2 \sqrt{R - \frac{q_2^2 + p_2^2}{4}},
   \end{array}
   \]
with the upper sign applying for counter-parallel ellipsoids and the lower one for co-parallel ellipsoids, see also the similar approach followed in \cite{fasso2001stability}. 
\\
   
Notice that the local coordinates $( q_1, q_2, p_1, p_2 )$ are canonical, as they preserve the Poisson structure associated to $m_i$, $i = 1, \ldots, 6$; thus, the whole transformation is symplectic.
\\

Naming $u = ( \bar b, q_1, q_2, \bar c, p_1, p_2 )$, the set of Cartesian coordinates (also said rectangular), the next step is performing a Taylor expansion around $u = 0$ up to polynomials of degree two. We determine the quadratic form $H_2( u ) = \frac{1}{2} \, u^T \cdot ( -{\mathcal J}_8 {\mathcal L} ) \cdot u$, where ${\mathcal J}_8$ is the usual $( 8 \times 8 )$-skew symmetric matrix, whereas $H_2$ refers to the Hamiltonian function of the linearised system around $u = 0$ with linearisation matrix ${\mathcal L}$. The entries of this matrix are provided explicitly in Appendix \ref{CoefficientsL} for the co-parallel case. Notice that they are similar in the counter-parallel regime and that they have been placed in the {\sc Mathematica} file supplied with the paper.
\\

Next, we apply Markeev's algorithm described in Appendix \ref{Linearization} to bring $H_2$ to normal form. We arrive at the following conclusions:
\\
   
\begin{enumerate}

{\small
\item [i.] Counter-parallel ellipsoids are Liapunov stable because the Hamiltonian corresponding to the linearised system in the normal-form coordinates $z = ( x_1, x_2, x_3, y_1, y_2, y_3 )$ becomes
\[
 H_2( z ) = \frac{\omega_1}{2} ( x_1^2 + y_1^2 ) + \frac{\omega_2}{2} ( x_2^2 + y_2^2 ) + \frac{\omega_3}{2} ( x_3^2 + y_3^2 ) +  \frac{\omega_4}{2} ( x_4^2 + y_4^2 ),
 \]
where the frequencies $\omega_i$ appear in (\ref{omegas}) and have to be understood such that the associated coefficients $\ell_{i,j}$ are those specific for the counter-parallel regime of the $S_2$-ellipsoids. The entries $\ell_{i,j}$ are given in the {\sc Mathematica} file. In this subregion of $\mathcal{B}_{S_2}$ the transformation matrix ${\mathcal T}$ appearing in Appendix \ref{Linearization} is real and the $\omega_i$, $i = 1, \ldots, 4$, coefficients are positive. Thus, applying Dirichlet Stability Theorem \cite{meyeroffin}, Liapunov stability is achieved. 
\\

On the boundary of the subregion, i.e. on the curve $G^{S}_{+}( b_1^\ast, b_2^\ast, b_3^\ast ) = G^{S}_{-}( b_1^\ast, b_2^\ast, b_3^\ast )$ (green line in Fig. \ref{fig:S2}) the system has three degrees of freedom. We set $m_1 = m_2 = m_3 = 0$
and take $u = ( \bar b, q_2, \bar c, p_2 )$. The equilibrium has coordinates $( b_1^\ast,b_2^\ast, 0,0, 0,R,0 )$ and the symplectic matrix ${\mathcal T}_I$ is defined accordingly. We note that all matrices ${\mathcal L}_I$, ${\mathcal T}_I$ and ${\mathcal J}_6$ involved are $( 6 \times 6 )$-dimensional. The normal-form Hamiltonian truncated at degree two in this case is
\[
 H_2( z ) = \frac{\omega_1}{2} ( x_1^2 + y_1^2 ) + \frac{\omega_2}{2} ( x_2^2 + y_2^2 ) + \frac{\omega_3}{2} ( x_3^2 + y_3^2 ).
 \]
 The frequencies $\omega_i$ appearing in (\ref{omegas}) satisfy $\omega_i > 0$ for $i = 1, 2, 3$. Thereby, Liapunov stability also holds on the boundary of this subregion. 
 \\

\item [ii] (and (iii)) Now we focus on co-parallel ellipsoids. The linear normal form is determined by applying Markeev's approach as in item i. That being said, we wish to obtain a normal-form Hamiltonian that remains valid not only for the linearly stable part, but also for the unstable one. Thus, we need to make a slight modification in the calculation of matrix ${\mathcal T}$. Indeed, it is enough to do $x_4 \rightarrow \sqrt{\omega_4} x_4$, $y_4 \rightarrow y_4/\sqrt{\omega_4}$. By applying this change and doing some simplifications the resulting matrix, that we also name ${\mathcal T}$ is, is real and well defined everywhere above, below and on the red line. (The $\omega_4$ in the denominator of $y_4$ compensates with a factor in the denominator that also vanishes on the bifurcation line leading to a valid formula which makes sense even when $\omega_4 = 0$.) The final linear transformation remains symplectic, and is given in Appendix \ref{Linearization}. 
\\

Matrix ${\mathcal T}$ lies in the range of the so called versal normal form. The theory was developed by Arnold \cite{Arnold} to overcome the difficulty that for matrices that depend on parameters their transformations into Jordan canonical form could become singular. In our context, $\mathcal T$ depends smoothly on $\omega_4$ and it is real and non-singular in a neighbourhood, at least in a narrow strip surrounding the curve $G = 0$.  
\\

The transformed quadratic Hamiltonian function is
 \[
 H_2( z ) = \frac{\omega_1}{2} ( x_1^2 + y_1^2 ) + \frac{\omega_2}{2} ( x_2^2 + y_2^2 ) + \frac{\omega_3}{2} ( x_3^2 + y_3^2 ) - \frac{1}{2} ( \omega_4^2 x_4^2 + y_4^2 ),
 \]
 with $\omega_i$ given in (\ref{omegas}), where this time $\ell_{i,j}$ are provided in an explicit way in Appendix \ref{CoefficientsL}. We stress that $\omega_i > 0$ for $i = 1, 2, 3$, while $\omega_4$ can be positive, pure imaginary with negative imaginary part or zero. More specifically $\omega_4 > 0$ for co-parallel ellipsoids above the bifurcation line $G = 0$ (i.e., the red line in Fig. \ref{fig:S2}), whereas $\omega_4 = \imath \bar{\omega}_4$, $\bar{\omega_4} < 0$ for co-parallel ellipsoids below the red curve, and $\omega_4 = 0$ on the curve $G = 0$. Actually, $G = 0$ is equivalent to $\omega_4 = 0$ at the points of the parametric plane where the bifurcation takes place. Thereby, both above and below the line $G = 0$ the Hamiltonian is semisimple.
 \\ 
 
 The origin $z = 0$ is linearly stable above $G = 0$, as the linearisation is of the type centre $\times$ centre $\times$ centre $\times$ centre with three positive signs in front of the $\omega_i > 0$ and one negative. Below the red curve, since 
 \[ - \mbox{$\frac{1}{2}$} ( \omega_4^2 x_4^2 + y_4^2 ) = - \mbox{$\frac{1}{2}$}( -\bar{\omega}_4^2 x_4^2 + y_4^2 ), \quad \mbox{with} \,\, \bar{\omega}_4^2 > 0, 
 \]
 the equilibrium is unstable with linearisation centre $\times$ centre $\times$ centre $\times$ saddle. On the red curve the Hamiltonian is no longer semisimple, as it has the nilpotent term $-y_4^2/2$. 
 \\

In order to prove that a quasi-periodic Hamiltonian pitchfork bifurcation takes place we need to determine the non-linear terms up to degree four. For that, we extend the computation of the normal form up to quartic terms in the $z$ coordinates and express the normal form in complex/real-symplectic coordinates, say $Z = ( X_1, X_2, X_3, X_4, Y_1, Y_2, Y_3, Y_4 )$, such that
\begin{equation}
\label{complex}
\begin{array}{lcl}
    x_i &=& \mbox{$\frac{1}{\sqrt{2}}$}( X_i + \imath Y_i ), \quad
    y_i = \mbox{$\frac{1}{\sqrt{2}}$}( \imath X_i + Y_i ), \quad i = 1, \ldots, 3, \\[1ex]
    x_4 &=& X_4, \quad y_4 = Y_4.
\end{array}
\end{equation}
Then, 
\[ 
H_2( Z ) = \imath \omega_1 X_1 Y_1 + \imath \omega_2 X_2 Y_2 + \imath \omega_3 X_3 Y_3 -
\mbox{$\frac{1}{2}$}( \omega_4^2 X_4^2 + Y_4^2 ).
\\[1ex] \]

It is time to apply the linear changes passing from the $u$ coordinates to the $Z$ and execute two steps of the Lie transformation method \cite{Deprit}, proceeding in a symbolic fashion. The first order of the generating function, ${\mathcal W}_1$, is determined in such a way that the associated normal form, $H_1$, be zero. For that, we deal with the homological equation solving 120 linear equations with 120 unknowns (these unknowns are the coefficients of the terms of ${\mathcal W}_1$, i.e., monomials of degree three in $Z$). Cubic terms are neither present in the Hamiltonian funciton written in normal-form coordinates. This is due to the reversible character of the perturbation in case of type-$S$ ellipsoids.
\\

For computing the normal-form Hamiltonian, say $H_4$, and the associated generating function ${\mathcal W}_2$ we impose that the terms in the normal form are combinations of $X_1 Y_1$, 
$X_2 Y_2$, $X_3 Y_3$, $X_4^2$. As $H_4$ is of degree four in $Z$ (by an abuse of notation we also name $Z$ the transformed coordinates), we set
\[
\begin{array}{lcl}
H_4( Z ) &=& Q_1 ( X_1 Y_1 )^2 + Q_2 ( X_2 Y_2 )^2 + Q_3 ( X_3 Y_3 )^2 + Q_4 ( X_4^2 )^2 \\[1ex]
&& + \, Q_5 X_1 Y_1 X_2 Y_2 + Q_6 X_1 Y_1 X_3 Y_3 + Q_7 X_2 Y_2 X_3 Y_3 \\[1ex]
&& + \, Q_8 X_1 Y_1 X_4^2 + Q_9 X_2 Y_2 X_4^2 + Q_{10} X_3 Y_3 X_4^2.
\end{array}
\]
One has to expect a transformed Hamiltonian like $H_4$ due to the nilpotent term in $H_2$ when $\omega_4 = 0$, see \cite{meyeroffin}. The $Q_i$ coefficients $( i = 1, \ldots, 10 )$ and the ones forming the function ${\mathcal W}_2$ are determined by solving a consistent underdetermined linear system with 330 equations and 340 unknowns. Out of all unknowns, 330 correspond to the coefficients of ${\mathcal W}_2$ written in terms of the monomials of degree four in $Z$, and the other ten are the $Q_i$.\\

The transformation is well defined excepting certain resonance values. We determine the resonances by taking the denominators in the generating functions ${\mathcal W}_1$ and ${\mathcal W}_2$, evaluating them along the bifurcation line ($\omega_4 = 0$) and selecting the ones that pass through zero. There are two resonances of order $3$ and one of orders $2$ and $4$, specifically
\[ 
-\omega_1 + \omega_3, \, 
-\omega_1 + 2\, \omega_2, \, 
-2\, \omega_2 + \omega_3, \, 
-\omega_1 + 3\, \omega_2. 
\] 
They are shown in Fig. \ref{fig:resonances}. The order-two resonance is $-\omega_1 + \omega_3$. The order-three ones are $-\omega_1 + 2\, \omega_2$ and $-2\, \omega_2 + \omega_3$. The resonance of order $4$ is $-\omega_1 + 3\, \omega_2$. Consequently, in order to avoid the appearance of vanishing denominators in the expressions, we need to remove from the line $G = 0$ those points $( b_2^\ast/b_1^\ast, b_3^\ast/b_1^\ast )$ where the linear combinations of the frequencies become zero, since for these points of the parametric plane the approach is not valid. Furthermore, by continuity of the formulae with respect to the parameters $b^\ast$, we discard small neighbourhoods (balls centred at the points where the denominators are exactly zero), since some terms in the generating functions become unbounded there. More precisely, $-\omega_1 + \omega_3 = 0$ for $b_2^\ast/b_1^\ast \approx 0.3602$, $-\omega_1 + 2\, \omega_2 = 0$ for $b_2^\ast/b_1^\ast \approx 0.2716$, $-2\, \omega_2 + \omega_3 = 0$ for $b_2^\ast/b_1^\ast \approx 0.2802$ and  $-\omega_1 + 3\, \omega_2 = 0$ for $b_2^\ast/b_1^\ast \approx 0.1518$ and the values of $b_3^\ast/b_1^\ast$ are determined after solving the equation $G = 0$. 
\\

\begin{figure}
    \centering
  \includegraphics[width=0.6\textwidth]{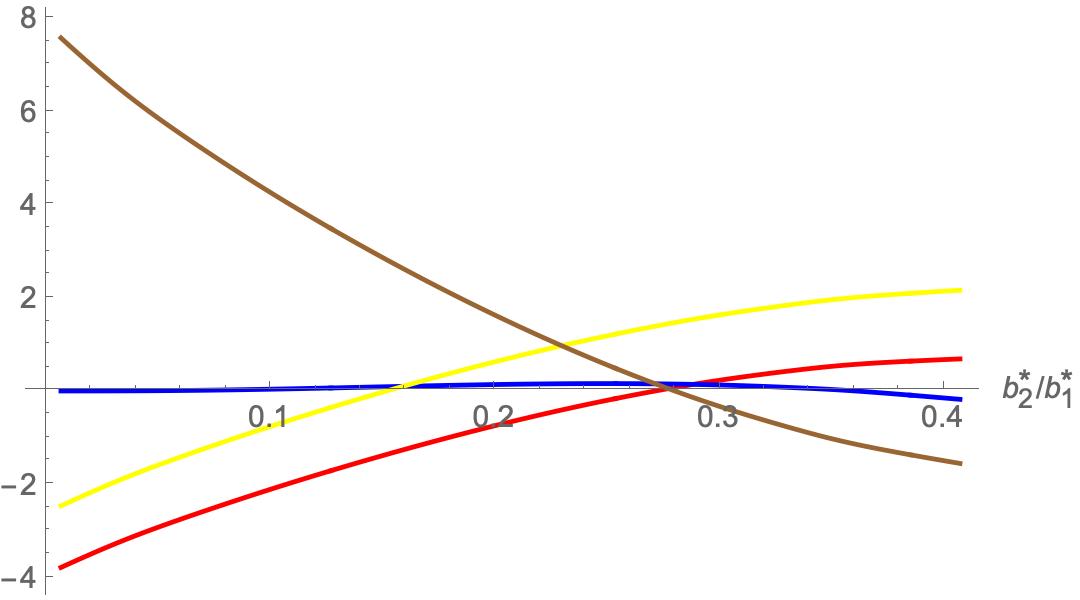}
    \caption{Resonances of different orders. Order 2: blue, $-\omega_1 + \omega_3$. Order 3: red, $-\omega_1 + 2\, \omega_2$; brown, $-2\omega_2 + \omega_3$. Order 4: yellow, $-\omega_1 + 3\, \omega_2$. Despite the blue curve looks very close to the axis $b_2^\ast/b_1^\ast$, it starts on the left taking the value $-\omega_1 + \omega_3 = 7.44... \cdot 10^{-6}$ for $b_2^\ast/b_1^\ast \approx 0$, then, it increases reaching its maximum at around $b_2^\ast/b_1^\ast = 0.24$ and decreases crossing the horizontal axis at a unique point around $b_2^\ast/b_1^\ast = 0.3602$}
    \label{fig:resonances}
\end{figure}

Next we check the specific conditions for establishing the occurrence of a supercritical quasi-periodic Hamiltonian pitchfork bifurcation on the curve $G = 0$, see Theorem 4.13 in \cite{hanssmann2006local}. 
\\

We return to a real normal form by introducing the actions $I_i = \imath X_i Y_i$, $i = 1, 2, 3$. The linearised system has as Hamiltonian function
\[
H_2( I, X_4, Y_4 ) = \omega_1 I_1 + \omega_2 I_2 + \omega_3 I_3 - \mbox{$\frac{1}{2}$}( \omega_4^2 X_4^2 + Y_4^2 ), \,\,\,\, 
I = ( I_1, I_2, I_3 ).
\]

At this point we consider the truncated normal form $H^4 = H_2 + \frac{1}{2} H_4$ in terms of $I$, $X_4$, $Y_4$. We write it as
\begin{equation}
\begin{array}{rcl}
H^4( I, X_4, Y_4 ) &=& \omega_1 I_1 + \omega_2 I_2 + \omega_3 I_3 - \mbox{$\frac{1}{2}$}( \omega_4^2 X_4^2 + Y_4^2 ) \\[1ex]  && - \,
\mbox{$\frac{1}{2}$} \left( Q_1 I_1^2 + Q_2 I_2^2 + Q_3 I_3^2 + Q_5 I_1 I_2 + Q_6 I_1 I_3 + Q_7 I_2 I_3 \right. \\[1ex] &&
 \left. \hspace*{0.68cm} + \, \imath \, ( Q_8 I_1 + Q_9 I_2 + Q_{10} I_3 ) X_4^2 - Q_4 X_4^4 \right), 
\end{array} 
\label{H4S2}
\end{equation}
which is a real function because $Q_8$, $Q_9$, $Q_{10}$ are pure imaginary while the other $Q_i$ are real. 
\\

Notice that the coefficient of $Y_4^2$ in $H^4$ is negative. Besides, we take a careful look at the coefficients of $X_4^2$ and $X_4^4$, respectively, 
\[ 
-\mbox{$\frac{1}{2}$} \left( \omega_4^2 + \imath ( Q_8 I_1 + Q_9 I_2 + Q_{10} I_3 ) \right), \quad \mbox{$\frac{1}{2}$} Q_4. 
\\
\] 
Firstly, the coefficient of $X_4^2$ is zero for $I = 0$, $\omega_4 = 0$ but it does not vanish when $\omega_4 \neq 0$ and $I = 0$, that is, in a neighbourhood of the line $G = 0$. Proving that $Q_4 \neq 0$ when $\omega_4 = 0$ requires more effort. We need to prove that both coefficients vanish only when $\omega_4 = 0$. We also check that the coefficient of $X_4^4$ is different from zero. Let us stress that $Q_4$ is computed in an explicit way on the whole line $G = 0$ and also in a neighbourhood of it, and it is given in terms of $b^\ast$ and supplied in the {\sc Mathematica} file. However, for the sake of proving that it does not vanish on the bifurcation line we proceed by replacing $Q_4$ in terms of the values $b^\ast$ take on the line $G = 0$. This step is numerical but we have performed it with very high precision of the calculations. We conclude that $Q_4 < 0$ in all points of the bifurcation line, see Fig. \ref{fig:Q4}. By continuity of the formulae with respect to the parameters and variables it is also negative on a narrow strip of the bifurcation line in the parametric plane. This bifurcation is of supercritical type as the coefficient of $Y_4^2$ is negative and the coefficient of $X_4^4$ remains negative as well.
\\

\begin{figure}
    \centering
  \includegraphics[width=0.6\textwidth]{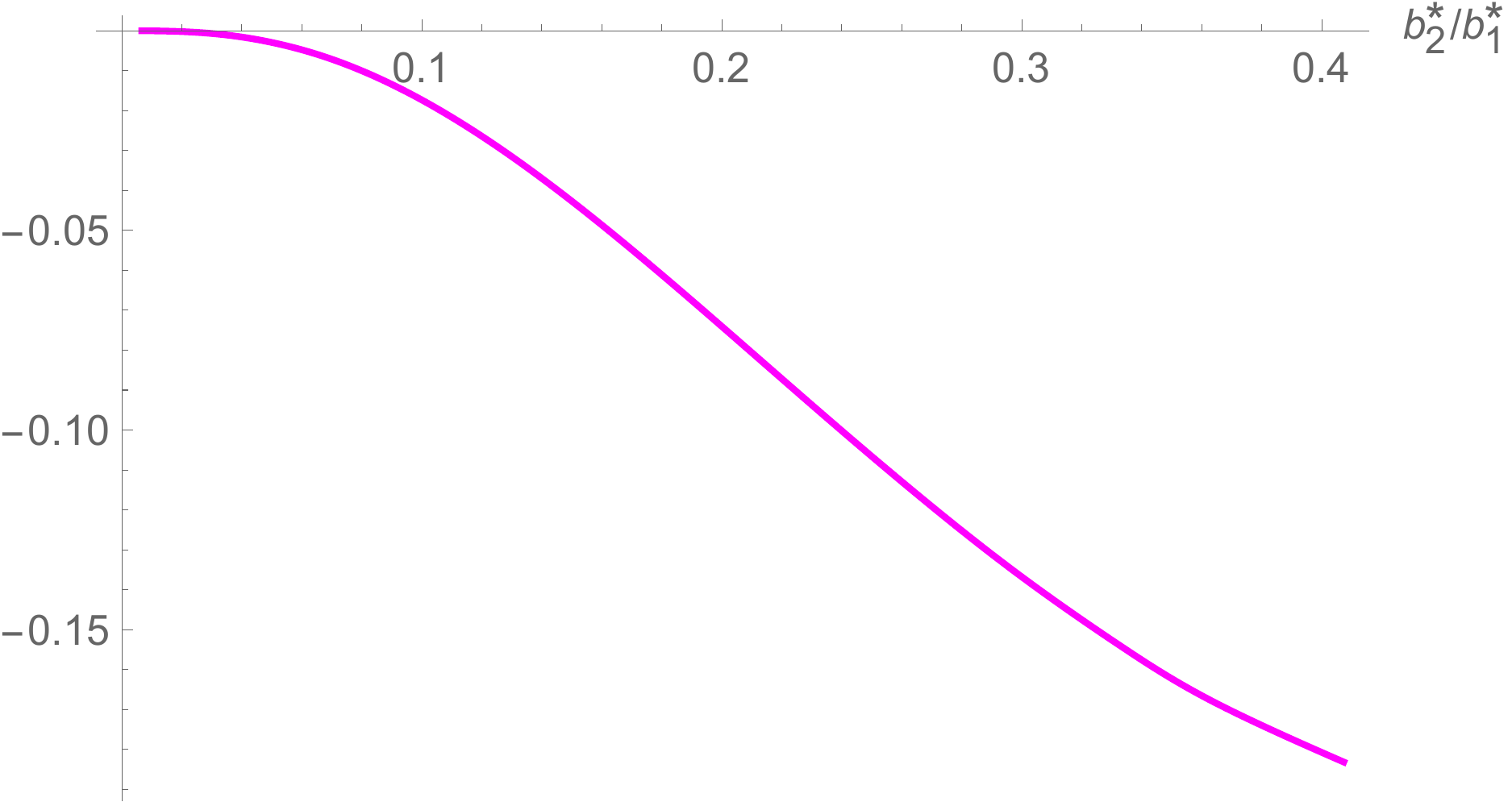}
    \caption{Coefficient $Q_4$ evaluated along the line $G = 0$. It is always below the horizontal axis. In fact, 
    it starts taking the value $Q_4 = -7.341... \cdot 10^{-6}$ for $b_2^\ast/b_1^\ast \approx 0$, and it is a decreasing function. Thus, $Q_4$ is negative for $b^\ast$ on the bifurcation line}
    \label{fig:Q4}
\end{figure}

At this point we need to prove the persistence under perturbation of the invariant tori related to the bifurcation. For this purpose we introduce 
\[ 
\tilde{c}( I ) = - \mbox{$\frac{1}{2}$} \omega_4^2 - \mbox{$\frac{\imath}{2}$} ( Q_8 I_1 + Q_9 I_2 + Q_{10} I_3 ), \quad
\tilde{\omega}_i( I ) = \frac{\partial H^4( I, 0, 0 )}{\partial I_i}, 
\]
and define the map $\xi: I \rightarrow \left( \tilde{c}( I ), \tilde{\omega}_1( I ), \tilde{\omega}_2( I ), \tilde{\omega}_3( I ) \right)$. Notice that $\tilde{c}( I )$ is taken as the coefficient of $X_4^2$ in $H^4$.
\\

We have to prove that $\xi$ a submersion at $\omega_4 = 0$, i.e. that the map is differentiable and the differential is surjective everywhere. Following \cite{hanssmann2006local}, on the one hand we get $D \tilde{c}( 0 ) = -\frac{\imath}{2} ( Q_8, Q_9, Q_{10} ) \neq ( 0, 0, 0 )$ along the curve $G = 0$, excepting at the resonance combinations which lead to very small or null denominators. More precisely, the evaluation of the norm of $D \tilde{c}( 0 )$ on a grid of points along the curve $G = 0$, remains positive and its minimum value is approximately $0.1617$. On the other hand, we build the ($3 \times 3$)-matrix $\mathtt M$ whose first row is $( \tilde{\omega}_1( I ), \tilde{\omega}_2( I ), \tilde{\omega}_3( I ))$, the second and third rows are the partial derivatives of the first row with respect to $I_1$ and $I_3$. The determinant of $\mathtt M$ at $I = 0$ yields
\[ 
\mbox{$\frac{1}{4}$} \left( 
\omega_1 ( 2\, Q_3 Q_5 - Q_6 Q_7 ) + 
\omega_2 ( Q_6^2 - 4\, Q_1 Q_3 ) + 
\omega_3 ( 2\, Q_1 Q_7 - Q_5 Q_6 ) \right). 
 \] 
This expression remains positive in a fine grid of points $b^\ast$ chosen homogeneously along the bifurcation curve $\omega_4 = 0$. We have to exclude the resonance values, where the normal-form computations do not make sense. Then, we conclude that the map $\xi$ is a submersion. The related calculations are provided in the {\sc Mathematica} file. 
\\

This gives the persistence of the invariant tori that interplay in the bifurcation. 
}
\end{enumerate}
\end{proof}

\begin{remark}
Determining the validity of the linear normal-form transformations, that is, whether ${\mathcal T}$, ${\mathcal T}_I$ are real matrices with non-vanishing denominators in the corresponding subregions of $\mathcal{B}_{S_2}$ where they are built, is not easy to accomplish. The same happens with the frequencies $\omega_i$. For instance, they are strictly positive for the ellipsoids of item i in the proof. Thus, one concludes Liapunov stability. In fact, we have all the associated expressions given explicitly in terms of $b^\ast$, but they are too big so that we can check our requisites. One can prove some partial results, for instance: $\omega_1 > \omega_2$, $\omega_3 > \vert \omega_4 \vert$. An alternative is checking the validity of our claim on a fine grid in ${\mathcal B}_{S_2}$, with the values $\ell_{i,j}$ take accordingly to the subregion we are considering. This can be performed with {\sc Mathematica} using the routine {\tt RegionPlot[]}, that makes plots of the provided formulae on specified regions in a two-dimensional grid. The approach is numerical but one can use high precision for the internal calculations. For instance, one can check that $\omega_i > 0$ on the green line and above it. The approach is similar for the behaviour of $\omega_4$ close to the bifurcation line. Proceeding like this we observe that the constructions we present are all right.
\label{remposiS2}    
\end{remark}

\begin{remark}
The invariant $3$-tori persisting in the co-parallel region (above and below the bifurcation line) are surrounded by families of invariant Lagrangian $4$-tori. This is established applying the standard Kolmogorov's non-degeneracy condition.
\\
\label{remposiS2b}    
\end{remark}

\begin{remark}
The principal terms of the invariant $3$ or $4$-tori that persist the small perturbations are trivially derived from Hamiltonian $H^4$ in (\ref{H4S2}) in the normal-form coordinates $I$, $X_4$,$Y_4$. It is possible to obtain them in the original coordinates by undoing the normal-form transformations.
\\
\label{remposiS2c}    
\end{remark}
   
\section{Stability of $S_{3}$-ellipsoids}
\label{SectionS3}

We deal with the $S_3$-ellipsoids. Riemann already proved that they are Liapunov stable \cite{Riemann}. Our contribution stems from the fact that the calculations are symbolically in the entire region of existence. This time $\mu^{\pm}_{S_3}( b^{\ast} )$ of Table \ref{table:nonlin} is
\begin{equation}
    \label{Ness}
N^{S}_{\pm}( b_1^\ast, b_3^\ast, b_2^\ast ) e_3 = \left( 0, 0, \mbox{$\frac{1}{2}$} \Big( \sqrt{G^{S}_{+}( b_1^\ast, b_3^\ast, b_2^\ast )} \pm \sqrt{G^{S}_{-}( b_1^\ast, b_3^\ast, b_2^\ast )} \Big) \right).
\end{equation}

The region of the parametric plane where these Riemann ellipsoids exist is given by 
\[ 
{\mathcal B}_{S_3} = \Big\{ b \in {\mathcal B} \, : \, G_+^S( b_1^\ast, b_3^\ast, b_2^\ast ) \geq 0 \Big\} 
\]
and is represented in Fig. \ref{fig:S3}. Condition $G_{-}^{S}( b_1^\ast, b_3^\ast, b_2^\ast ) \geq 0$ is satisfied because
\[
\begin{array}{lcl}
G_{-}^{S}( b_1^\ast, b_3^\ast, b_2^\ast ) &=& \displaystyle
\frac{( b_1^\ast - b_2^\ast )^4}{b_1^{\ast 3} b_2^{\ast 3}}
\Big( \big( b_1^{\ast 4} b_2^{\ast 4} + b_1^\ast b_2^\ast - b_1^{\ast 2} - b_2^{\ast 2} \big) C_1( b_1^\ast, b_3^\ast, b_2^\ast ) \\ && \displaystyle \hspace*{1.8cm} + \, \big( b_1^{\ast 3} b_2^{\ast 3} - 1 \big) C_2( b_1^\ast, b_3^\ast, b_2^\ast ) \Big)
\end{array}
\] 
and $C_1( b_1^\ast, b_3^\ast, b_2^\ast )$, $C_2( b_1^\ast, b_3^\ast, b_2^\ast )$, together with their coefficients, are also positive by condition \eqref{eq:B}. In so doing, the $S_3$-ellipsoids are properly defined in ${\mathcal B}_{S_3}$.
\\

\begin{figure}[ht]
    \centering
    \includegraphics[width=0.6\textwidth]{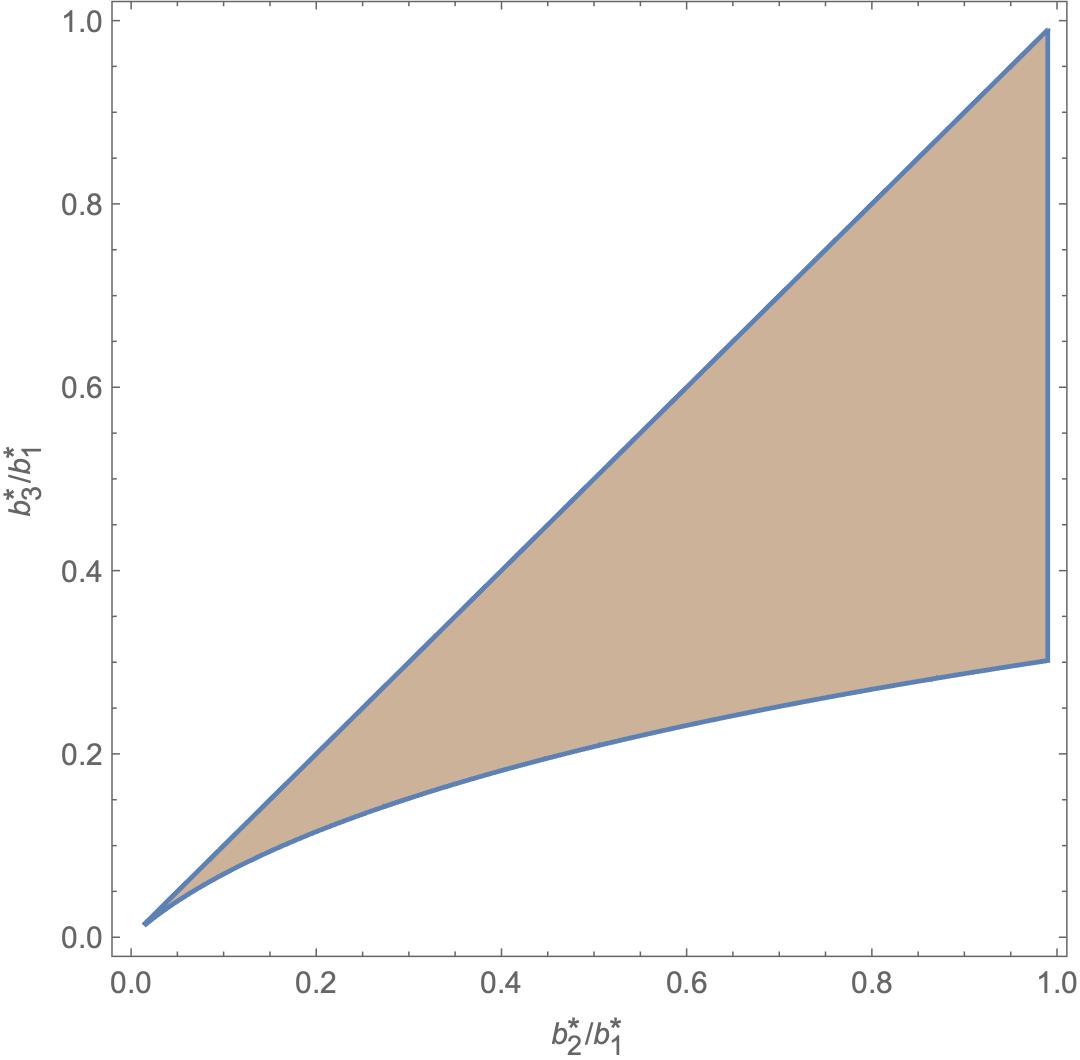}
    \caption{${\mathcal B}_{S_3}$: Region of existence of the $S_3$-Riemann ellipsoids}
    \label{fig:S3}
\end{figure}

We state the main result of this section.

\begin{thm} 
$S_3$-ellipsoids are Liapunov stable equilibria in their entire region of existence. 
\end{thm}
\begin{proof}
We take the cue from the scheme pursued in the proof of Theorem \ref{S2Theorem}. The coordinates of the $S_3$-ellipsoids are $( b_1^\ast, b_2^\ast, 0,0, 0,0,-L, 0,0,R )$, with $L$ and $R$ as in the case of $S_2$-ellipsoids. Projecting them onto $S^2_L \times S^2_R$ they correspond to the South-North poles of the two-spheres. Hence, $S_3$-ellipsoids are counter-parallel and the following change to symplectic variables is applied
\[
  ( \bar b, \bar c, q_1, q_2, p_1, p_2 ) \rightarrow
  ( b, c, \eta_l, \eta_r ),
\]
where $\bar b = ( \bar b_1, \bar b_2 )$, $\bar c = ( \bar c_1, \bar c_2)$ and
 \[
   \begin{array}{lcl}
   b_i &=& b_i^\ast + {\bar b_i}, \quad
   c_i = {\bar c_i}, \\[1ex]
   m_1 &=& p_1 \sqrt{L - \frac{q_1^2 + p_1^2}{4}}, \quad
   m_2 = q_1 \sqrt{L - \frac{q_1^2 + p_1^2}{4}}, \quad
   m_3 = -L + \frac{q_1^2 + p_1^2}{2}, \\[1.3ex]
   m_4 &=& p_2 \sqrt{R - \frac{q_2^2 + p_2^2}{4}}, \quad
   m_5 = -q_2 \sqrt{R - \frac{q_2^2 + p_2^2}{4}}, \quad
   m_6 = R - \frac{q_2^2 + p_2^2}{2}.
   \end{array}
   \]
In this manner we translate the equilibrium to the origin. Then, a Taylor expansion up to degree 2 is computed in the coordinates $u$. The coefficients of the linearisation matrix ${\mathcal L}$ are similar to those given in Appendix \ref{CoefficientsL} for the $S_2$-ellipsoids in the co-parallel regime, and they are provided in the {\sc Mathematica} file accompanied to this text. Next, we apply Markeev's procedure described in Appendix \ref{Linearization} to bring the quadratic form $H_2( u )$ to normal form. The linearisation in Cartesian coordinates results to be
\[
H_2( z ) = \frac{\omega_1}{2}( x_1^2 + y_1^2 ) + \frac{\omega_2}{2} ( x_2^2 + y_2^2 ) + \frac{\omega_3}{2} ( x_3^2 + y_3^2 ) + \frac{\omega_4}{2} ( x_4^2 + y_4^2 ),
 \]
with $\omega_i$ appearing in (\ref{omegas}) and the corresponding $\ell_{i,j}$ are given in the {\sc Mathematica} file. One has that $\omega_i > 0$ for $i = 1, \ldots, 4$. Consequently, $S_3$-ellipsoids are Liapunov stable.
\end{proof}

\begin{remark}
Analogous considerations to those made in Remark \ref{remposiS2} apply for the $S_3$-ellipsoids. They also apply in the bifurcation study of types-II and III ellipsoids.
\\
\label{remposiS3}    
\end{remark}

\section{Linear stability of type-I irrotational ellipsoids}
\label{TypeI}

This section is devoted to the analysis of the stability of type-I ellipsoids. Their existence domain is  
\bas
\mathcal{B}_{{\rm I}} = \Big\{ b\in \mathcal{B} \, : \, b_1^\ast \leq 2 b_2^\ast - b_3^\ast \Big\},
\eas
and is represented in Fig. \ref{fig:I}.\\
\begin{figure}[ht]
    \centering
    \includegraphics[width=0.6\textwidth]{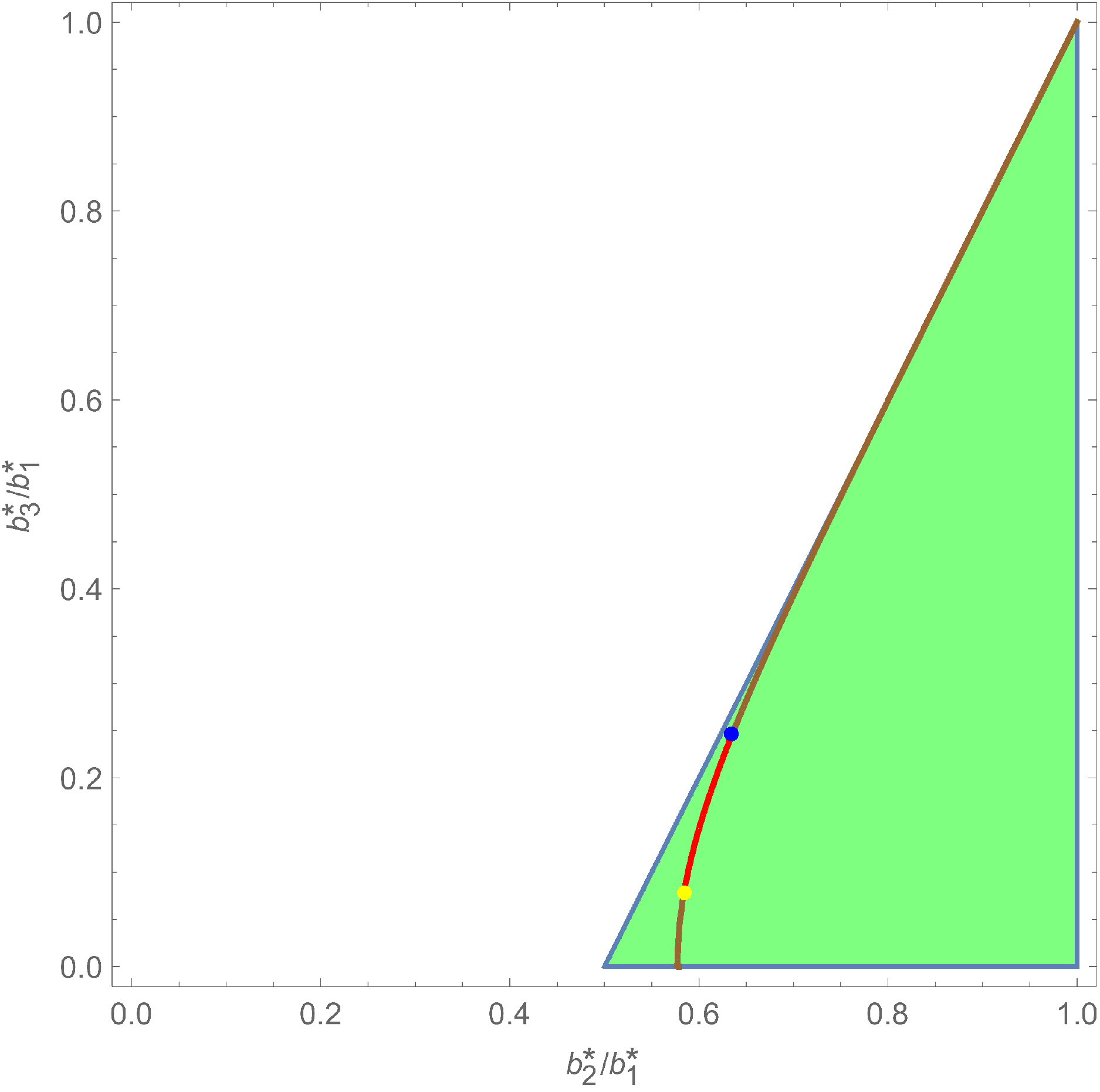}
    \caption{${\mathcal B}_{{\rm{I}}}$: Region of existence of type-I Riemann ellipsoids. The curve corresponds to irrotational ellipsoids. In the red portion there is linear stability, whereas there is instability in the brown ones. The brown part continues to the point $( 1, 1 )$. The yellow and blue points are likely to be Hamiltonian-Hopf bifurcations}
    \label{fig:I}
\end{figure}

Type-I ellipsoids can be linearly stable or unstable, passing from stable to unstable through bifurcation lines whose distribution may be very subtle \cite{fasso2001stability}. We have not detected Liapunov stability in this region just from a linear analysis. It could arise after studying the non-linear terms, but this is far from obvious. The bifurcation lines likely correspond to quasi-periodic Hamiltonian-Hopf bifurcations \cite{broer2007quasi, meyeroffin}. We do not present the analysis of these bifurcations in this section, as it is similar to the one we shall detail for type-III ellipsoids. 
\\

The procedure to study the linear stability of these ellipsoids follows the ideas of previous sections. We start by choosing the ellipsoid with $( S^2_L \times S^2_R )$-coordinates 
$( \mu^{+}_{N}( b_1^\ast, b_3^\ast, b_2^\ast ), \mu^{-}_{N}( b_3^\ast, b_1^\ast, b_2^\ast ))$, as the analysis for its adjoint is equivalent. Parameters $L$, $R$ satisfy 
\[ 
\begin{array}{rcl}
L &=& \sqrt{
N^{R}_{+}( b_1^\ast, b_3^\ast, b_2^\ast )^2 +
N^{R}_{+}( b_3^\ast, b_1^\ast, b_2^\ast )^2}, \\[1ex]
R &=& \sqrt{
N^{R}_{-}( b_1^\ast, b_3^\ast, b_2^\ast )^2 +
N^{R}_{-}( b_3^\ast, b_1^\ast, b_2^\ast )^2}, 
\end{array}
\]
where $N^R_{\pm}( x, y, z )$ are given in (\ref{eq:Gs}).

\begin{remark}
We have checked that the coordinates of the type-I ellipsoids make sense in the domain ${\mathcal B}_{\rm I}$, that is, that according to (\ref{eq:Gs}) and Table \ref{table:nonlin}, the terms inside the square roots of $N^R_{\pm}$ are non-negative for all $b^\ast$ in ${\mathcal B}_{\rm I}$. For doing it we take into account that $C_i \ge 0$ as well as the other restrictions delimiting the set ${\mathcal B}_{\rm I}$. The computations involved discussing that some rational functions cannot be negative imposing additional restrictions. They appear in the {\sc Mathematica} file. Similarly, we have proved that the coordinates 
of the ellipsoids of types II and III are right in ${\mathcal B}_{\rm II}$ and ${\mathcal B}_{\rm III}$, respectively.
\end{remark}

The symplectic coordinates introduced for this case are
\begin{equation}
\begin{array}{lcl}
   b_i &=& b_i^\ast + {\bar b_i}, \quad
   c_i = {\bar c_i}, \\[0.9ex]
   m_1 &=& -q_1 \sqrt{L - \frac{q_1^2 + p_1^2}{4}}, \quad
   m_2 = L -\frac{q_1^2 + p_1^2}{2}, \quad
   m_3 = p_1 \sqrt{L - \frac{q_1^2 + p_1^2}{4}}, \\[1.5ex]
   m_4 &=& -q_2 \sqrt{R - \frac{q_2^2 + p_2^2}{4}}, \quad
   m_5 = R - \frac{q_2^2 + p_2^2}{2},\quad
   m_6 = p_2 \sqrt{R - \frac{q_2^2 + p_2^2}{4}}.
   \end{array}
   \label{chIrr}
 \end{equation}
Unlike the previous cases, after the expansion of the Hamiltonian as a Taylor series up to degree two around the origin, we have not performed all the computations in a symbolic way. More precisely, we have not carried out the detailed numerical sweep of \cite{fasso2001stability}, in the sense that we have not taken a very fine grid and have done the linear analysis in that grid, but simply we have picked different points in the regions encountered by Fass\`o and Lewis \cite{fasso2001stability}. Our results agree with the ones obtained by them, regardless of the fact that short of spectral stability we go a bit further and obtain linear stability. The reader can see Figs. 3(a) and 4(a) in \cite{fasso2001stability} for the stability regions. 
\\

The linear analysis has been performed by following Appendix \ref{Linearization}, in case that the equilibria were elliptic points. On the boundary of the region $( b_1^\ast = 2\, b_2^\ast - b_3^\ast )$ we find either linearly stable points with linearisation 
\[
 -\frac{\omega_1}{2}( x_1^2 + y_1^ 2) - \frac{\omega_2}{2}( x_2^2 + y_2^2 ) + \frac{\omega_3}{2}( x_3^2 + y_3^2 ) + \frac{\omega_4}{2}( x_4^2 + y_4^2 ),
 \]
with $\omega_i > 0$ for $i = 1, \ldots, 4$ (upper part of the line) or unstable of the type centre $\times$ centre $\times$ focus (lower part of the line).
\\

In the interior of region ${\mathcal B}_{\rm I}$ there are stable points with linearisation
\begin{equation}
\label{linIa}
 -\frac{\omega_1}{2}( x_1^2 + y_1^2 ) + \frac{\omega_2}{2}( x_2^2 + y_2^2 ) + \frac{\omega_3}{2}( x_3^2 + y_3^2 ) + \frac{\omega_4}{2}( x_4^2 + y_4^2 ),
\end{equation}
 or
\begin{equation}
\label{linIb}
 \frac{\omega_1}{2}( x_1^2 + y_1^2 ) - \frac{\omega_2}{2}( x_2^2 + y_2^2 ) - \frac{\omega_3}{2}( x_3^2 + y_3^2 ) + \frac{\omega_4}{2}( x_4^2 + y_4^2 ),
 \end{equation}
with $\omega_i > 0$ for $i = 1, \ldots, 4$ or unstable of the type centre $\times$ centre $\times$
focus. 
\\

Type-I irrotational ellipsoids satisfy
\begin{equation}
\label{irrocurve}
b_1^{\ast 2} + b_2^{\ast 2} + b_1^{\ast 4} b_2^{\ast 4} - 3\, b_1^{\ast 2} b_2^{\ast 6} = 0
\end{equation}
and are represented by the curve shown in Fig. \ref{fig:I}. The irrotational curve crosses the subregions of stability and instability. The effect of passing through the irrotational curve, both in the linearly stable and in the unstable regimes of region ${\mathcal B}_{\rm I}$, is that on the left-hand side of the irrotational curve the minus sign in front of one of the $\omega_i$ becomes a positive sign when passing to the right-hand side of the curve and the term is zero on the line. This happens regardless of the nature of the point (either with linearisation centre $\times$ centre $\times$ centre $\times$ centre or centre $\times$ centre $\times$ focus). This effect was already observed for the $S_2$-ellipsoids. Here there is no Liapunov stability, at  from the linear analysis, though. 
\\

The transition from the linearly-stable parts of the irrotational curve to the unstable ones is likely to be made through two Hamiltonian-Hopf bifurcations. These bifurcations are no longer curves in the parametric plane but isolated points on the curve (\ref{irrocurve}). To obtain the values of these points we impose the condition on the eigenvalues to be in the right resonance relation, that is, the $1$:$-1$, with non-null nilpotent part. One passes from linear stability with quadratic Hamiltonian in normal form given by (\ref{linIa}) or by (\ref{linIb}) to instability, where the unstable character of the points is manifested by the appearance of a complex quadruplet of eigenvalues while one of the imaginary pairs remain imaginary.  
\\

Now we focus our study on the irrotational regime. The approach is analytical. The Hamiltonian system has three degrees of freedom and our aim is to deal with the changes between stability and instability behaviour. We detail how to obtain the normal-form Hamiltonian in case of linear stability. The analysis in the unstable case is similar but with the transformation to normal form dealing with the focus character of the unstable degrees of freedom, i.e. the ones corresponding to the quadruplet related to the eigenvalues $\pm a \pm \imath  b$ (with $a, b > 0$). This requires a different approach (see for instance \cite{LaubMeyer1974}) that we do not handle here.
\\

Let $A \approx ( 0.58419, 0.07787 )$, $B \approx ( 0.63527, 0.24613 )$ be the points on the parametric plane $b_2^\ast/b_1^\ast$--$b_3^\ast/b_1^\ast$ corresponding, respectively, with the yellow and blue points in Fig. \ref{fig:I}. The main result in this section is the following.

\begin{thm}
\label{ThmtypeI}
     Type-I irrotational ellipsoids are linearly stable between the points $A$ and $B$ of the parametric plane. 
     They are unstable elsewhere.
\end{thm}
\begin{proof}
As already said on the irrotational curve we work with three degrees of freedom. One of the two spheres 
is reduced to a point. We introduce a rotation matrix 
\[ {\mathcal R}( \gamma ) = \left(
\begin{array}{ccc}
   \cos \gamma  & 0 & \sin \gamma \\
   0 & 1 & 0 \\
   \sin \gamma & 0 & \cos \gamma
\end{array} \right) \]
with the aim of using the angle $\gamma$ to get a simpler expression of the quadratic Hamiltonian
and related linearisation matrix. After setting $q_1 = p_1 = 0$ the transformation (\ref{chIrr}) 
results in
\[
\begin{array}{l}
m_1 = m_2 = m_3 = 0,\\[1ex]
m_4 = -q_2 \sqrt{R - \frac{q_2^2 + p_2^2}{4}}, \quad
m_5 = R - \frac{q_2^2 + p_2^2}{2}, \quad
m_6 = p_2 \sqrt{R - \frac{q_2^2 + p_2^2}{4}}.
\end{array}
\]
Applying the rotation matrix ${\mathcal R}( \gamma )$ to $\eta_r$, and using the same name for the $m_i$, we get the symplectic change
\[
\begin{array}{ccl}
m_4 &=& -q_2 \sqrt{R - \frac{q_2^2 + p_2^2}{4}} \cos \gamma + p_2 \sqrt{R - \frac{q_2^2 + p_2^2}{4}} \sin \gamma, \\[1ex]
m_5 &=& R - \frac{q_2^2 + p_2^2}{2}, \\[1ex]
m_6 &=& p_2 \sqrt{R - \frac{q_2^2 + p_2^2}{4}} \cos \gamma + q_2 \sqrt{R - \frac{q_2^2 + p_2^2}{4}} \sin \gamma. 
\end{array}
\]

Now we evaluate the coordinates $( q_2, p_2 )$ at the equilibrium, taking into account that $R = ( q_2^{\ast 2} + p_2^{\ast 2} )/2$. Then, the following relations are obtained:
\[
 q_2^\ast = \frac{\sqrt{2}( -m_4^\ast \cos \gamma + m_6^\ast \sin \gamma )}{( m_4^{\ast 2} + m_6^{\ast 2} )^{1/4}}, \quad
 p_2^\ast = \frac{\sqrt{2}( m_6^\ast \cos \gamma + m_4^\ast \sin \gamma )}{( m_4^{\ast 2} + 
 m_6^{\ast 2} )^{1/4}},
\]
where $m_i^\ast$ are given at the equilibrium. These expressions in terms of $b^\ast$ appear explicitly in the {\sc Mathematica} file.
\\

The final transformation is $( q_2, p_2 ) \rightarrow ( \bar q_2, \bar p_2 )$, where $q_2 = q_2^\ast + \bar q_2$ and $p = p_2^\ast + \bar p_2$. We apply it to Hamiltonian (\ref{Hamiltonian}) and expand up to terms of degree 2. Then, following \cite{fasso2001stability} we select $\gamma$ such that the term containing $\bar b_1 \bar p_2$ vanishes. 
\\

We arrive at
\[ \gamma = \arccos \left( \frac{m^\ast_6}{\sqrt{m_4^{\ast 2} + m_6^{\ast 2}}} \right). 
\]

The next step follows the footsteps of the study of $S_2$ and $S_3$. We apply Markeev's method described in Appendix \ref{Linearization} to bring the quadratic form $-\frac{1}{2}\, u^T \cdot {\mathcal J}_6 {\mathcal L}_I \cdot u$, with $u = ( \bar b_1, \bar b_2, \bar q_2, \bar c_1, \bar c_2, \bar p_2 )$, to normal form. The linearisation matrix is
 \bas
   {\mathcal L}_I = \left( 
   \begin {array}{cccccc} 
   0 & 0 & 0 & \ell_{1,4} & \ell_{1,5} & 0
	\\ \noalign{\medskip}
   0 & 0 & 0 & \ell_{1,5} & \ell_{2,5} & 0
    \\ \noalign{\medskip}
   0 & 0 & 0 & 0 & 0 & \ell_{3,6}
    \\ \noalign{\medskip}
   \ell_{4,1} & \ell_{4,2} & \ell_{4,3} & 0 & 0 & 0
    \\ \noalign{\medskip}
   \ell_{4,2} & \ell_{5,2} & \ell_{5,3} & 0 & 0 & 0
	\\ \noalign{\medskip}
   \ell_{4,3} & \ell_{5,3} & 0 & 0 & 0 & 0
   \end{array}
	\right),
	\eas 
where the coefficients $\ell_{i,j}$ are functions of $b^\ast$. Besides, the frequencies $\omega_i$ are expressed in terms of $\ell_{i,j}$, similarly to what we showed for the type-$S$ ellipsoids, much as through more involved expressions. The $\ell_{i,j}$ and $\omega_i$ are provided in the {\sc Mathematica} file. 
\\

Moving along the curve (\ref{irrocurve}) by means of different numerical samples, we observe basically two behaviours. Either the eigenvalues of the linearisation matrix are pure imaginary and the eigenvectors span ${\mathbb R}^6$ or there is a pair of pure imaginary eigenvalues and a quadruplet $\pm a \pm \imath b$, with $a, b > 0$. Imposing the frequencies to be in $1$:$-1$ non-semisimple resonance, we obtain two points in the parametric plane with coordinates $A \approx ( 0.58419, 0.07787 )$, $B \approx ( 0.63527, 0.24613 )$ such that for $b^\ast$ in (\ref{irrocurve}) without including $A$ and $B$ (red part of the curve in Fig. \ref{fig:I}), the quadratic normal-form Hamiltonian in the rectangular coordinates $z$ is
\[
H_2( z ) = -\frac{\omega_1}{2}( x_1^2 + y_1^2 ) + \frac{\omega_2}{2}( x_2^2 + y_2^2 ) + \frac{\omega_3}{2}( x_3^2 + y_3^2 ),
\]
with $\omega_i > 0$. Thus, on the stable part of the irrotational curve we obtain linear stability. On the curve (\ref{irrocurve}) outside the segment $( A, B )$ (brown parts of the curve plotted Fig. \ref{fig:I}) we get instability through the behaviour explained in the paragraphs previous to this theorem. Thus, $A$, $B$ are likely to correspond to two points where Hamiltonian-Hopf bifurcations take place. In fact, thinking of the Hamiltonian system with four degrees of freedom, points $A$ and $B$ belong to two curves in the parametric plane where Hamiltonian-Hopf bifurcations occur. They correspond to the two main curves in Fig. 3(a) of \cite{fasso2001stability}, accounting for the transition between spectral stability and instability.
\end{proof}

The rest of lines in the parametric plane exhibiting changes in stability have not been tackled in detail, but our numerical approach suggests that they are related to Hamiltonian-Hopf bifurcations. See also Figs. 3(a), 4(a) in \cite{fasso2001stability}.

\section{Quasi-periodic saddle-centre bifurcation of type-II ellipsoids}
\label{TypeII}

This section is devoted to the analysis of the stability and bifurcations of type-II ellipsoids. As it occurred with type-I ellipsoids, we give a numerical description of the different regimes appearing in $\mathcal{B}_{\rm II}$ and study one of the two types of bifurcations analytically. The other one is left for the section related to type-III ellipsoids. Recall that
\[
{\mathcal B}_{\rm II} = \Big\{ b \in {\mathcal B} \, : \, b_1^\ast \geq 2\, b_2^\ast + b_3^\ast, D( b_1^\ast, b_3^\ast, b_2^\ast ) < 0 \Big\}, 
\]
with $D$ given in (\ref{eq:Gs}). This region is represented in Fig. \ref{fig:II}.
\\

\begin{figure}[htb]
    \centering
    \includegraphics[width=0.6\textwidth]{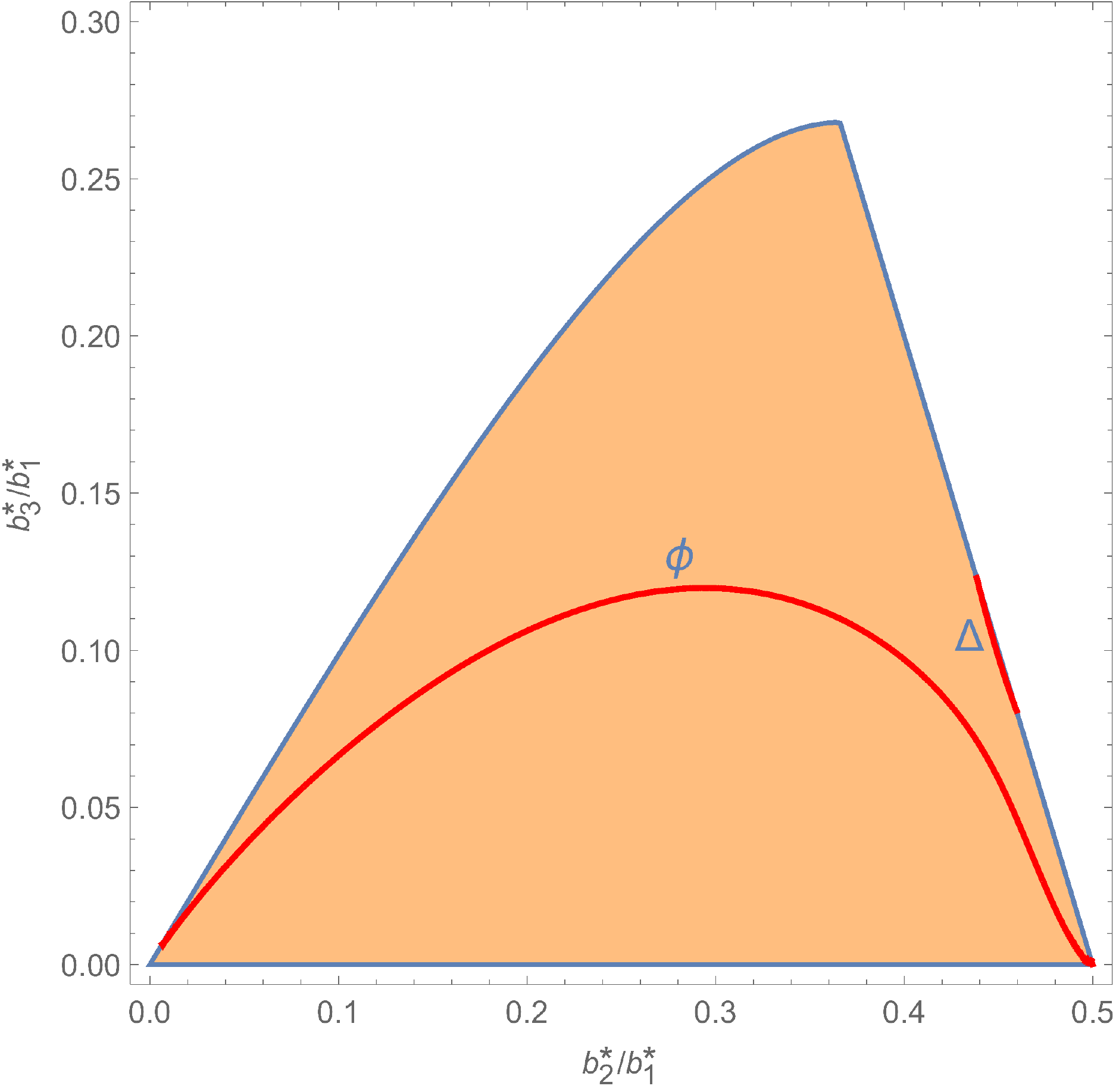}
    \caption{${\mathcal B}_{{\rm{II}}}$: Region of existence of the type-II Riemann ellipsoids. The red lines $\Phi$ and $\Delta$ correspond to saddle-centre bifurcations}
    \label{fig:II}
\end{figure}

We select the ellipsoid with $( S^2_L \times S^2_R )$-coordinates $( \mu^{+}_{N}( b_1^\ast, b_3^\ast, b_2^\ast ),$ $\mu^{-}_{N}( b_3^\ast, b_1^\ast, b_2^\ast ))$. The analysis for its adjoint essentially follows suit. Now $L$, $R$ are related to $b^\ast$ by 
\[ 
\begin{array}{rcl}
L &=& \sqrt{
N^{R}_{+}( b_1^\ast, b_3^\ast, b_2^\ast )^2 +
N^{R}_{-}( b_3^\ast, b_1^\ast, b_2^\ast )^2}, \\[1.3ex]
R &=& \sqrt{
N^{R}_{-}( b_1^\ast, b_3^\ast, b_2^\ast )^2 +
N^{R}_{+}( b_3^\ast, b_1^\ast, b_2^\ast )^2}. \\
\end{array}
\\
\]

The symplectic transformation suited for this case is
\begin{equation}
\label{changeTypeII}
   \begin{array}{lcl}
   b_i &=& b_i^\ast + {\bar b_i}, \quad
   c_i = {\bar c_i}, \\[1ex]
   m_1 &=& -q_1 \sqrt{L - \frac{q_1^2 + p_1^2}{4}}, \quad
   m_2 = L - \frac{q_1^2 + p_1^2}{2},\quad
   m_3 = p_1 \sqrt{L - \frac{q_1^2 + p_1^2}{4}}, \\[1ex]
   m_4 &=& -q_2 \sqrt{R - \frac{q_2^2 + p_2^2}{4}}, \quad
   m_5 = R - \frac{q_2^2 + p_2^2}{2}, \quad
   m_6 = p_2 \sqrt{R - \frac{q_2^2 + p_2^2}{4}},
   \end{array}
   \end{equation}
with 
\[ q_i = q_i^\ast + \bar q_i, \quad p_i = p_i^\ast + \bar p_i, \]\smallskip\par
\noindent and such that
\[
\begin{array}{lcllcl}
   L &=& \frac{1}{2}( p_1^{\ast 2} + q_1^{\ast 2} ), &
   R &=& \frac{1}{2}( p_2^{\ast 2} + q_2^{\ast 2} ), \\[1.5ex]
   q_1^\ast &=& \frac{-\sqrt{2} \, m_1^\ast}{( m_1^{\ast 2} + m_3^{\ast 2} )^{1/4}}, &
   q_2^\ast &=& \frac{-\sqrt{2} \, m_4^\ast}{( m_4^{\ast 2} + m_6^{\ast 2} )^{1/4}}, \\[1.5ex]
   p_1^\ast &=& \frac{\sqrt{2} \, m_3^\ast}{( m_1^{\ast 2} + m_3^{\ast 2} )^{1/4}}, &
   p_2^\ast &=& \frac{\sqrt{2} \, m_6^\ast}{( m_4^{\ast 2} + m_6^{\ast 2} )^{1/4}}.\\
\end{array}
\]
\smallskip\par

We notice that the specific values of $L$, $R$, $q_i^\ast$, $p_i^\ast$, $m_i^\ast$ are related to each other through the previous formulae. All of them are ultimately explicitly written as functions of the parameters  $b^\ast$. 
\\

By picking some samples in the region of the parametric plane denoted by ${\mathcal B}_{\rm II}$, we have numerically determined that the unstable ellipsoids of type II have either focus $\times$ focus or centre $\times$ centre $\times$ focus linearisation in the upper part of their region of existence. At some point, the focus $\times$ focus equilibria changes to linearisation of type centre $\times$ centre $\times$ focus. Eventually, a Hamiltonian-Hopf bifurcation occurs and they become linearly stable (centre $\times$ centre $\times$ centre $\times$ centre) with two positive signs in front of the $\omega_i > 0$ and two negative ones. 
\\
    
On the boundary line $b_1^\ast = 2\, b_2^\ast + b_3^\ast$ the equilibria are unstable of focus $\times$ saddle $\times$ saddle or centre $\times$ centre $\times$ centre $\times$ saddle type. At some point they become linearly stable with linearisation of centre $\times$ centre $\times$ centre $\times$ centre type, the quadratic normal-form Hamiltonian being indefinite with with two positive signs in front of the $\omega_i > 0$ and two negative ones. They are linearly stable in a very narrow strip and then change stability to become unstable with linearisation centre $\times$ centre $\times$ centre $\times$ saddle. This change is in correspondence with a saddle-centre bifurcation, as we shall show with detail in Theorem \ref{saddle_centre}. This exploration is compatible with the findings in \cite{fasso2001stability}, although Fass\`o and Lewis refer to spectral stability, it is indeed linear stability, which is a bit stronger.
\\
    
It is likely that there are two kinds of bifurcations involving type-II ellipsoids: on the one hand, Hamiltonian-Hopf bifurcations, where a linearly-stable equilibrium loses its stable character and two of the four pure imaginary eigenvalues change to become a quadruplet of complex eigenvalues. On the other hand, saddle-centre bifurcations, where a pair of pure imaginary eigenvalues becomes real. To prove that such a bifurcation takes place one has to compute higher-order terms because a mere linear analysis is not enough, as we show in Theorem \ref{saddle_centre} below. Besides, there is a transition line between focus $\times$ focus to centre $\times$ centre $\times$ focus regime, but we do not pay attention to it.
\\

In Fig. \ref{fig:II} we represent in red the curves corresponding to the saddle-centre bifurcations. Here we do not plot the Hamiltonian-Hopf bifurcations because we leave the study of this bifurcation for type-III ellipsoids. They can be seen in Fig. 4 of \cite{fasso2001stability}. An interesting degenerate case is the point in parametric plane where the Hamiltonian-Hopf and the saddle-centre bifurcations meet. 
\\

Now we focus on the quasi-periodic saddle-centre bifurcation. We follow the ideas of 
\cite{BroerHuitemaSevryuk, HanssmannCentreSaddle}, adapting them to our setting. The main result of the section is given next. 
\\

\begin{thm}
\label{saddle_centre}
    Type-II ellipsoids undergo a quasi-periodic saddle-centre bifurcation. Additionally there is a degenerate case corresponding to a tangency between two curves, one regarding a saddle-centre bifurcation and the other one regarding a Hamiltonian-Hopf bifurcation.
\end{thm}
\begin{proof}
We begin our proof in a similar way as we initiated the proof of Theorem \ref{ThmtypeI}. We introduce two rotation matrices ${\mathcal R}( \gamma_1 )$ and ${\mathcal R}( \gamma_2 )$ to construct a symplectic transformation that allows us to get a simpler expression of the quadratic Hamiltonian in normal form. 
\\

Combining the initial change of coordinates (\ref{changeTypeII}) with the two rotation matrices we obtain the symplectic change
\[
   \begin{array}{lcl}
   m_1 &=& -q_1 \sqrt{L - \frac{q_1^2 + p_1^2}{4}} \cos \gamma_1 + p_1 \sqrt{L- \frac{q_1^2 + p_1^2}{4}} \sin \gamma_1, \quad
   m_2 = L -\frac{q_1^2 + p_1^2}{2}, \\[1.7ex]
   m_3 &=& p_1 \sqrt{L - \frac{q_1^2 + p_1^2}{4}} \cos \gamma_1 + q_1 \sqrt{L - \frac{q_1^2 + p_1^2}{4}} \sin \gamma_1, \\[1.7ex]
   m_4 &=& -q_2 \sqrt{R - \frac{q_2^2 + p_2^2}{4}} \cos \gamma_2 + p_2 \sqrt{R - \frac{q_2^2 + p_2^2}{4}} \sin \gamma_2, \quad
   m_5 = R - \frac{q_2^2 + p_2^2}{2}, \\[1.7ex]
   m_6 &=& p_2 \sqrt{R - \frac{q_2^2 + p_2^2}{4}} \cos \gamma_2 + q_2 \sqrt{R - \frac{q_2^2 + p_2^2}{4}} \sin \gamma_2.
   \end{array}
   \]
The angles $\gamma_i$ are introduced with the goal of getting some zero entries in the linearisation matrix that we are going to build. 
\\
   
The relations at the equilibrium are
\[
\begin{array}{lcllcl}
 q_1^\ast &=& \displaystyle \frac{\sqrt{2}( -m_1^\ast \cos \gamma_1 + m_3^\ast \sin \gamma_1)}{( m_1^{\ast 2} + m_3^{\ast 2} )^{1/4}}, \quad
 & p_1^\ast &=& \displaystyle \frac{\sqrt{2}( m_3^\ast \cos \gamma_1 + m_1^\ast \sin \gamma_1)}{( m_1^{\ast 2} + m_3^{\ast 2} )^{1/4}}, \\[2ex]
 q_2^\ast &=& \displaystyle \frac{\sqrt{2}( -m_4^\ast \cos \gamma_2 + m_6^\ast \sin \gamma_2)}{( m_4^{\ast 2} + m_6^{\ast 2} )^{1/4}}, \quad
 & p_2^\ast &=& \displaystyle \frac{\sqrt{2}( m_6^\ast \cos \gamma_2 + m_4^\ast \sin \gamma_2)}{( m_4^{\ast 2} + m_6^{\ast 2} )^{1/4}}.
 \end{array}
\]

We select $\gamma_1$ and $\gamma_2$ so that the terms of the transformed quadratic Hamiltonian containing $\bar b_1 \bar p_1, \bar b_2 \bar p_2, \bar b_1 \bar p_2, \bar b_2 \bar p_1$ are zero. We end up with
\[
\gamma_1 = \arccos \left(\frac{m^\ast_3}{\sqrt{m_1^{\ast 2} + m_3^{\ast 2}}} \right), \quad
\gamma_2 = -\arccos \left(\frac{m^\ast_6}{\sqrt{m_4^{\ast 2} + m_6^{\ast 2}}} \right). \\ \]

Next, we apply Markeev's procedure (see Appendix \ref{Linearization}) to obtain the corresponding diagonal linear normal form in rectangular coordinates, say $z$. The linearisation matrix is 
\bas
   {\mathcal L} = \left( \begin{array}{cccccccc} 
   0 & 0 & 0 & 0 & \ell_{1,5} & \ell_{1,6} & 0 & 0
	\\ \noalign{\medskip}
   0 & 0 & 0 & 0 & \ell_{1,6} & \ell_{2,6} & 0 & 0
    \\ \noalign{\medskip}
   0 & 0 & 0 & 0 & 0 & 0 & \ell_{3,7} & \ell_{3,8}
    \\ \noalign{\medskip}
   0 & 0 & 0 & 0 & 0 & 0 & \ell_{3,8} & \ell_{4,8}
    \\ \noalign{\medskip}
   \ell_{5,1} & \ell_{5,2} & \ell_{5,3} & \ell_{5,4} & 0 & 0 & 0 & 0
	\\ \noalign{\medskip}
   \ell_{5,2} & \ell_{6,2} & \ell_{6,3} & \ell_{6,4} & 0 & 0 & 0 & 0
	\\ \noalign{\medskip}
   \ell_{5,3} & \ell_{6,3} & \ell_{7,3} & \ell_{7,4} & 0 & 0 & 0 & 0
	\\ \noalign{\medskip}
   \ell_{5,4} & \ell_{6,4} & \ell_{7,4} & \ell_{8,4} & 0 & 0 & 0 & 0 \end{array}
	\right),
	\eas 
where the coefficients $\ell_{i,j}$ are given in terms of $b^\ast$. The relationship between both kinds of parameters is much more involved than in case of $S$-ellipsoids. Nevertheless, it has to be expected since the linearisation matrix has less zero blocks than the ones appearing in Appendix \ref{Linearization}. Anyway, we have succeeded in obtaining closed formulae and they are supplied in the {\sc Mathematica} file. The expressions of the $\omega_i$ in terms of $\ell_{i,j}$ are also cumbersome but they can be computed explicitly. The reason is that they are obtained from the roots of the characteristic equation of an $( 8 \times 8 )$-matrix, but this equation contains only even powers in the unknown, say $\lambda$. Thus, it is indeed a polynomial equation of degree four whose roots are derived in closed form. We have achieved this, arriving at formulae of quite big sizes but still manageable to work with them. The related eigenvectors are also provided. The values of $\omega_i$ as functions of $\ell_{i,j}$ are also presented in the {\sc Mathematica} file.
\\

As in the previous cases we apply the procedure due to Markeev and delineated in Appendix \ref{Linearization}. A last step is needed to make the approach valid when the frequency $\omega_4$ vanishes. We set $x_4 \rightarrow \sqrt{\omega_4} x_4$, $y_4 \rightarrow y_4/\sqrt{\omega_4}$, exactly as we did for the $S_2$-ellipsoids in the co-parallel regime. We arrive at the quadratic Hamiltonian function in normal form:
\[ H_2( z ) = -\frac{\omega_1}{2} ( x_1^2 + y_1^2 ) + \frac{\omega_2}{2} ( x_2^2 + y_2^2 ) + 
  \frac{\omega_3}{2} (x_3^2 + y_3^2) - \frac{1}{2} ( \omega_4^2 x_4^2 + y_4^2 ),
\]
with $\omega_i > 0$, $i = 1, 2, 3$, while $\omega_4$ can be positive, zero or pure imaginary such that $\omega_4 = \imath \bar \omega_4$ with $\bar \omega_4 < 0$. The corresponding transformation matrix $\mathcal T$, the one in charge of bringing $H_2$ to normal form, is symplectic and has real entries for $\omega_4 \ge 0$ or pure imaginary. Moreover, $\mathcal T$ depends smoothly on $\omega_4$, thus it is a versal normal form. Up to this step the form of the quadratic Hamiltonian in normal form is the same as the one occurring in the co-parallel regime of the $S_2$-ellipsoids, excepting one sign. However the bifurcation is going to be different, but this will be concluded after analysing the higher-order terms. 
\\
 
As we have seen, the responsible for the bifurcation is the frequency $\omega_4$. Expressing $\omega_4$ as an explicit function of $b_2^\ast/b_1^\ast$ and $b_3^\ast/b_1^\ast$ is hard. Nevertheless, taking into account that the determinant of the linearisation matrix ${\mathcal L}$ is equal to the product of its eigenvalues, so $\vert {\mathcal L} \vert = ( \omega_1 \omega_2 \omega_3 \omega_4 )^2$. Then $\omega_4 = 0$ implies $\vert {\mathcal L} \vert = 0$. This fact allows us to determine an analytical expression of the bifurcation curves subsequent to removing spurious terms. The relevant factor corresponding to the bifurcation lines in the parametric plot is given by a compact formula, and is provided in the {\sc Mathematica} file. It has been depicted in Fig. \ref{fig:II} (red lines $\Phi$ and $\Delta$). Bifurcation line $\Phi$ corresponds to the lower arch in region ${\mathcal B}_{\rm II}$ in \cite{fasso2001stability} (see Fig. 4 (b), (c) and (d)). Line $\Delta$ also appears in Fig. 4 (d) of \cite{fasso2001stability}, although it requires some clarification that we do below.
\\

After applying the linear transformation obtained above to terms of degree three and four we apply a Lie transformation \cite{Deprit} to compute the corresponding normal form up to terms of degree four in rectangular coordinates. We want to see that the requested non-degeneracy conditions needed to prove that a quasi-periodic saddle-centre bifurcation takes place are fulfilled. To achieve this, we pass to complex variables by means of the change (\ref{complex}). 
\\

The linear normal form in complex/real variables, that is, in $Z$ defined for the $S_2$-ellipsoids, has as Hamiltonian function
\[
H_2( Z ) = -\imath \omega_1 X_1 Y_1 + \imath \omega_2 X_2 Y_2 + \imath \omega_3 X_3 Y_3 -\mbox{$\frac{1}{2}$} ( \omega_4^2 X_4^2 + Y_4^2 ).
\\
\]

Due to the structure of the zeroth-order Hamiltonian $H_2$ and taking into account that due to the lack symmetries if compared to $S_2$-ellipsoids, some terms of degree three have to be retained in the transformed Hamiltonian, we impose the first-order normal-form Hamiltonian, which is composed by homogeneous polynomials of degree three in complex/real coordinates, to be of the form
\[ H_3( Z ) = {\mathsf C}_1 X_1 Y_1 X_4 + {\mathsf C}_2 X_2 Y_2 X_4 + {\mathsf C}_3 X_3 Y_3 X_4 + {\mathsf C}_4 X_4^3,
\]
with ${\mathsf C}_i$ some parameters (real or complex) that have to be determined. Notice that we use the same name for the transformed and untransformed coordinates. The reason for the monomials chosen to get $H_3$ is due to the form $H_2$ has and in particular due to the nilpotent part of $H_2$ for $\omega_4 = 0$. The same will happen for higher-order terms. Introducing the generating function $\mathcal{W}_1$ as a homogeneous polynomial in $Z$ of degree three with undetermined coefficients, we impose that the related homological equation be satisfied. This leads to a system of linear equations whose unknowns are the coefficients of $\mathcal{W}_1$ and the ${\mathsf C}_i$. This is an underdetermined system with 120 linear equations and 124 unknowns that has been solved. Coefficients ${\mathsf C}_i$ depend explicitly on the $\ell_{i,j}$, $\omega_i$ and $b^\ast$.
\\

Passing from the $X_i/Y_i$ to the actions $I_i$, $i = 1, 2, 3$, the truncated normal-form Hamiltonian at first order, 
that is, $H_2 + H_3$ reads as
\[
\begin{array}{lcl}
  H^3( I, X_4, Y_4 ) &=& -\omega_1 I_1 + \omega_2 I_2 + \omega_3 I_3 -\mbox{$\frac{1}{2}$} ( \omega_4^2 X_4^2 + Y_4^2 ) \\[1ex]
  && - \, \imath ( {\mathsf C}_1 I_1 + {\mathsf C}_2 I_2 + {\mathsf C}_3 I_3 ) X_4 + {\mathsf C}_4 X_4^3.
  \end{array}
\]

In the process of getting the suitable normal form for the saddle-centre bifurcation we make a shift in $X_4$ 
\[
X_4 = \bar X_4 + X_{40}, \quad
Y_4 = \bar Y_4,
\]
where 
\[
X_{40} = \frac{\omega_4^2}{6 {\mathsf C}_4}
\]
and apply it to $H^3$. In this way we absorb the term in $\bar X_4^2$, providing the denominator of 
$X_{40}$ does not vanish at the bifurcation curve, that is, ${\mathsf C}_4 \neq 0$, arriving at a suitable pattern for proving the existence of a saddle-centre bifurcation. We get
\[
 H^3( I, \bar X_4, \bar Y_4 ) = -\omega_1 I_1 + \omega_2 I_2 + \omega_3 I_3 - \mbox{$\frac{1}{2}$} \bar Y_4^2 + {\mathsf C}_5( I ) + {\mathsf C}_6( I ) \bar X_4  + {\mathsf C}_4 \bar X_4^3,
\]
where ${\mathsf C}_5$ and ${\mathsf C}_6$ depend linearly on the actions and on the parameters of the problem.
\\

In a bid to get the persistence of KAM tori associated to the bifurcation, the normal form $H^3$
is still too degenerate and the related Hessian that we have to check is zero. For this reason we have to compute the order two (second step of the normal form procedure) in order to incorporate a quadratic dependence in the actions and obtain the required rank (three) to prove this persistence.
\\

The second-order normal form, i.e. the terms of degree four in complex/real variables $Z$, 
is of the form
\[
\begin{array}{lcl}
H_4( Z ) &=& Q_1 ( X_1 Y_1 )^2 + Q_2 ( X_2 Y_2 )^2 + Q_3 ( X_3 Y_3 )^2 + Q_4 X_4^4 + Q_5 X_1 Y_1 X_2 Y_2 \\[1ex]
&& + \, Q_6 X_1 Y_1 X_3 Y_3 + Q_7 X_2 Y_2 X_3 Y_3 + Q_8 X_1 Y_1 X_4^2 + Q_9 X_2 Y_2 X_4^2 \\[1ex]
&& + \, Q_{10} X_3 Y_3 X_4^2,
\end{array}
\] 
where the $Q_i$ coefficients are determined, together with the ones of the generating function ${\mathcal W}_2$. This is achieved by solving a linear system of 330 equations and 340 unknowns. After some simplifications and arrangements done with the aim of controlling that the denominators of the monomials forming the generating function do not vanish when $\omega_4 = 0$, we have ended up with concrete expressions for $Q_i$ and ${\mathcal W}_2$, which in turn are explicit functions of $b^\ast$.
\\

Now we consider the truncated normal form at degree four in the (transformed) coordinates $I$, $X_4$, $Y_4$. Hamiltonian $H^4 = H_2 + H_3 + \frac{1}{2} H_4$ reads as
\[ 
H^4( I, X_4, Y_4 ) = F_1( I ) - \mbox{$\frac{1}{2}$} Y_4^2 + F_2( I ) X_4 + F_3( I ) X_4^2 + F_4 X_4^3 + F_5 X_4^4, 
\]
with $F_i( I )$, $i = 1, 2, 3$, polynomials in $I$ that depend on $b^\ast$ whereas $F_4$, $F_5$ are functions of $b^\ast$. Specifically 
\[
\begin{array}{rcl}
F_1( I ) &=& -\omega_1 I_1 + \omega_2 I_2 + \omega_3 I_3 \\[0.8ex] && - \, \mbox{$\frac{1}{2}$}( Q_1 I_1^2 + Q_2 I_2^2 + Q_3 I_3^2 + Q_5 I_1 I_2 - Q_6 I_1 I_3 + Q_7 I_2 I_3 ), \\[1ex]
F_2( I ) &=& -\imath ( {\mathsf C}_1 I_1 + {\mathsf C}_2 I_2 + {\mathsf C}_3 I_3 ), \\[1ex]
F_3( I ) &=& 
-\mbox{$\frac{1}{2}$} \omega_4^2 - \mbox{$\frac{\imath}{2}$} ( Q_8 I_1 + Q_9 I_2 + Q_{10} I_3 ), \\[1ex]
F_4 &=& {\mathsf C}_4, \\[1ex]
F_5 &=& \mbox{$\frac{1}{2}$} Q_4. 
\end{array}
\]
Hamiltonian $H_3$ is taken above without doing the shift. 
\\

We can eliminate the term depending on $X_4^2$ as before, but this time it is a bit more involved. Calling $\psi_i$ the angles conjugate to $I_i$ we introduce the transformation 
\[
\begin{array}{rcl}
X_4 &=& \bar X_4 + X^\ast_{40}( \bar I ), \quad Y_4 = \bar Y_4, \\[1ex]
\psi_i &=& \displaystyle \bar \psi_i + \bar Y_4 \frac{\partial {X^\ast_{40}( \bar I )}}{\partial \bar I_i}, \quad I_i = \bar I_i.
\end{array}
\]
The modification done on $\psi_i$ is due to the fact that $X^\ast_{40}$ depends on the $\bar I$. Then, the change is symplectic. An additional detail is that when we solve the equation for determining $X^\ast_{40}$, there are two possible solutions (it is obtained by solving a second-degree equation) and the right choice depends on the sign of the coefficient of $X_4^3$ in the Taylor expansion for the specific values of $b^\ast$. We remark that $X^\ast_{40}( 0 ) = 0$ when $\omega_4 = 0$.
\\

The resulting (truncated) normal-form Hamiltonian becomes
\[ 
H^4( \bar I, \bar X_4, \bar Y_4 ) = F^\ast_1( \bar I ) - \mbox{$\frac{1}{2}$} \bar Y_4^2 + F^\ast_2( \bar I ) \bar X_4 + F^\ast_4( \bar I ) \bar X_4^3 + F_5 \bar X_4^4, 
\]
with $F^\ast_1$, $F^\ast_2$, $F^\ast_4$ functions of $\bar I$. The term $F_5 \bar X_4^4$ can be considered of higher order for $\bar X_4$ small enough.
\\

At this point we examine the possible resonances introduced in the Lie transformation process. More precisely we have checked whether the denominators of the terms of the generating functions vanish when $\omega_4 = 0$. Focusing on the line $\Phi$, being the approach the same for $\Delta$, we have found three fourth-order resonances, namely  
\[ -\omega_1 + 3\, \omega_3, \quad -2\, \omega_1 + \omega_2 + \omega_3, \quad \omega_1 - \omega_2 + 2\, \omega_3. 
\]
These values are removed from our study and are represented in Fig. \ref{fig:IIReso}. Concretely, we 
have to discard from the parametric plane the points $( b_2^\ast/b_1^\ast, b_3^\ast/b_1^\ast )$ such that the linear combinations of the frequencies given above become zero. We obtain: $-\omega_1 + 3\, \omega_3 = 0$ for $b_2^\ast/b_1^\ast \approx 0.144, 0.272, 0.438$; $-2\, \omega_1 + \omega_2 + \omega_3 = 0$ for $b_2^\ast/b_1^\ast \approx 0.240, 0.468, 0.485$; $\omega_1 - \omega_2 + 2\, \omega_3 = 0$ for $b_2^\ast/b_1^\ast \approx 0.177, 0349, 0.379$. The corresponding values of $b_3^\ast/b_1^\ast$ are obtained after solving $\omega_4 = 0$. By a continuity argument we also remove small neighbourhoods of these points, because some denominators of the formulae become very small. These roots and their neighbourhoods have to be discarded from our analysis.
\\

\begin{figure}[ht]
    \centering
    \includegraphics[width=0.6\textwidth]{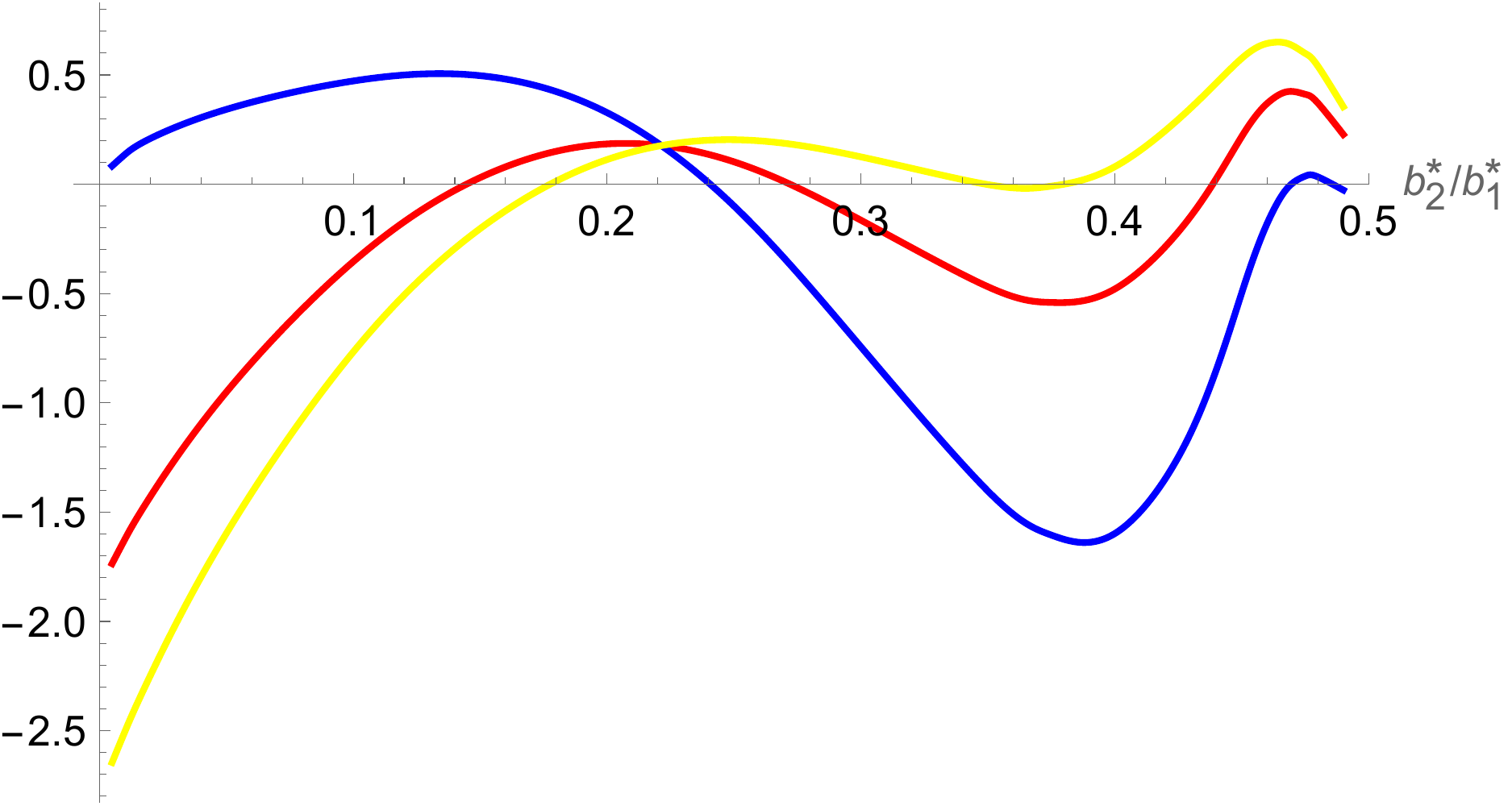}
    \caption{Resonances appearing in the computation of the normal form to determine the quasi-periodic saddle-centre bifurcation for type-II ellipsoids. The red line corresponds to $-\omega_1 + 3\, \omega_3$, the blue one to $-2\, \omega_1 + \omega_2 + \omega_3$ and the yellow one to $\omega_1 - \omega_2 + 2\, \omega_3$}
    \label{fig:IIReso}
\end{figure}

We deal now with the conditions that $H^4$ has to fulfill in a bid to establish the occurrence of the saddle-centre bifurcation. We apply Theorem 4.4 of \cite{hanssmann2006local}. We need to study the behaviour of $F^\ast_2$ and $F^\ast_4$, respective coefficients of $\bar X_4$ and $\bar X_4^3$. In particular we have to prove that $F^\ast_2$ vanishes for $\bar I = 0$, $\omega_4 = 0$ but $F^\ast_4( 0 ) \neq 0$ for $\omega_4 = 0$. On the one hand, as for $\omega_4 = \bar I = 0$ we know that the coefficient $X^\ast_{40}$ vanishes, one has that $F^\ast_2( 0 ) = F_2( 0 ) = 0$. On the other hand we observe that the coefficient $F^\ast_4( 0 )$ (with $\omega_4 = 0$) is equal to the coefficient ${\mathsf C}_4$, which is given in an explicit way on the bifurcation line $\Phi$ in terms of $b^\ast$ and placed in the {\sc Mathematica} file. However, in a bid to check that it does not vanish on the line $\Phi$ we need to proceed numerically though with very high precision in the computations. Thus, we check how it evolves along the bifurcation curve $\Phi$. We depict in Fig. \ref{fig:cof3} the variation of this coefficient when $b^\ast$ is in $\Phi$. 
\begin{figure}
    \centering
  \includegraphics[width=0.6\textwidth]{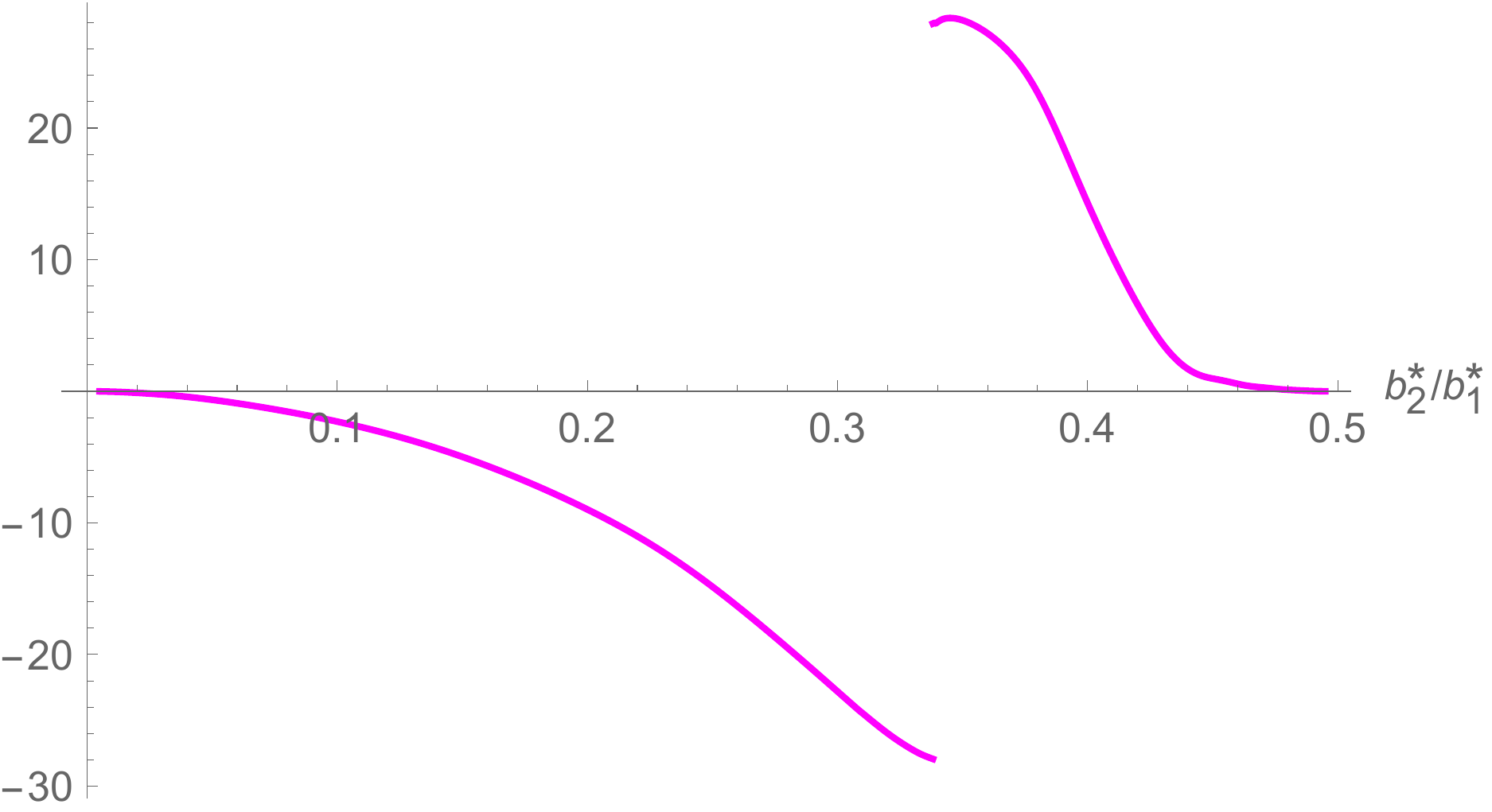}
    \caption{Coefficient ${\mathsf C}_4$ with $\omega_4 = 0$ evaluated along the line $\Phi$. It never touches the horizontal axis. In fact, when $b_2^\ast/b_1^\ast \approx 0$, the coefficient is $-0.0078$ and it goes on decreasing up to the point $b_2^\ast/b_1^\ast \approx 0.3381$ where a discontinuity occurs. Its corresponding $b_3^\ast/b_1^\ast$ jumps from approximately $-27.9366$ to $27.9366$. Then ${\mathsf C}_4$ becomes positive and remains above the horizontal axis. When $b_2^\ast/b_1^\ast \approx 1/2$, ${\mathsf C}_4$ takes its minimum value, $0.0038$}
    \label{fig:cof3}
\end{figure}
\\

The values of $b^\ast$ related to the resonances presented above and small balls around them are not taken into account for our analysis. The value $b_2^\ast/b_1^\ast \approx 0.3381$ (see Fig. \ref{fig:cof3}) has nothing to do with the resonances. In fact, it is related to the manner the linear normal form has been built. In spite of that, the linear change is properly defined for this value. The corresponding ratio $b_3^\ast/b_1^\ast$ (obtained imposing that the point $b^\ast$ belongs to $\Phi$) is approximately $0.1163$ though it does not affect the overall study since this coefficient does not vanish along the bifurcation curve.  
\\

In a final step we prove the persistence of the KAM tori related to the bifurcation. We introduce 
\[ 
\tilde{c}( I ) = F^\ast_2( \bar I ), \quad  
\tilde{\omega}_i( I ) = \displaystyle \frac{\partial F^\ast_1( \bar I )}{\partial \bar I_i}, 
\]
and define the map 
\[ 
\xi: \bar I \rightarrow \left( \tilde{c}( \bar I ), \tilde{\omega}_1( \bar I ), \tilde{\omega}_2( \bar I ), \tilde{\omega}_3( \bar I ) \right). 
\] 
We prove that it is a submersion at $\omega_4 = 0$, i.e. that $\xi$ is differentiable with its differential being surjective everywhere. We get $D \tilde{c}( 0 ) = -\imath ( {\mathsf C}_1, {\mathsf C}_2, {\mathsf C}_3 ) \neq ( 0, 0, 0 )$ along the curve $\Phi$, where we have discarded the resonance values but not the point $( b_2^\ast/b_1^\ast, b_3^\ast/b_1^\ast ) \approx ( 0.3381, 0.1163 )$. More precisely, evaluating ere the norm of $D \tilde{c}( 0 )$ along the bifurcation curve we have noticed that it is always positive, and tends to zero when approaching the right-end point of the curve, where it takes its minimum value, around $0.0089$. Analogously to the analysis made in Section \ref{SectionS2} we form the ($3 \times 3$)-matrix $\mathtt M$ where its first row is $( \tilde{\omega}_1( \bar I ), \tilde{\omega}_2( \bar I ), \tilde{\omega}_3( \bar I ))$, its second row is the partial derivative of the first one with respect to $\bar I_1$ and its third row is the partial derivative with respect to $\bar I_2$. The determinant of $\mathtt M$ evaluated at $\bar I = \omega_4 = 0$ is different from zero along the bifurcation curve $\Phi$, excepting a discrete set of points which have nothing to do with the resonances dealt with above. For those points where $| M | = 0$ we form the matrices but instead of deriving with respect to $I_1$, $I_2$, we do it with respect to $I_1$, $I_3$ or to $I_1$, $I_2$. We have checked that at least one of these two determinants does not vanish
on the points where $M$ is singular, so we conclude that the rank of $M$ (or of one of the other possibles matrices) is three. Of course, the values already removed where the analysis cannot be applied. Hence, it is concluded that the map introduced above is a submersion. The calculations are given in the {\sc Mathematica} file. This allows us to conclude the persistence of the invariant tori that interplay in the bifurcation. 
\\

As the sign of $\bar Y_4^2$ in $H^4$ is always negative, the KAM tori arise when the product $F_2^\ast(\bar I ) F_4^\ast( \bar I )$ is negative. Actually, this product is a function with non-trivial dependence on $\bar I$, as it contains rational terms and square roots, hence it can take positive, negative or complex values, according to the relative values among the $\bar I_i$. Thus, when $F_2^\ast(\bar I ) F_4^\ast( \bar I ) < 0$ there is a family of elliptic $3$-tori and another one of hyperbolic type and dimension $2$. These two families merge to become one on the bifurcation curve $\Phi$. Then, $F_2^\ast( \bar I ) = 0$ and the resulting tori become parabolic. The tori disappear when $F_2^\ast( \bar I ) F_4^\ast( \bar I ) > 0$. The parabolic $3$-tori also persist under perturbation.
\\
   
To conclude the proof of the theorem we have to analyse the point in the parametric plane where two types of bifurcations take place. The corresponding normal-form Hamiltonian has  been also obtained analytically. This tangency point appears in Fig. \ref{fig:Degenerate}. We have deferred the study of one of the Hamiltonian-Hopf bifurcations that take place in the cases of ellipsoids of types I, II and III to the next section. Here we only show how the two bifurcations are experienced by the Riemann ellipsoids in a single point of the parametric plane. This is indeed the only point where we have observed two bifurcations. 
  \begin{figure}[ht]
    \centering
    \includegraphics[width=0.6\textwidth]{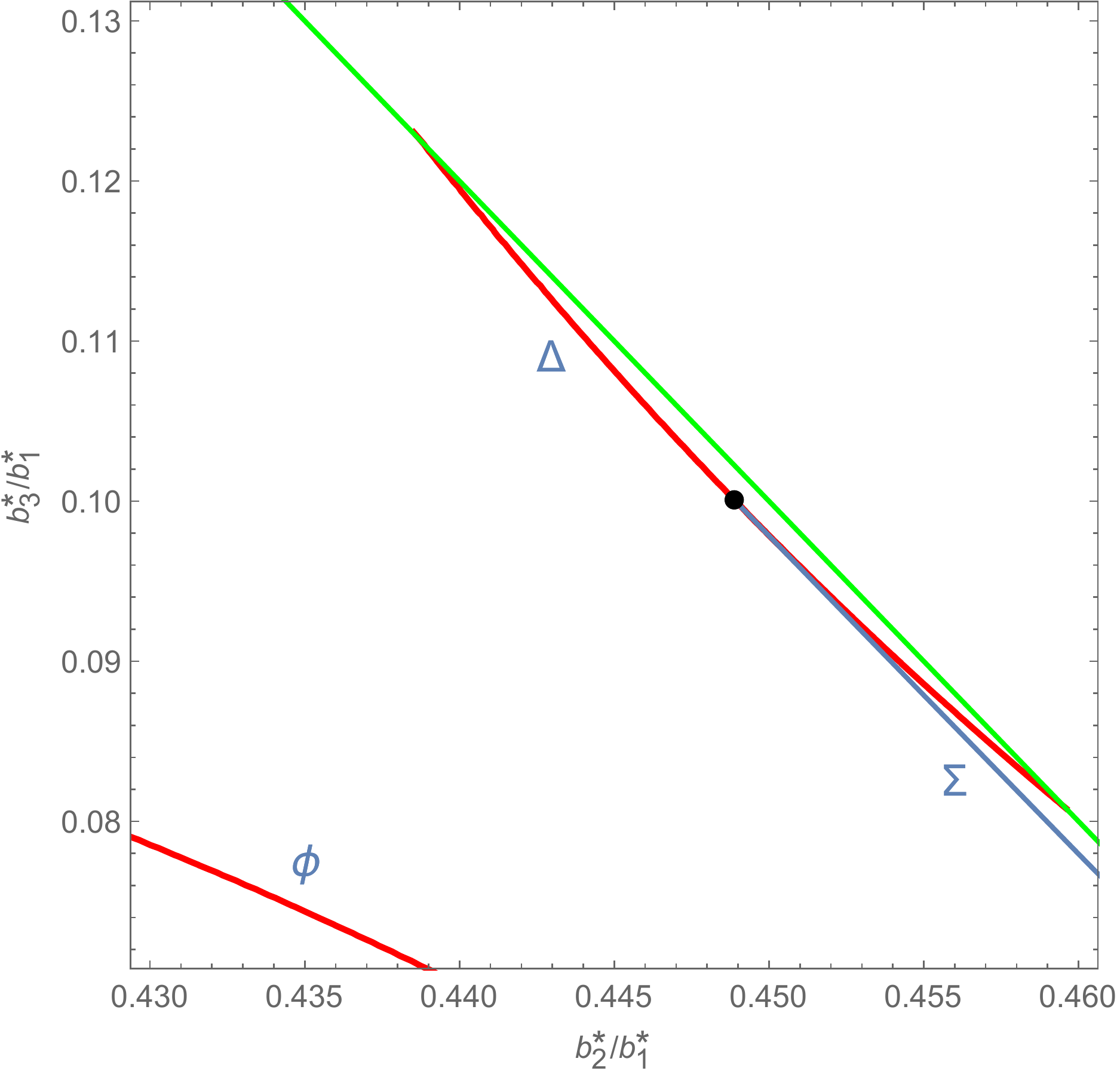}
    \caption{Tangency point between the Hamiltonian-Hopf bifurcation (in blue) and the saddle-centre bifurcation (in red). The tangency occurs at the black point. The green line corresponds to the boundary of ${\mathcal B}_{\rm II}$}
    \label{fig:Degenerate}
\end{figure}
\medskip\par

 In Section \ref{TypeIII} we shall explain how to determine the Hamiltonian-Hopf bifurcation lines, including the analysis of higher-order terms. For the moment we only say that in the curve $\Delta$ (see Fig. \ref{fig:Degenerate}) we have detected a point such that the linearised Hamiltonian system has also a behaviour related to a Hamiltonian-Hopf bifurcation. In fact, picking some points on the blue line $\Sigma$ (different from the black point) as well as other nearby points, we have checked that the transition in stability related to the crossing of the blue line corresponds to the occurrence of a quasi-periodic Hamiltonian-Hopf bifurcation. We name $\Sigma$ this bifurcation curve. 
\\

In particular, at the tangency point we observe that the four eigenvalues of the linearisation matrix of the Hamiltonian-Hopf bifurcation become zero. We follow a semi-numerical approach to determine the exact point where the two bifurcations take place at the same time. As we are looking for a highly-degenerate point we impose that the characteristic polynomial of the linearisation matrix has a zero eigenvalue with multiplicity $4$. This is because the four eigenvalues on the Hamiltonian-Hopf bifurcation are in $1$:$-1$ resonance and two of them are related to the fact that in the saddle-centre bifurcation they are always zero. Thus, we have that $\omega_3 = \omega_4 = 0$ while $\omega_1$, $\omega_2$ remain positive. The approximate coordinates in the parametric plane are $( 0.44890, 0.10002 )$. 
The matrix of the eigenvectors corresponding to the linearisation around this point has rank $5$.
\\

The analysis regarding linear stability of the different regions around the tangency point is as follows:

\begin{itemize}
{\small
\item [(a)] At the intersection of the two bifurcations the linearisation is centre $\times$ centre $\times$ degenerate $\times$ degenerate. 

\item [(b)] On the saddle-centr bifurcatione, to the left of the tangency point, the linearisation is centre $\times$ centre $\times$ saddle $\times$ degenerate. 

\item [(c)] On the saddle-centre bifurcation, to the right of the tangency point, the linearisation is centre $\times$ centre $\times$ centre $\times$ degenerate. 

\item [(d)] Below the tangency point we get centre $\times$ centre $\times$ focus. 

\item [(e)] On the Hamiltonian-Hopf bifurcation line, we have centre $\times$ centre $\times$ centre $\times$ centre, with the last two centres in $1$:$-1$ non-semisimple resonance. 

\item [(f)] Between the two bifurcation lines we find centre $\times$ centre $\times$ centre $\times$ 
centre. 

\item [(g)] Above the saddle-centre bifurcation line it is: centre $\times$ centre $\times$ centre $\times$ saddle.
}
\end{itemize}

In Figs. 4(c) and especially 4(d) of paper \cite{fasso2001stability} we can observe the 
two bifurcation lines very close to the boundary $b_1^\ast = 2\, b_2^\ast + b_3^\ast$. It appears 
that the (red) curve $\Delta$ corresponding to the saddle-centre bifurcation should be continued 
to meet again the boundary of ${\mathcal B}_{\rm II}$. Anyway, the fact that the little region 
between $\Delta$ and $\Sigma$ is spectrally (and linearly) stable is compatible with our conclusions
(item (f) above).
\\
 \end{proof}

\begin{remark}
In like manner the Hamiltonian pitchfork bifurcation studied in Section 
\ref{SectionS2}, when the KAM elliptic $3$-tori persist they are surrounded by invariant
$4$-tori. See also Remark \ref{remposiS2b}.
\\
\label{remposiIIa}    
\end{remark}

\begin{remark}
Up to our knowledge this is the first time that a degenerate situation produced by the coalescence of two quasi-periodic Hamiltonian bifurcations, one of saddle-centre type and the other one of Hamiltonian-Hopf type, is reported. A deeper analysis regarding the co-existence of invariant tori of various dimensions as well as their invariant manifolds could be carried out, leading to an interesting dynamics of these ellipsoids around the tangency point, but it is out of the scope of this paper.
\\
\label{remposiIIb}    
\end{remark}

\section{Quasi-periodic Hamiltonian-Hopf bifurcation of type-III ellipsoids}
\label{TypeIII}

This section is devoted to the study of the stability and bifurcations of type-III ellipsoids. 
Recall that their domain of existence is 
\[ {\mathcal B}_{\rm III} = \Big\{ b \in {\mathcal B} \, : \, b_1^\ast \geq b_2^\ast + 2\, b_3^\ast, G( b_1^\ast, b_2^\ast, b_3^\ast ) > 0 \Big\},
\]
where $G$ is given in (\ref{eq:Gs}), see Fig. \ref{fig:III}. Analogously to type-I ellipsoids there are only quasi-periodic Hamiltonian-Hopf bifurcations. We have made the complete analysis of one of these bifurcation lines analytically, excepting the checks on the non-degeneracy of higher-order terms
for the occurrence of the bifurcation, as well as the checks on the persistence of KAM tori, where we have given values along the curve in the parametric plane. As in the cases of the bifurcations studied previously, there are resonances of orders $3$ and $4$ that lead to some small balls in the parametric plane that are excluded from the analysis because the normal-form Hamiltonians are not well defined. Besides, to achieve the persistence of the invariant tori, following \cite{broer2007quasi} we have imposed on some frequencies Diophantine conditions, as we shall mention later on.
\\

\begin{figure}[ht]
    \centering
    \includegraphics[width=0.6\textwidth]{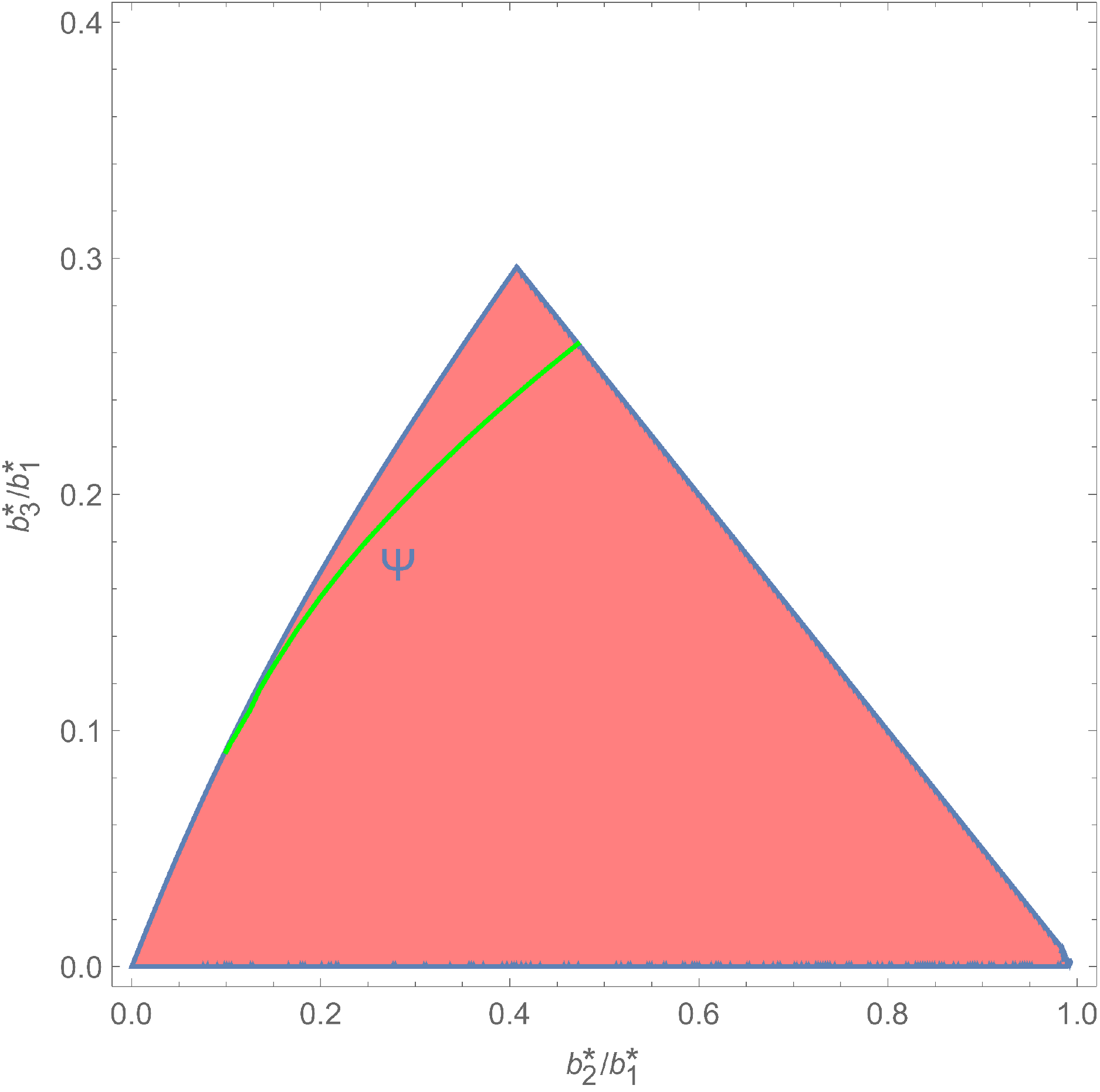}
    \caption{${\mathcal B}_{\rm{III}}$: Region of existence of the type-III Riemann ellipsoids. The green curve corresponds to the Hamiltonian-Hopf bifurcation we have chosen. This bifurcation line ends at the point $( 0, 0 )$ being tangent to the boundary $G = 0$ at this point}
    \label{fig:III}
\end{figure}

First of all we take the ellipsoid with $( S^2_L \times S^2_R )$-coordinates $( \mu^{+}_{N}( b_1^\ast, b_2^\ast, b_3^\ast ),$ $\mu^{-}_{N}( b_2^\ast, b_1^\ast, b_3^\ast ))$. The adjoint ellipsoid should be analysed equivalently. Parameters $L$ and $R$ satisfy 
\[ 
\begin{array}{rcl}
L &=& \sqrt{
N^{R}_{+}( b_1^\ast, b_2^\ast, b_3^\ast )^2 +
N^{R}_{-}( b_2^\ast, b_1^\ast, b_3^\ast )^2}, \\[1.5ex]
R &=& \sqrt{
N^{R}_{-}( b_1^\ast, b_2^\ast, b_3^\ast )^2 +
N^{R}_{+}( b_2^\ast, b_1^\ast, b_3^\ast )^2}. 
\end{array}
\]
We introduce the symplectic change of coordinates
\[
   \begin{array}{lcl}
   b_i &=& b_i^\ast + {\bar b_i}, \quad
   c_i = {\bar c_i}, \\[1ex]
   m_1 &=& p_1 \sqrt{L - \frac{q_1^2 + p_1^2}{4}}, \quad
   m_2 = -q_1 \sqrt{L - \frac{q_1^2 + p_1^2}{4}}, \quad
   m_3 = L -\frac{q_1^2 + p_1^2}{2}, \\[1.3ex]
   m_4 &=& p_2 \sqrt{R - \frac{q_2^2 + p_2^2}{4}}, \quad
   m_5 = -q_2\sqrt{R - \frac{q_2^2 + p_2^2}{4}}, \quad
   m_6 = R - \frac{q_2^2 + p_2^2}{2},
   \end{array}
   \]
   with 
\[
\begin{array}{l}
   q_i = q_i^\ast + \bar q_i, \quad
   p_i = p_i^\ast + \bar p_i, \\[1.5ex]
   L = \frac{1}{2}( p_1^{\ast 2} + q_1^{\ast 2} ), \quad
   R = \frac{1}{2}( p_2^{\ast 2} + q_2^{\ast 2} ), \\[1.5ex]
   q_1^\ast = \frac{-\sqrt{2} \, m_2^\ast}{( m_1^{\ast 2} + m_2^{\ast 2} )^{1/4}}, \quad
   q_2^\ast = \frac{-\sqrt{2} \, m_5^\ast}{( m_4^{\ast 2} + m_5^{\ast 2} )^{1/4}}, \\[1.5ex]
   p_1^\ast = \frac{\sqrt{2} \, m_1^\ast}{( m_1^{\ast 2} + m_2^{\ast 2} )^{1/4}}, \quad
   p_2^\ast = \frac{\sqrt{2} \, m_4^\ast}{( m_4^{\ast 2} + m_5^{\ast 2} )^{1/4}}.
   \end{array}
\\[1.2ex]
\]

We have checked by evaluating the normal form numerically in a sample of points that both on the boundary $b_1^\ast = b_2^\ast + 2\, b_3^\ast$ and in the interior of region ${\mathcal B}_{\rm{III}}$ there is either instability of the type centre $\times$ centre $\times$ focus or linear stability of the type centre $\times$ centre $\times$ centre $\times$ centre. In the latter case the linear normal form is indefinite with one negative frequency. We prove that the change of stability occurs through a Hamiltonian-Hopf bifurcation. A linearly-stable system loses its stable character and two of the four pure imaginary eigenvalues change to become a quadruplet of complex eigenvalues.
\\

We state the result regarding the bifurcation curve we have selected for our study. 
\begin{thm}
    Type-III ellipsoids undergo a quasi-periodic Hamiltonian-Hopf bifurcation on the curve $\Psi$.
\end{thm}
\begin{proof}
For the analysis of the Hamiltonian-Hopf bifurcation we introduce two rotation matrices with convenient angles $\gamma_1$ and $\gamma_2$ to simplify the Hamiltonian. At the equilibrium we get
\[
\begin{array}{lcllcl}
 q_1^\ast &=& \displaystyle \frac{-\sqrt{2}( m_2^\ast \cos \gamma_1 + m_1^\ast \sin \gamma_1 )}{( m_1^{\ast 2} + m_2^{\ast 2} )^{1/4}}, \quad
 & p_1^\ast &=& \displaystyle \frac{\sqrt{2}( m_1^\ast \cos \gamma_1 - m_2^\ast \sin \gamma_1 )}{( m_1^{\ast 2} + m_2^{\ast 2} )^{1/4}}, \\[1.8ex]
 q_2^\ast &=& \displaystyle \frac{-\sqrt{2}( m_5^\ast \cos \gamma_2 + m_4^\ast \sin \gamma_2 )}{( m_4^{\ast 2} + m_5^{\ast 2} )^{1/4}}, \quad
 & p_2^\ast &=& \displaystyle \frac{\sqrt{2}( m_4^\ast \cos \gamma_2 - m_5^\ast \sin \gamma_2 )}{( m_4^{\ast 2} + m_5^{\ast 2} )^{1/4}}.
 \end{array}
\]
The angles $\gamma_1$ and $\gamma_2$ are selected so that the terms factorised by $\bar b_1 \bar p_1, \bar b_2 \bar p_2, \bar b_1 \bar p_2, \bar b_2 \bar p_1$ are zero. It yields
\[
\gamma_1 = -\arccos \left( \frac{m^\ast_1}{\sqrt{m_1^{\ast 2} + m_2^{\ast 2}}} \right), \quad
\gamma_2 = -\arccos \left( \frac{m^\ast_4}{\sqrt{m_4^{\ast 2} + m_5^{\ast 2}}} \right). \\ \]

The linearisation matrix $\mathcal L$ takes the same block form as in the proof of the saddle-centre bifurcation of the previous section. The explicit entries have been obtained in terms of $b^\ast$ and are quite big. As well, the related frequencies $\omega_i$ can be derived in terms of the $\ell_{i,j}$ after solving the characteristic equation. Additionally the eigenvectors have been computed successfully. All this material is provided in the {\sc Mathematica} file.
The frequencies $\omega_i$ are obtained as usual from the eigenvalues of $\mathcal L$. Close enough to the bifurcation curve we intend to study we have $\omega_1, \omega_2 > 0$, whereas $\omega_3, \omega_4 > 0$ in the stable part of the bifurcation, $\omega_3 = \omega_4 > 0$ on the bifurcation curve and $\omega_3, \omega_4$ are complex of the form $\omega_3 = -a - \imath b$, $\omega_4 = a - \imath b$, with $a, b$ positive. This latter choice of the $\omega_i$ is done to make the process compatible to the way {\sc Mathematica} handles the eigenvalues of $\mathcal L$ when they are complex.
\\

To detect a bifurcation curve related to a Hamiltonian-Hopf bifurcation we proceed as follows. By observing that two of the four degrees of freedom have to be in $1$:$-1$ non-semisimple resonance we impose two conditions: (i) the frequencies $\omega_3$, $\omega_4$ are the same; (ii) the determinant of the matrix formed by the eigenvectors is zero. By doing so one might encounter other possible bifurcations but at least among them the Hamiltonian-Hopf bifurcations regarding frequencies $\omega_3$, $\omega_4$. In practice this is a long process and we have used a shortcut. We select a point on the bifurcation curve $\Psi$ by fixing $b_2^\ast/b_1^\ast$ and try to get the ratio $b_3^\ast/b_1^\ast$ such that the determinant formed by the eigenvectors of $\mathcal L$ vanishes. We use the secant method instead of the Newton-Raphson one in a bid to avoid the calculation of the Jacobian, since the function we use, i.e. the determinant of the eigenvectors, is very large. The convergence of the approach based on the secant method is satisfactory, as we get the desired points on the bifurcation with a few iterations. With this method we have obtained the line $\Psi$. More precisely, we have solved the equations with high accuracy, determining $23$ points along the curve, so that the resulting determinants on the bifurcation line are all upper bounded by $10^{-17}$. We have also checked that the frequencies $\omega_3$, $\omega_4$ on the bifurcation line are basically in $1$:$-1$ resonance. 
\\

Our goal now is getting a quadratic Hamiltonian function in normal form given by
\[
H_2( z ) = \frac{\omega_1}{2} ( x_1^2 + y_1^2 ) + \frac{\omega_2}{2}( x_2^2 + y_2^2 ) + 
 \frac{1}{2} ( x_3^2 + x_4^2 ) + \frac{M}{2}( y_3^2 + y_4^2 ) + N ( x_3 y_4 - x_4 y_3 ),
\]
for $z = ( x_1, x_2, x_3, x_4, y_1, y_2, y_3, y_4 )$ a set of rectangular coordinates, as well as the linear transformation that brings the Hamiltonian $H_2( u )$ with $u = 
( \bar b_1, \bar b_2, \bar q_1, \bar q_2$, $\bar c_1, \bar c_2, \bar p_1, \bar p_2 )$ to $H_2( z )$.

Notice that degrees of freedom $x_1/y_1$ and $x_2/y_2$ are uncoupled from the other ones, which are in fact the responsible of the bifurcation. The parameters $M$ and $N$ are functions of the frequencies $\omega_3$ and $\omega_4$. More specifically, when $\omega_3$ and $\omega_4$ are complex then,
 \[ 
 M = \mbox{$\frac{1}{4}$}( \omega_3 + \omega_4 )^2, \quad 
 N = \mbox{$\frac{1}{2}$}( -\omega_3 + \omega_4 ).
 \]
 In the rest of cases
\[ 
 M = \mbox{$\frac{1}{4}$}( \omega_3 - \omega_4 )^2, \quad 
 N = \mbox{$\frac{1}{2}$}( \omega_3 + \omega_4 ).
 \\
\]
With these choices of $M$ and $N$ we know that near the bifurcation curve one has $N > 0$ and $M$ is real but close to zero.

At this point, we adapt to our needs the procedure presented in \cite{Schmidt}. In particular, an important point that we request is that the eigenvectors used to compute the normal form make sense even for the degenerate case, that is, on the bifurcation. The reason is that on the bifurcation curve, the Hamiltonian function $H_2$ has non-null nilpotent part. Then, there is not a basis of eigenvectors, but the rank of the matrix containing the eigenvectors is six. Thus, it is still possible to use two eigenvectors out of four to build the transformation matrix. This considerably simplifies the construction of the linear normal form. 
\\

Thus, applying the procedure described in Appendix \ref{Linearization} we get the vectors $s_1$, $s_2$, $r_1$, $r_2$ together with the positive scalars $k_3$, $k_4$. We also obtain the vectors 
for $s_3$, $s_4$, $r_3$, $r_4$, noticing that $s_i$, $r_i$ are not independent when $\omega_3 = \omega_4$. Now we write down the symplectic matrix $\mathcal T$ responsible of the transformation as
\begin{equation}
{\mathcal T} = 
( -k_1 s_1 \,,\, -k_2 s_2 \,,\, t_3 \,,\, t_4 \,,\, k_1 r_1 \,,\, k_2 r_2 \,,\, t_7 \,,\, t_8 )^T.
\\
\label{eigenvectorsHH}
\end{equation}
such that 
\[
\begin{array}{rclrcl}
t_3 &=& ( 0, 0, 0, 0, \tau_{3,5}, \tau_{3,6}, \tau_{3,7}, \tau_{3,8} ), \,\, &
t_4 &=& ( \tau_{4,1}, \tau_{4,2}, \tau_{4,3}, \tau_{4,4}, 0, 0, 0, 0 ), \\[1ex]
t_7 &=& A_1 s_3 + A_2 s_4, \,\, &
t_8 &=& A_1 r_3 - A_2 r_4.
\end{array}
\\ \]
That is, four of the eight columns of $\mathcal T$ are built in the same way as in Markeev's procedure 
of Appendix \ref{Linearization} for dealing with elliptic points. Here we still need to determine the
entries $\tau_{i,j}$ as well as the coefficients $A_1$, $A_2$. We stress that we can take advantage of the block form of matrix ${\mathcal L}$ for the sake of setting four zero entries in the vectors $t_3$, $t_4$. 
\\

We determine the unknown quantities by imposing two conditions: (i) ${\mathcal L} \, {\mathcal T} = {\mathcal T} \, {\mathcal U}$, where ${\mathcal U}$ is the Hamiltonian matrix associated to $H_2 ( z )$; (ii) $\mathcal T$ is symplectic, then ${\mathcal T}^T \, {\mathcal J}_8 \, {\mathcal T} = {\mathcal J}_8$. The first condition simply says that ${\mathcal L}$ is transformed into $\mathcal U$ by means of $\mathcal T$.
\\

As we wish to get the transformation valid in the linearly stable and unstable regimes and on the bifurcation line we need to proceed carefully. Besides we require ${\mathcal T}$ to be real. 
We distinguish between being on the bifurcation line or outside but close to it. When we pick a point
$b^\ast$ that lies in the stable part of the parametric plane, we compute $\tau_{i,j}$ from condition (i)
and $A_1$, $A_2$ from (ii). Alternatively we can also get $A_2$ from $A_1$ observing that 
$A_2/A_1 = \sqrt{\bar n_3/\bar n_4}$ (with $\bar n_i$ introduced in Appendix \ref{Linearization}).
We remark that $A_1, A_2 > 0$ in the linearly stable part and complex conjugate in the unstable one. 
On the bifurcation curve we determine $\tau_{3,5}$, $\tau_{3,6}$, $\tau_{3,7}$, $\tau_{3,8}$, $\tau_{4,2}$, $\tau_{4,3}$ and $\tau_{4,4}$ from (i). Although now $\bar n_3 = \bar n_4 = 0$ we still have that
 $A_1 = A_2 > 0$ on the line $\Psi$. As a second step using condition (ii) we obtain $\tau_{4,1}$, $A_1$ (and $A_2$). Finally, when $b^\ast$ is in the unstable part of the parametric plane, we use the same approach as in the stable part but replacing $\omega_3$ by $-\omega_3$. This completes the symbolic construction of $\mathcal T$. The transition between the different regimes
 (stable to unstable through degeneracy) is such that the transformation matrix is smooth with respect to the parameters $b^\ast$. This is the versal normal form of the transformation matrix, see \cite{Arnold}. The final entries of ${\mathcal T}$ are determined in terms of $\ell_{i,j}$, $\omega_i$ and some of them very large. We have placed the calculations with the resulting coefficients in the {\sc Mathematica} file.
\\

Next we want to compute the normal form of the Hamiltonian corresponding to type-III ellipsoids in a neighbourhood of the Hamiltonian-Hopf bifurcation curve, with the aim of establishing the occurrence of the bifurcation. We need to reach terms of degree four in rectangular coordinates for the Hamiltonian normal form, thereby we need two steps in the Lie transformation approach. Indeed, the ultimate goal of this calculation is to check the non-degeneracy conditions needed to prove that a quasi-periodic Hamiltonian-Hopf bifurcation takes place in the Riemann ellipsoid problem. It is expected that analogous approaches apply for other bifurcation curves of the same type in other parts of the parametric plane, not only for type-III but also for types-I and II ellipsoids. We follow the ideas of \cite{Schmidt, hanssmann2006local, broer2007quasi} and references therein.
\\

We develop the Hamiltonian function up to terms of degree four in the $u$ coordinates. 
Then we apply the linear transformation built through the matrix $\mathcal T$, that 
is we write the Hamiltonian function in terms of the $z$ coordinates. We follow a similar approach to that of \cite{Schmidt}, but generalising it to four degrees of freedom. Then, we define the linear transformation to complex coordinates given by
\[
\begin{array}{lcllcl}
x_1 &=& \frac{1}{\sqrt{2}}( X_1 + \imath Y_1 ), &
x_2 &=& \frac{1}{\sqrt{2}}( X_2 + \imath Y_2 ), \\[2ex]
x_3 &=& \frac{1}{\sqrt{2}}( X_3 + X_4 ), &
x_4 &=& \frac{\imath}{\sqrt{2}}( -X_3 + X_4 ), \\[2ex]
y_1 &=& \frac{1}{\sqrt{2}}( \imath X_1 + Y_1 ), &
y_2 &=& \frac{1}{\sqrt{2}}( \imath X_2 + Y_2 ), \\[2ex]
y_3 &=& \frac{1}{\sqrt{2}}( Y_3 + Y_4 ), &
y_4 &=& \frac{\imath}{\sqrt{2}}( Y_3 - Y_4 ), \\[2ex]
\end{array}
\]

Calling $Z = ( X_1, X_2, X_3, X_4, Y_1, Y_2, Y_3, Y_4 )$ the quadratic part of the normal-form Hamiltonian becomes
\[
H_2( Z ) = \imath \omega_1 X_1 Y_1 + \imath \omega_2 X_2 Y_2 + X_3 X_4 + M Y_3 Y_4 + \imath N ( X_3 Y_3 - X_4 Y_4 ).
\]
On the bifurcation line $M = 0$ and $N = \omega_3 = \omega_4$. 
\\

We apply a Lie transformation \cite{Deprit} to compute the higher-order terms in the normal form up to order $2$, that is, two steps of the procedure. The Hamiltonian in normal form of degree three can be taken $0$. The associated generating function, ${\mathcal W}_1$, is determined after solving a linear system of equations whose unknowns are the coefficients of the monomials of ${\mathcal W}_1$. We have $120$ equations for $120$ unknowns. Then we keep on with the second step. The terms of order $2$ are quartic polynomials in the complex variables $Z$. Those that remain in the transformed Hamiltonian are functions of the first integrals 
\[ 
I_1 = \imath X_1 Y_1, \quad I_2 = \imath X_2 Y_2, \quad 
S = \mbox{$\frac{\imath}{2}$}( X_3 Y_3 - X_4 Y_4 ), \quad V = Y_3 Y_4,
\]
or written in the $z$ coordinates
\[
I_1 = \mbox{$\frac{1}{2}$}( x_1^2 + y_1^2 ), \quad I_2 = \mbox{$\frac{1}{2}$}( x_2^2 + y_2^2 ), 
\quad S = \mbox{$\frac{1}{2}$}( x_3 y_4 - x_4 y_3 ), \quad V = \mbox{$\frac{1}{2}$}( y_3^2 + y_4^2 ). \\[1.2ex] 
\]

We put the normal form in terms of the invariants $I_1$, $I_2$, $S$, $V$ and $U = X_3 X_4 = ( x_3^2 + x_4^2 )/2$. Naming $T = ( x_3 y_3 + x_4 y_4 )/2$ the invariants satisfy the constraint  
\[ S^2 + T^2 = U V. \]
Hamiltonian $H_2$ reads as
\[ H_2( I, S, U, V ) = \omega_1 I_1 + \omega_2 I_2 + U + M V + 2 N S, \qquad I = ( I_1, I_2 ). 
\\[1.2ex]
\]

Now we impose that the normal form $H_4$ be expressed as
\[ 
\begin{array}{rcl}
H_4( I, S, V ) &=& -Q_1 I_1^2 - Q_2 I_2^2 + Q_3 S^2 + Q_4 V^2- Q_5 I_1 I_2 
\\[1ex] && -\, \imath Q_6 I_1 S - \imath Q_7 I_1 V - \imath Q_8 I_2 S - \imath Q_9 I_2 V 
+ Q_{10} S V,
\end{array} 
\]
where the coefficients $Q_i$ have to be determined. In effect, it is the case that in absence of resonances between the $\omega_i$, the normal form at any order is always a polynomial in the invariants $I$, $S$ and $U$, see \cite{Ken, Meer, hanssmann2006local}.
\\

The associated homological equation is solved for the coefficients of the monomials of ${\mathcal W}_2$ and the $Q_i$. As in previous normal-form computations we arrive at a system with $330$ linear equations and $340$ unknowns. After some manipulations and simplifications, especially on the coefficients of the generating function, we get a solution for $H_4$ and ${\mathcal W}_2$ that makes sense on the bifurcation curve and in a neighbourhood of it, excepting the resonance values that we shall analyse later. The concrete expressions of these functions are provided in the {\sc Mathematica} file.
\\

Seeking possible null or small denominators in the monomials of the generating functions ${\mathcal W}_i$ is similar to the approach we have described for the pitchfork and saddle-centre bifurcations. However, since the denominators depend also on $M$ and $N$, we put them first in terms of $\omega_3$, $\omega_4$, set $\omega_3 = \omega_4$ and select those combinations between the $\omega_i$ that could become zero or very small at some points $b^\ast$ on the bifurcation curve. Plotting these curves using high precision calculations, we have found six situations such that the resonance curves cross the horizontal axis or are very close to it. It means that for these combinations the generating functions are not well defined on some small neighbourhoods of the points $b^\ast$ on the curve where zero or very small denominators arise. Thence we discard from the Hamiltonian-Hopf bifurcation analysis we perform, see Fig. \ref{fig:resoIII}, otherwise they could lead to erroneous conclusions. In principle, the non-linear normal form for these resonant cases would carry out the appearance of angle-terms and a different analysis would be accomplished. We get the following six resonances together with the approximate values $b_2^\ast/b_1^\ast$ where the small denominators arise in the generating functions: 
\[ 
\begin{array}{rcl}
-2\, \omega_1 + \omega_2 = 0 & \,\mbox{for}\, & b_2^\ast/b_1^\ast \approx 0.223, 0.227, \\[0.8ex]
- \omega_1 + 2\, \omega_4 = 0 & \,\mbox{for}\, & b_2^\ast/b_1^\ast \approx 0.164, 0.468, \\[0.8ex]
\omega_1 - \omega_2 + \omega_4 = 0 & \,\mbox{for}\, & b_2^\ast/b_1^\ast \approx 0.356, \\[0.8ex]
-\omega_1 + 3\, \omega_4 = 0 & \,\mbox{for}\, & b_2^\ast/b_1^\ast \approx 0.108, \\[0.8ex]
-2\, \omega_1 + \omega_2 + \omega_4 = 0 & \,\mbox{for}\, & b_2^\ast/b_1^\ast \approx 0.112, 0.405, \\[0.8ex]
-\omega_2 + 3\,\omega_4 = 0 & \,\mbox{for}\, & b_2^\ast/b_1^\ast \approx 0.250. \end{array}
\]
The related values of $b_3^\ast/b_1^\ast$ are obtained after imposing that $b^\ast$ belongs to the bifurcation curve.
\\
  
\begin{figure}[ht]
    \centering
    \includegraphics[width=0.6\textwidth]{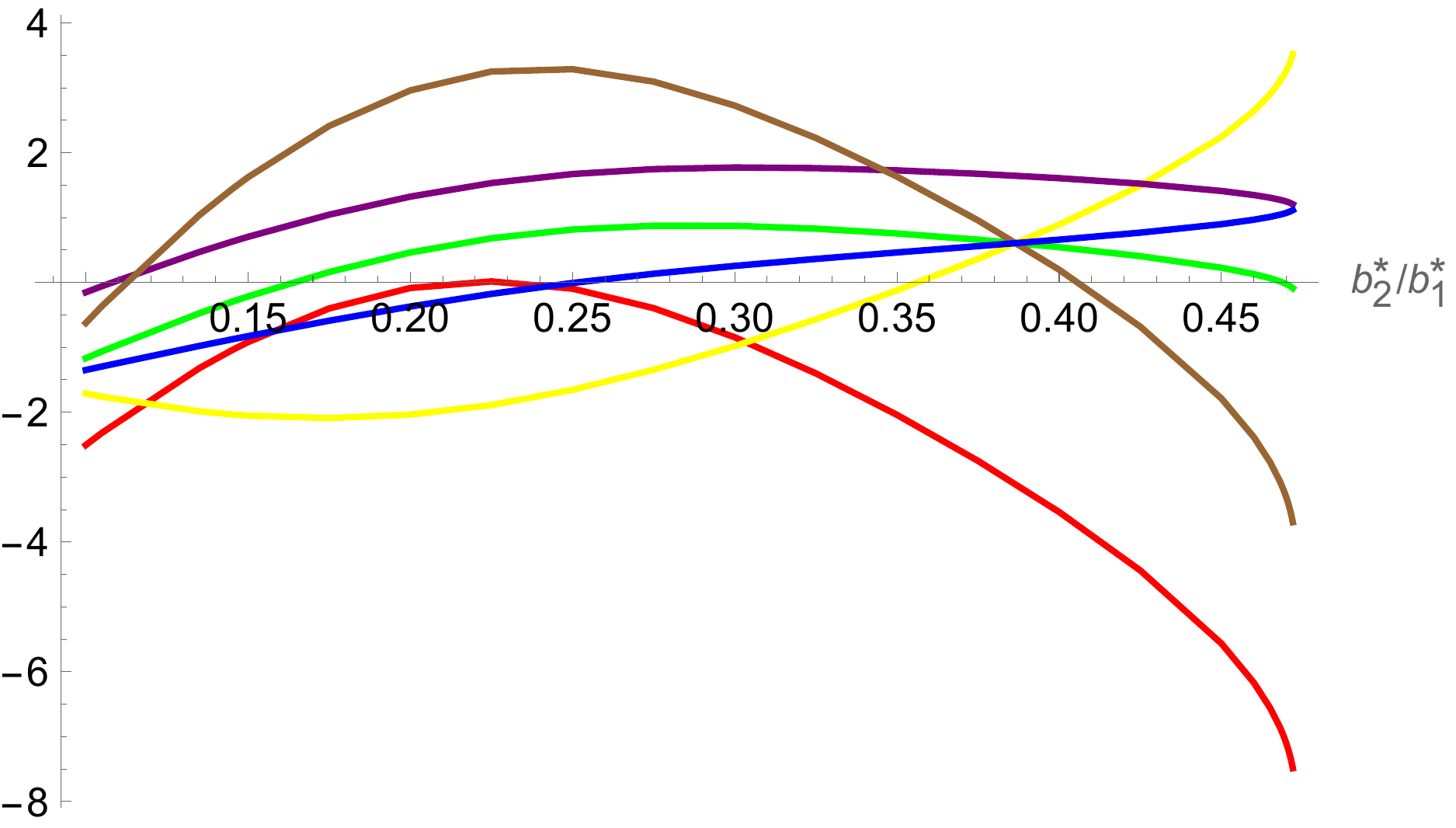}
    \caption{Resonances occurring in the computation of the normal form Hamiltonian on the Hamiltonian-Hopf bifurcation of type-III ellipsoids.
    Third-order resonances are
    $-2\, \omega_1 + \omega_2$ (red),  
    $- \omega_1 + 2\, \omega_4$ (green) and 
    $\omega_1 - \omega_2 + \omega_4$ (yellow).  
    Fourth-order resonances are
    $-\omega_1 + 3\, \omega_4$ (purple),  
    $-2\, \omega_1 + \omega_2 + \omega_4$ (brown) and
    $-\omega_2 + 3\,\omega_4$ (blue) 
    }
    \label{fig:resoIII}
\end{figure}

For the occurrence of the bifurcation our plan is to apply Theorem 4.27 of \cite{hanssmann2006local} to the normal-form Hamiltonian $H^4 = H_2 + \frac{1}{2} H_4$. Reorganising the terms conveniently we end up with
\[ 
\begin{array}{rcl}
H^4( I, S, U, V ) &=& \displaystyle 
    \omega_1 I_1 + \omega_2 I_2 - \mbox{$\frac{1}{2}$} ( Q_1 I_1^2 + Q_2 I_2^2 + Q_5 I_1 I_2 ) 
     \\[1.5ex]
&& \displaystyle + \, U + \left( 2 N - \mbox{$\frac{\imath}{2}$} ( Q_6 I_1 + Q_8 I_2 ) \right) S 
+ \left( M - \mbox{$\frac{\imath}{2}$} ( Q_7 I_1 + Q_9 I_2 ) \right) V \\[1.5ex]
&& \displaystyle  + \, \mbox{$\frac{1}{2}$} ( Q_3 S^2 + Q_4 V^2 + Q_{10} S V ).
\end{array} 
\\
\]

Terms factorised by $S^2$ and $S V$ can be brought to higher order by means of the uneven symplectic scaling proposed by Meyer and Schmidt \cite{Ken} slightly modified to take into account the degrees of freedom associated to $I$. The remaining terms of $H^4$ are relevant in the application of the Hamiltonian-Hopf bifurcation Theorem as it appears in \cite{hanssmann2006local}.
\\

First of all we observe that the factor of $U$ is positive. Then we need to examine the behaviour of the coefficients of $V$ and $V^2$ as they vary on the bifurcation line when $\omega_3 = \omega_4$. We introduce the functions $\tilde{c}( I )$ and $\Omega( I )$ respectively as the coefficients of $V$ and $S$ in the second row of $H^4$, $\tilde{b}( I )$ as the factor of $V^2$ and $F( I )$ as the part of $H^4$ independent of $S$, $U$, $V$, that is,

\[ 
\begin{array}{rcl}
\tilde{c}( I ) &=& M - \mbox{$\frac{\imath}{2}$} ( Q_7 I_1 + Q_9 I_2 ), \\[1ex]
\tilde{b}( I ) &=& \mbox{$\frac{1}{2}$} Q_4, \\[1ex]
F( I ) &=& \omega_1 I_1 + \omega_2 I_2 - \mbox{$\frac{1}{2}$} ( Q_1 I_1^2 + Q_2 I_2^2 + Q_5 I_1 I_2 ), \\[1ex]
\Omega( I ) &=& 2 N - \mbox{$\frac{\imath}{2}$} ( Q_6 I_1 + Q_8 I_2 ).
\end{array} \\
\] 
On the one hand we know that for $I = 0$ and $\omega_3 = \omega_4$, $\tilde{c}( 0 ) = 0$ and $\tilde{b}( 0 ) = Q_4/2$, and this is the term that needs to be controlled. In the real coordinates $z$ this term corresponds to the coefficient of $( y_3^2 + y_4^2 )^2/4$. Although we have the specific formula of $Q_4$ in terms of $b^\ast$, it is gigantic, so we have to perform a numerical check to analyse how this term evolves along the bifurcation curve $\Psi$. We stress that the intermediate steps previous to the numerical check are done symbolically, without replacements of the parameters by specific values. In Fig. \ref{fig:Q412} we depict the variation of $Q_4$ versus the ratio $b_2^\ast/b_1^\ast$ when $b^\ast$ takes values on the bifurcation curve. 
\begin{figure}[ht]
    \centering
    \includegraphics[width=0.475\textwidth]{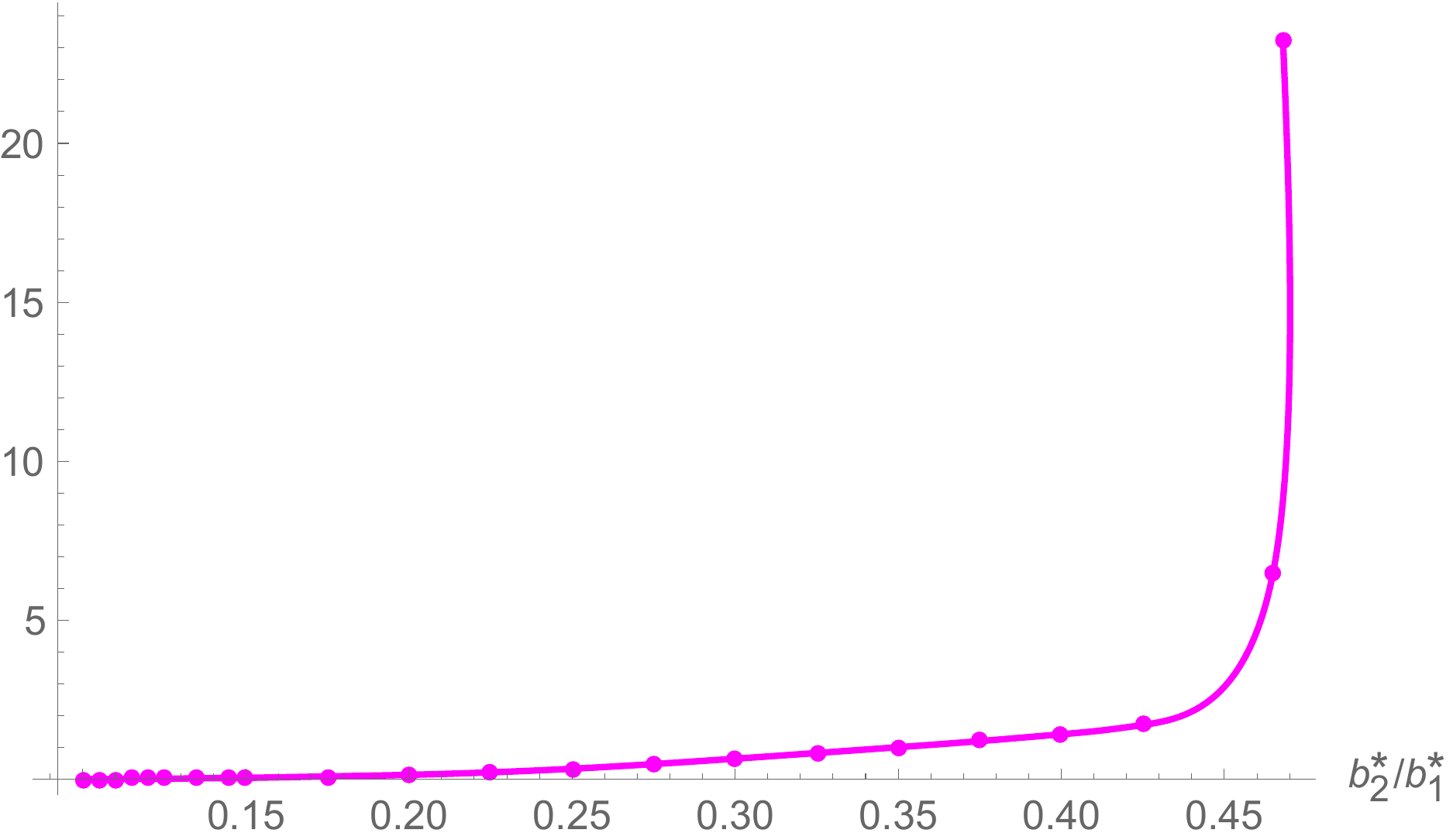} \,\,\,\,
    \includegraphics[width=0.475\textwidth]{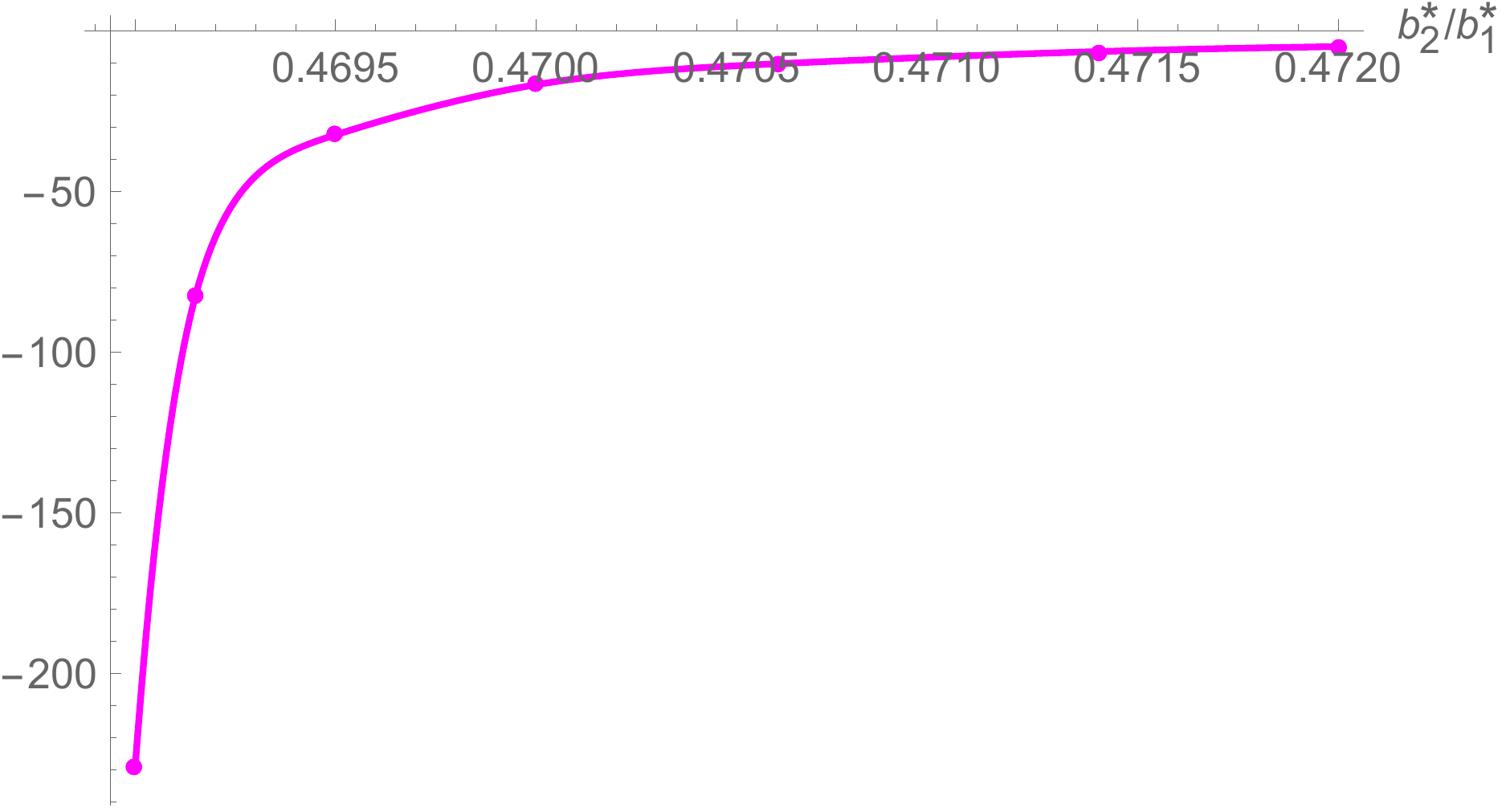}
    \caption{Variation of $Q_4$ on the curve $\Psi$ when $b_2^\ast/b_1^\ast \lesssim 0.468$ (left) and $b_2^\ast/b_1^\ast \gtrsim 0.468$ (right). It is clearly deduced that $Q_4$ does vanish for the allowed values $b^\ast$ on $\Psi$}
    \label{fig:Q412}
\end{figure}

\smallskip\par
The limit value $b_2^\ast/b_1^\ast \approx 0.468$ corresponds to the resonance $-\omega_1 + 2\, \omega_4 = 0$. When the sign of the coefficient of $Q_4$ is positive the bifurcation is supercritical (left picture in Fig. \ref{fig:Q412}), otherwise it is subcritical (right picture in Fig. \ref{fig:Q412}). The bifurcations occur at $M = 0$, equivalently at $\omega_3 = \omega_4$.
\\

At this point we face the analysis on the non-degeneracy conditions necessary to establish the persistence of KAM tori of dimensions $2$, $3$ and $4$, accordingly to the pattern of a Hamiltonian-Hopf bifurcation. We follow Theorem 4.27 in \cite{hanssmann2006local}. 
\\

We introduce the map
\[ \xi: I \rightarrow \left( \tilde{c}( I ), \Omega( I ), \tilde{\omega}_1( I ), 
\tilde{\omega}_2( I ) \right) \quad 
\mbox{with} \quad \tilde{\omega}_i( I ) = \displaystyle \frac{\partial F( I )}{\partial I_i}, \]
intending to prove that the matrix
\[ 
{\mathtt M} = \left(
\begin{array}{cccc}
\tilde{c}( I ) & \Omega( I ) & \tilde{\omega}_1( I ) & \tilde{\omega}_2( I ) \\[1ex]
\frac{\partial \tilde{c}( I )}{\partial I_1} & \frac{\partial \Omega( I )}{\partial I_1} & \frac{\partial \tilde{\omega}_1( I )}{\partial I_1} & \frac{\partial \tilde{\omega}_2( I )}{\partial I_1} \\[1ex]
\frac{\partial \tilde{c}( I )}{\partial I_2} & \frac{\partial \Omega( I )}{\partial I_2} & \frac{\partial \tilde{\omega}_1( I )}{\partial I_2} & \frac{\partial \tilde{\omega}_2( I )}{\partial I_2}
\end{array}
\right),
\]
spans $\mathbb{R} \times \mathbb{R} \times \mathbb{R}^2$ on the bifurcation line. 
\\

On the one hand, for $I = 0$ and $\omega_3 = \omega_4$ we get 
\[ 
\begin{array}{rcl}
\left \| \left( \tilde{c}( 0 ), \frac{\partial \tilde{c}}{\partial I_1}( 0 ), \frac{\partial \tilde{c}}{\partial I_2}( 0 )\right) \right \| &=& 
\frac{1}{2} \sqrt{Q_7^2 + Q_9^2}, \\[1.5ex]
\left \| \left( \Omega( 0 ), \frac{\partial \Omega}{\partial I_1}( 0 ), \frac{\partial \Omega}{\partial I_2}( 0 )\right) \right \| &=& 
\frac{1}{2} \sqrt{16 \omega_3^2 + Q_6^2 + Q_8^2},
\end{array}
\]
and aim to prove that they do not vanish on the bifurcation curve. Evaluating them on a discrete set of points chosen on $\Psi$ the minima of the vectors' norms are approximately  $0.2305...$ and $5.1393...$ respectively, thus we conclude that the norms are positive, that is, $( \tilde{c}, D \tilde{c}( 0 )) \neq 0$, $( \Omega( 0 ), D \Omega( 0 )) \neq 0$. (Notice also that $\omega_3 > 0$ along the bifurcation line, thus preventing the second vector to become null.)  
\\

On the other hand we compute the determinant 
\[ 
\left\vert
\begin{array}{cc}
\frac{\partial \tilde{\omega}_1}{\partial I_1}( 0 ) & \frac{\partial \tilde{\omega}_1}{\partial I_2}( 0 ) \\[1.3ex]
\frac{\partial \tilde{\omega}_2}{\partial I_1}( 0 ) & \frac{\partial \tilde{\omega}_2}{\partial I_2}( 0 )
\end{array}
\right\vert = Q_1 Q_2 - \frac{Q_5^2}{4}.
\]
In Fig. \ref{fig:detQQ} we depict this determinant along the bifurcation curve.
The determinant crosses the horizontal line only once at $b_2^\ast/b_1^\ast \approx 0.356$, corresponding to the vanishing of the resonance $\omega_1 - \omega_2 + \omega_4$. A small neighbourhood of this ratio as well as the other values where the non-linear normal form is not valid have to be deleted from the curve $\Psi$.
\begin{figure}[ht]
    \centering
    \includegraphics[width=0.475\textwidth]{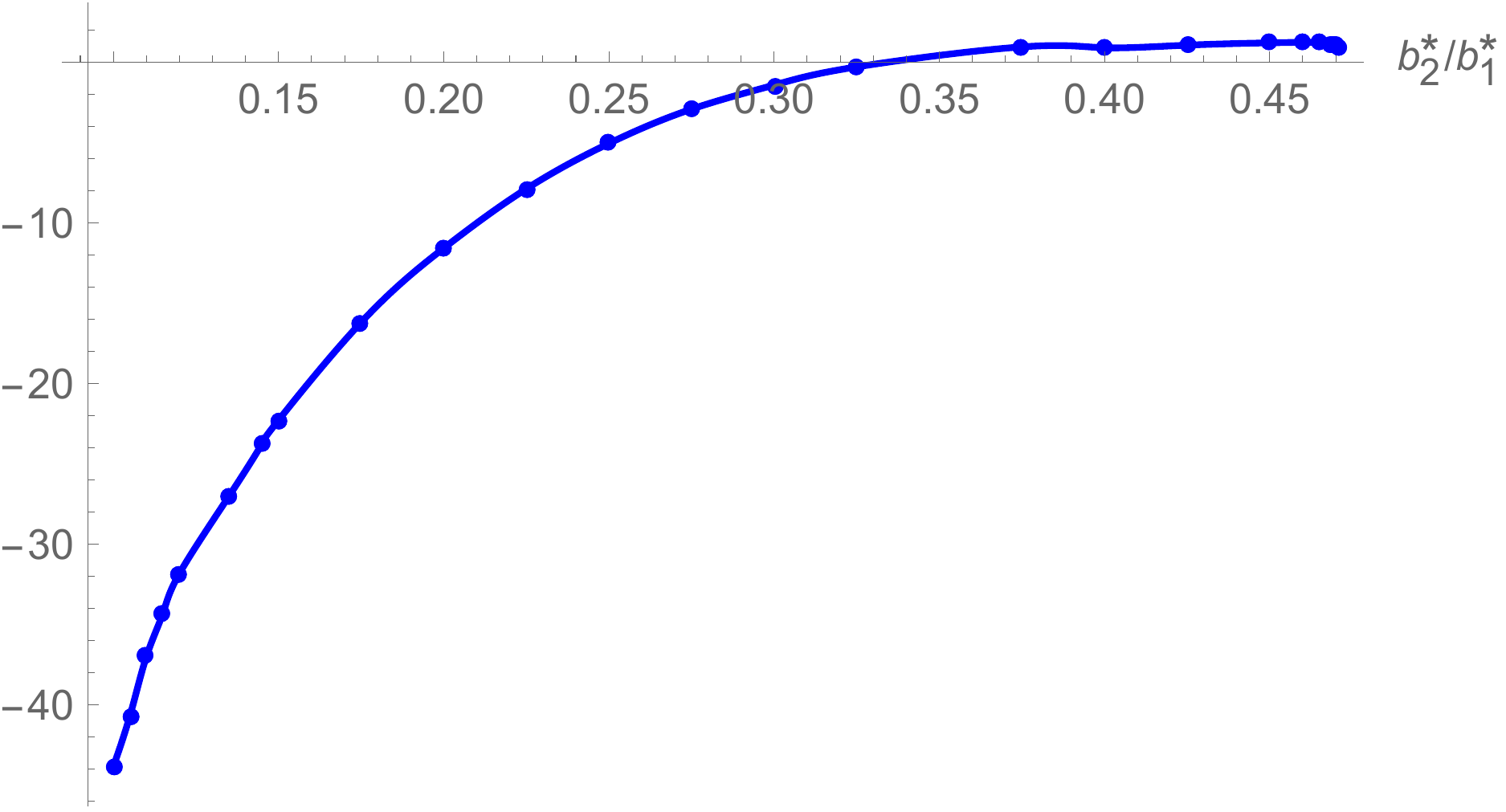}
    \caption{Variation of $Q_1 Q_2 - Q_5^2/4$ when $b^\ast$ is on the bifurcation
    line $\Psi$. The determinant becomes for $b_2^\ast/b_1^\ast \approx 0.356$, the value for which $\omega_1 - \omega_2 + \omega_4 = 0$}
    \label{fig:detQQ}
\end{figure}

The above analysis allows us to conclude that $\mathtt M$ spans $\mathbb{R} \times \mathbb{R} \times \mathbb{R}^2$ and a quasi-periodic Hamiltonian-Hopf bifurcation is displayed on the curve $\Psi$.
\\

The occurrence of the bifurcation guarantees the persistence of invariant tori of various dimensions associated to the bifurcation, but we describe with some detail what is going on. The persistence analysis is essentially based on KAM theory. Specifically there are invariant $2$, $3$ and $4$-tori. This $4$-tori are maximally dimensional tori, that is, Lagrangian tori, and its persistence is provided applying standard Kolmogorov’s condition, see also Theorem 4.15 in \cite{broer2007quasi}. Indeed, excluding the values of $b^\ast$ leading to zero or small denominators in the resonance cases studied above, is enough to conclude the persistence of the Lagrangian tori. However, for the $2$ and $3$-tori some Diophantine conditions on the frequencies
are needed to be imposed.  
\\

As said in Remark \ref{remposiS2c} the normal forms obtained throughout the analyses we do for the different Riemann ellipsoids, are useful to compute the first terms of the parametrisation of the related invariant tori, in case of persistence under perturbations. We begin with the $1$:$-1$ resonant $2$-torus, that we call $T_{\nu_0}$, existing for $I = 0$ and $\omega_3 = \omega_4$ and consider the family of invariant tori depending on parameters $\omega_3$, $\omega_4$ that we denote by $T_{\nu}$, meaning that when $\omega_3 = \omega_4$, then $\nu = \nu_0$. If $\Gamma( \nu )$ denotes the Floquet $( 4 \times 4 )$-matrix built from the subsystem in $x_3$, $x_4$, $y_3$  $y_4$ derived from $H_2$, after setting $I = 0$ $\Gamma( \nu )$ has a double pair of pure imaginary eigenvalues with a non-trivial nilpotent part when $\nu = \nu_0$. Standard KAM theory on the persistence of elliptic $2$-tori cannot be applied and one resorts to an adapted version of the KAM theorem \cite{BroerHooNaudot} to this special setting, and the results in \cite{broer2007quasi} apply. 
\\

With the aim of getting persistence of the invariant $2$ and $3$-tori, Diophantine conditions among the frequencies involved in $H_2$ are required. In particular, if $\omega_3^N$, and $\omega_4^N$ represent the imaginary parts of the eigenvalues of
$\Gamma( \nu )$ (the so-called normal frequencies), we write $\omega^N = ( \omega_3^N, \omega_4^N )$ observing that $\omega^N = ( \omega_3, \omega_4 )$ in the stable part of the bifurcation while $\omega^N = ( -\Im \omega_3, -\Im \omega_4 )$ in the unstable one. For instance, in the analysis of the persistence of $3$-tori we require the following condition to be satisfied:
for a constant $\tau > 1$ and for a positive parameter $\gamma$, we have
\[ 
\lvert \omega \cdot k  + 
\omega^N \cdot l \rvert \ge \gamma \lvert k \rvert ^{-\tau}, 
\]
for $\omega = ( \omega_1, \omega_2 )$, $k \in {\mathbb Z}^2 \setminus \{ 0 \}$ and $l \in {\mathbb Z}^2$
with $\lvert l \rvert \le 2$. (For an $n$-dimensional vector $v$ the norm $\lvert v \rvert$ refers to $v_1 + \ldots + v_n$.) Other related conditions are imposed to accomplished persistence of hyperbolic $2$-tori. The persistence of invariant $3$-tori is deeply analysed in \cite{hanssmann2006local} and in \cite{broer2007quasi}. As Theorem 4.27 of \cite{hanssmann2006local} applies in our setting, the persistence of these KAM tori is guaranteed under pertinent Diophantine conditions.
\\

In the supercritical piece of the bifurcation curve ($b_2^\ast/b_1^\ast \lesssim 0.468$, i.e., $Q_4 > 0$), when the point $b^\ast$ is in a narrow strip (neighbourhood) above the curve $\Psi$, a single invariant $2$-torus is elliptic and loses its stability when crossing the line (becoming parabolic) and turns hyperbolic $2$-tori below the bifurcation line. Above $\Psi$, emanating from this invariant torus, there is a two-dimensional Cantor family of normally elliptic invariant $3$-tori of large relative measure. (We remark that Diophantine conditions define Cantor sets.) This family of tori remains on the bifurcation line and when crossing it remaining below $\Psi$ it detaches from the hyperbolic torus moving away from it. 
\\

In the subcritical part of the bifurcation ($b_2^\ast/b_1^\ast \gtrsim 0.468$, thus $Q_4 < 0$)  the invariant elliptic $2$-torus follow the same pattern as in the supercritical case, but now for a point $b^\ast$ belonging to a strip above the curve $\Psi$, a family of elliptic $3$-tori and a family of hyperbolic $2$-tori meet in a subordinate quasi-periodic saddle-centre bifurcation. Persistence of these tori can be achieved. Additionally, there is a Cantor set of persistent parabolic tori on the bifurcation line involved in the subordinate saddle-centre bifurcation. When the point $b^\ast$ crosses $\Psi$ remaining close and below the curve the hyperbolic $2$-torus is not surrounded by invariant tori.
\\

Non-trivial monodromy is obtained in the supercritical case, in the family of invariant
$4$-tori when the bifurcating $2$-tori become hyperbolic.
\\

Apart from the various families of invariant tori of various dimensions, when the Riemann ellipsoid of type III is not of elliptic type and is non-degenerate, it possesses stable and unstable invariant manifolds attached to it (having the corresponding dimensions). The theory on the persistence of these manifolds and how they evolve is not really developed, excepting of course the Hamiltonian-Hopf bifurcation of equilibria. In the present work we do not need to handle this.
\end{proof}

\section{Global bifurcation of ellipsoids}
\label{global}

Apart from the local bifurcations accounted in the previous sections, there is a global bifurcation of equilibria due to an interplay between $S_2$-ellipsoids and type-III ellipsoids. This phenomenon is related to the fission theory described in the survey paper by Lebovitz \cite{lebovitz1998}. The underlying idea is that a rotating ellipsoid can undergo an evolution such that it loses stability to a non-axisymmetric disturbance,
and then splits into two ellipsoids. This theory was tackled by Liapunov and Poincar\'e, among others for some simplified ellipsoids. In his memoir \cite{chandrasekhar1969ellipsoidal} Chandrasekhar observed that type-III ellipsoids branch off from the ellipsoids of type $S$ along a curve of bifurcation, and this line
coincides with the line where the $S_2$-ellipsoids lose their stability. Our aim in the next paragraphs is to clarify these findings putting them in the perspective of our presentation. 
\\

The red curve ($G = 0$) mentioned in Section \ref{SectionS2} on $S_2$-ellipsoids is also a global bifurcation of equilibria, as we intend to explain now. Below the curve but close to it, the ellipsoids of type III are linearly stable and coexist with unstable $S_2$-ellipsoids with linearisation of type centre $\times$ centre $\times$ centre $\times$ saddle. On the bifurcation line, type-III ellipsoids disappear whereas ellipsoids of type $S_2$ become linearly stable above the curve. It looks like a Hamiltonian bifurcation of pitchfork type involving only equilibria in four degrees of  freedom, that is, Riemann ellipsoids. 
\\

In a bid to check whether this type of bifurcation actually occurs, one has to analyse the possible transitions between the two types of ellipsoids, proving that the equilibria of type $S_2$ can merge with equilibria of type III for some combinations of the parameters $b^\ast$. 
\\

The key is that the equilibrium points corresponding to the ellipsoids of type III can be obtained approximately from the normal form computed for the ellipsoids of type $S_2$, supporting the fact that $S_2$-ellipsoids bifurcate to ellipsoids of type III fitting the pattern of a Hamiltonian pitchfork bifurcation.
\\

\begin{thm}
There is a global bifurcation of Hamiltonian-pitchfork type involving $S_2$ and type-III ellipsoids.
\end{thm}

\begin{proof}
We begin with the vector field (\ref{reduced_system}). After equating it to zero and solving the system in the brown region in Fig. \ref{fig:S2}, i.e. below the red curve ($G = 0$), we find two solutions corresponding to $S_2$-equilibria and four more solutions close to $S_2$. The former match with equilibria related type-III ellipsoids, at least apparently, as the computations soon become unwieldy. As these critical are very close to the ellipsoids of type $S_2$ and in fact, they bifurcate from them (as we wish to prove), we intend to access their coordinates and stability character using the local approach described in Section \ref{SectionS2} for $S_2$.   
\\
    
The normal-form Hamiltonian (\ref{H4S2}) associated to co-parallel $S_2$-ellipsoids has been computed by expanding the original Hamiltonian around the equilibrium point corresponding to $S_2$-ellipsoids. Expressing the actions $I_i$ in terms of  coordinates, $X_i/Y_i$, $i = 1, 2, 3$, from the equations of motion of $H^4$ in the rectangular coordinates $Z$, we compute the critical points, obtaining $81$ solutions. All of them excepting three are non-isolated, so we exclude them. Besides, the null solution corresponds to the equilibrium point representing $S_2$. Thus, we have two critical points, say ${\mathcal E}_1$, ${\mathcal E}_2$, that are our candidates to be the (approximate) coordinates of type-III ellipsoids. Concretely we get $X_4^0 = \pm \imath \omega_4/( \sqrt{2} Q_4 )$ while the rest of $X_i^0$, $Y_i^0$ are $0$. (Recall that below the curve $G = 0$ the frequency $\omega_4 = \imath \bar{\omega}_4$ with $\bar{\omega}_4 < 0$.) Notice that $Q_4 \neq 0$ on the bifurcation curve, thus, by continuity it does not vanish if $b^\ast$ is close to the curve in the unstable side of the bifurcation.
\\

The eigenvalues of the matrices ${\mathcal L}_{{\mathcal E}_i}$ related to the linearisation around ${\mathcal E}_1$, ${\mathcal E}_2$ are:
\[ 
\pm \imath \omega_1 \pm \frac{\omega_4^2 Q_8}{4 Q_4}, \quad
\pm \imath \omega_2 \pm \frac{\omega_4^2 Q_9}{4 Q_4}, \quad
\pm \imath \omega_3 \pm \frac{\omega_4^2 Q_{10}}{4 Q_4}, \quad
\pm \sqrt{2}\, \omega_4. \\
\]

Coefficients $Q_8$, $Q_9$, $Q_{10}$ are pure imaginary whereas $Q_4$ is real and negative.
Moreover, $\omega_4$ is pure imaginary (and close to zero) but with negative imaginary part. As $\omega_i$, $i = 1, 2, 3$ are positive, it is straightforward to see that the other eigenvalues are also pure imaginary. Additionally the related eigenvectors form a basis of ${\mathbb R}^8$. Consequently ${\mathcal E}_1$, ${\mathcal E}_2$ are linearly equilibria of $H^4$ with linearisation centre $\times$ centre $\times$ centre $\times$ centre. 
\\

From Theorem \ref{S2Theorem} we know that ${\mathcal E}_1$, ${\mathcal E}_2$ are associated to the elliptic $4$-tori arising in the Hamiltonian pitchfork bifurcation displayed by the $S_2$-ellipsoids below the bifurcation curve, see remark \ref{remposiIIa}, in the sense that the approximate frequencies of these tori are readily obtained from the eigenvalues computed above. Pushing the normal form up to higher orders we would end up with better approximations of the coordinates of those invariant tori. However, thinking globally the points ${\mathcal E}_i$ are the equilibria corresponding to the type-III ellipsoids. Indeed, this affirmation is supported by the fact that type-III ellipsoids are linearly stable when $b^\ast$ is in ${\mathcal B}_{\rm III}$ but close to the boundary $G = 0$. This behaviour is corroborated with the numerical computations we have performed on a strip around $G = 0$ but below the curve. Specifically we have obtained the eigenvalues related with the critical points corresponding to the ellipsoids of type III using the vector field (\ref{reduced_system}). See also Figs. 3(b), 4(e), 4(f) of \cite{fasso2001stability}, where (spectral) stability is readily seen.
\\

We can conclude that the points corresponding to the ellipsoids of type III can be approximated from the normal form $H^4$ computed for the ellipsoids of type $S_2$, supporting the fact that the ellipsoid $S_2$ bifurcates to ellipsoids of type III. This is a bifurcation of equilibrium points experienced by a Hamiltonian system of four degrees of freedom and the bifurcation is of pitchfork type.
\\

Due to the discrete symmetries of Hamiltonian $H$ on the manifold ${\mathcal P}_{L, R}$ we can be more precise. Indeed we state that two linearly-stable $S_2$-ellipsoids split into four stable ellipsoids of type III by means of a Hamiltonian pitchfork bifurcation curve, $G = 0$, and the $S_2$-ellipsoids become unstable. The two linearly-stable $S_2$-ellipsoids correspond to the one with coordinates $( b^\ast, 0, \mu^{+}_\alpha( b^{\ast} ), \mu^{-}_\alpha( b^{\ast} ))$ and its adjoint $( b^\ast, 0, \mu^{-}_\alpha( b^{\ast} ), \mu^{+}_\alpha( b^{\ast} ))$, according to the notation in Table \ref{table:nonlin}.
 
\end{proof}

\section{Conclusions}
\label{conclusions}

This paper comes as a continuation of the studies on Riemann ellipsoids by addressing the analysis of parametric bifurcations, which was a pending issue in the literature. It has been observed a remarkable dynamical richness, with plenty of  bifurcations, most of them being of Hamiltonian-Hopf type in the Riemann ellipsoids of types I, II and III.
\\

We have analytically proved the existence of three kinds of bifurcations, that is, a quasi-periodic Hamiltonian pitchfork bifurcation of $S_2$-ellipsoids, two quasi-periodic saddle-centre bifurcations of type-II ellipsoids and a quasi-periodic Hamiltonian-Hopf bifurcation of type-III ellipsoids.
\\

The computations related to linear and non-linear normal forms and the associated transformations have been performed symbolically. Only a few checks have been done numerically, namely: (i) the non-degeneracy of some coefficients regarding the higher-order terms of a Hamiltonian function in normal form; (ii) the possibilities of introducing small denominators for some resonance combinations through the process based on Lie transformations; (iii) the determinants associated to the frequency maps that need to be non-zero in order to prove that KAM tori persist under small perturbations. In all these cases the numerical approximations have been carried out with {\sc Mathematica} performing the calculation with high precision, including the evaluation of some elliptic integrals. For the intermediate steps we have manipulated expressions in integer arithmetic.
\\

The existence of a global bifurcation involving $S_2$ and type-III ellipsoids was already noticed by Chandrasekhar \cite{chandrasekhar1969ellipsoidal} using a numerical approach. In this respect we have analytically clarified the underlying mechanism, relating it with the pitchfork bifurcation displayed by
the $S_2$-ellipsoids.
\\

The stability of both $S_2$ and $S_3$-ellipsoids has been established analytically. We have detected Liapunov stability for $S_2$-ellipsoids in part of their domain of existence and for $S_3$ in the whole ${\mathcal B}_{S_3}$. This strong stability was already established by Riemann but here we provide an alternative straightforward proof based on linear normal-forms transformations. The same procedure could be followed to prove Liapunov stability in the case of symmetric ellipsoids (with angular velocity and vorticity parallel to the same principal axis of the body) as a different approach to \cite{OlmosSousa}.
\\

From the linear analysis we have not detected Liapunov stability of types-I, II and III ellipsoids as the corresponding Hamiltonians of the linearised systems written in normal-form coordinates in case of linearly-stable points are always indefinite functions. 
\\

Irrotational ellipsoids of type $S_2$ and I have been properly analysed as three-degree-of-freedom Hamiltonian systems. We have obtained Liapunov stability for $S_2$-ellipsoids but linear stability or instability for type-I ellipsoids.
\\

In future we intend to delve deeper into the non-linear stability of Riemann ellipsoids from the viewpoint of formal stability. We believe that the application of the techniques described in \cite{Carcamo2021} would lead to the extension of the stability domains with exponentially long-time estimates compared to the ones established in \cite{fasso2001stability, fasso2014erratum}, as it happens in other problems with more than two degrees of freedom \cite{Carcamo2021sate}. More specifically with the use of Lie stabilty one can relax the quasi-convexity condition needed to achieve exponential estimates for time in Nekhorosev theory is applied in the setting of elliptic equilibria. 
\\

\begin{appendices}

\section{Expressions of $C_1$ and $C_2$}
\label{C1C2}

In this appendix we present analytical expressions of the integrals $C_1( x, y, z )$ and $C_2( x, y, z )$ appearing in \eqref{eq:c1}.
\\
 
Define
\[
\begin{array}{lcl}
P_1 &=& x ( x^2 - y^2 )^{3/2} ( x^2 - z^2 )^2 ( y^2 - z^2 )^2, \\[1ex]
P_2 &=& z ( x^2 - z^2 )^{3/2} ( x^2 - y^2 )^2 ( y^2 - z^2 )^2, \\[1ex]
E_1 &=& \displaystyle E \left( \arcsin\left(\frac{\sqrt{x^2 - y^2}}{x} \right) \,\bigg\vert \, 
        \frac{x^2 - z^2}{x^2 - y^2} \right), \\[3ex]
F_1 &=& \displaystyle F \left( \arcsin\left(\frac{\sqrt{x^2 - y^2}}{x} \right) \,\bigg\vert \,
        \frac{x^2 - z^2}{x^2 - y^2} \right), 
        
  \end{array} \]
\[
\begin{array}{lcl}
E_2 &=& \displaystyle E \left( \arcsin\left(\frac{\sqrt{z^2 - x^2}}{z} \right) \,\bigg\vert \,
		\frac{y^2 - z^2}{x^2 - z^2} \right), \\[3ex]
F_2 &=& \displaystyle F \left( \arcsin\left(\frac{\sqrt{z^2 - x^2}}{z} \right) \,\bigg\vert \,
        \frac{y^2 - z^2}{x^2 - z^2} \right),
 \end{array}\] 
where $E(\phi \mid k)$ stands for the incomplete elliptic integral of the second kind, for \( \phi \in (-\frac{\pi}{2}, \frac{\pi}{2}) \) and $k < 1$, i.e.
	\bas
	{E( \phi\mid k ) = \int _0^{\phi} \left(1 - k \sin ^2( \theta ) \right)^{1/2} d\theta}.\\
	\eas

\begin{lem} 
\label{C12lem}
The improper integrals of \eqref{eq:c1} are given by
\bas
	C_1 &=&
 \frac{4 \g \pi}{P_1} \Big(
 y z \sqrt{x^2 - y^2} ( x^2 - z^2 ) (z^2 - 2\,x^2 + y^2)
    \\ && \hspace*{0.9cm} + \,
  x \big( ( x^4 + y^2 z^2 ) ( y^2 + z^2 )  
  + x^2 ( y^4 - 6\, y^2 z^2 + z^4 )  \big) E_1
     \\ && \hspace*{0.9cm} - \,
     x ( y^2 - z^2 ) \big( y^2 z^2 + x^2 ( y^2 - 2\,z^2 ) \big) F_1 \Big),
\eas
and
\bas
	C_2 &=& \frac{-4\,\g \pi}{P_2}\Big(x y
		\sqrt{ x^2 - z^2 } ( z^2 - y^2 ) \big( y^2 z^2 + x^2 ( z^2 - 2\,y^2 )  \big) 
  \\&&
    \hspace*{1.1cm} + \, 2\,\imath z
		\big( y^4 z^4 - x^2 y^2 z^2 ( y^2 + z^2 ) + x^4 ( y^4 - y^2 z^2 + z^4 ) 
		\big) E_2  
  \\ &&
  \hspace*{1.1cm} - \, \imath z ( x^2 - y^2 ) \big( -y^2 z^2 
   ( y^2 + z^2 ) + x ( y^4 + z^4 ) \big) F_2  
    \Big).
	\eas
\end{lem}

\begin{proof}
We have used {\sc Mathematica} in order to evaluate the integrals. The expression of $C_1$ 
is achieved after applying some simplifications involving properties of elliptic integrals. 
In the case of $C_2$ {\sc Mathematica} is unable to compute the improper integral but it finds 
an antiderivative of the integrand function. Applying the Fundamental Theorem of Calculus 
and simplifying the resulting expressions using the properties of elliptic integrals, we arrive 
at the solution given above. The calculations are lengthy and have been placed in the 
{\sc Mathematica} file we attach. It is worth mentioning that the imaginary terms of $C_2$ 
cancel out for all values of $x$, $y$, $z$ and then, the integral makes sense always leading  to a
real integral.
\end{proof}
Having closed-form expressions of $C_1$, $C_2$, we have explicit formulae for the coordinates
of the five types of Riemann ellipsoids in $P_{L, R}$, for the boundaries in the parametric plane
accounting for the existence of all ellipsoids and for the expressions of the irrotational
and pitchfork bifurcation curves in the parametric plane for $S_2$-ellipsoids.
\\

The elliptic integrals are manipulated with {\sc Mathematica} using the functions {\tt EllipticE[]}
and {\tt EllipticF[]}, evaluating them only when necessary, for instance to produce the plots in
the manuscript or to approximate some formulae numerically. This is achieved using the high 
precision of {\sc Mathematica}.

\section{Coefficients of the linearisation matrix ${\mathcal L}$}
\label{CoefficientsL}

The linearisation matrix for $S$-type ellipsoids of Sections \ref{SectionS2}, \ref{SectionS3} can be written
as
\bas
   {\mathcal L} = \left( \begin{array}{cccccccc} 
   0 & 0 & 0 & 0 & \ell_{1,5} & \ell_{1,6} & 0 & 0
	\\ \noalign{\medskip}
   0 & 0 & 0 & 0 & \ell_{1,6} & \ell_{2,6} & 0 & 0
    \\ \noalign{\medskip}
   0 & 0 & 0 & 0 & 0 & 0 & \ell_{3,7} & \ell_{3,8}
    \\ \noalign{\medskip}
   0 & 0 & 0 & 0 & 0 & 0 & \ell_{3,8} & \ell_{4,8}
    \\ \noalign{\medskip}
   \ell_{5,1} & \ell_{5,2} & 0 & 0 & 0 & 0 & 0 & 0
	\\ \noalign{\medskip}
   \ell_{5,2} & \ell_{6,2} & 0 & 0 & 0 & 0 & 0 & 0
	\\ \noalign{\medskip}
   0 & 0 & \ell_{7,3} & \ell_{7,4} & 0 & 0 & 0 & 0
	\\ \noalign{\medskip}
   0 & 0 & \ell_{7,4} & \ell_{8,4} & 0 & 0 & 0 & 0 \end{array}
	\right),
	\eas 
\noindent with the exception of the irrotational regime in the $S_2$-ellipsoids, in which case
 it is
 \bas
   {\mathcal L}_I = \left( 
   \begin {array}{cccccc} 
   0 & 0 & 0 & \ell_{1,4} & \ell_{1,5} & 0
	\\ \noalign{\medskip}
   0 & 0 & 0 & \ell_{1,5} & \ell_{2,5} & 0
    \\ \noalign{\medskip}
   0 & 0 & 0 & 0 & 0 & \ell_{3,6}
    \\ \noalign{\medskip}
   \ell_{4,1} & \ell_{4,2} & 0 & 0 & 0 &0
    \\ \noalign{\medskip}
   \ell_{4,2} & \ell_{5,2} & 0 & 0 & 0 & 0
	\\ \noalign{\medskip}
   0 & 0 & \ell_{6,3} & 0 & 0 & 0
   \end{array}
	\right).
	\eas 
\begin{lem} 
\label{Coefellem}
The following holds with matrix ${\mathcal L}$:
\begin{itemize}
\item [(i)]
In the co-parallel regime of $S_2$-ellipsoids the coefficients of ${\mathcal L}$ in terms of $b^\ast$ are given by
 \bas
\ell_{1,5} &=& 1 - \frac{b_2^{\ast 2}}{b_1^{\ast 4} b_2^{\ast 4}
+ b_1^{\ast 2} + b_2^{\ast 2}},
\\
\ell_{1,6} &=& -\frac{b_1^\ast b_2^\ast}{b_1^{\ast 4} b_2^{\ast 4} + b_1^{\ast 2} + b_2^{\ast 2}},
\\
\ell_{2,6} &=& 1 - \frac{b_1^{\ast 2}}{b_1^{\ast 4} b_2^{\ast 4} + b_1^{\ast 2} + b_2^{\ast 2}},
\\
\ell_{3,7} &=& \frac{\sqrt{\pi\g} \alpha_1^{+}[b_1^\ast, b_2^\ast]}{
b_2^\ast b_1^{\ast 2}
\sqrt{b_1^{\ast 2} b_2^{\ast 4}-1}
\left( b_1^{\ast 4} b_2^{\ast 2} - 1 \right)^{3/4}
\left( b_1^{\ast 2} - b_2^{\ast 2} \right)^{5/2}},
\\[1ex]
\ell_{3,8} &=&
\frac{2 \sqrt{2\pi\g} \sqrt{\alpha_{13}[b_1^\ast, b_2^\ast] + \alpha_{14}[b_1^\ast, b_2^\ast] E_1 + \alpha_{15}[b_1^\ast, b_2^\ast] F_1}}{b_1^\ast ( b_1^{\ast 2} - b_2^{\ast 2} )^{5/2}
\sqrt{( b_1^{\ast 4} b_2^{\ast 2} - 1 )( b_1^{\ast 2} b_2^{\ast 4}-1 )}},
\\
\ell_{4,8} &=& \frac{\sqrt{\pi\g} \alpha_1^{-}[b_1^\ast, b_2^\ast]}{
b_2^\ast b_1^{\ast 2}
\sqrt{b_1^{\ast 2} b_2^{\ast 4}-1}
\left( b_1^{\ast 4} b_2^{\ast 2} - 1 \right)^{3/4}
\left( b_1^{\ast 2} - b_2^{\ast 2} \right)^{5/2}},
\\
\ell_{5,1} &=& \frac{-4 \g \pi b_2^\ast \left( 
\alpha_2[b_1^\ast, b_2^\ast] + \alpha_3[b_1^\ast, b_2^\ast] E_1 +
\alpha_4[b_1^\ast, b_2^\ast] F_1 \right)}{b_1^{\ast 2} ( b_1^{\ast 4} b_2^{\ast 2} - 1 )^{3/2} ( b_1^{\ast 2} b_2^{\ast 4} - 1 )^2 ( b_1^{\ast 2} - b_2^{\ast 2} )^2},
\\[1ex]
\ell_{5,2} &=& \frac{4 \g \pi \left( \alpha_5[b_1^\ast, b_2^\ast] + \alpha_6[b_1^\ast, b_2^\ast] E_1 + \alpha_7[b_1^\ast, b_2^\ast] F_1 \right)}{b_1^\ast ( b_1^{\ast 4} b_2^{\ast 2} - 1 )^{3/2} ( b_1^{\ast 2} b_2^{\ast 4} - 1 )^2 ( b_1^{\ast 2} - b_2^{\ast 2} )^2},
\\[1ex]
\ell_{6,2} &=&
\frac{-4 \g \pi \left(\alpha_8[b_1^\ast, b_2^\ast] + \alpha_9[b_1^\ast, b_2^\ast] E_1 + \alpha_{10}[b_1^\ast, b_2^\ast] F_1 \right)}{b_2^\ast 
( b_1^{\ast 4} b_2^{\ast 2} - 1 )^{3/2} ( b_1^{\ast 2} b_2^{\ast 4} - 1 )^{2} ( b_1^{\ast 2} - b_2^{\ast 2} )^2},
\\
\ell_{7,3} &=& \gamma^{+}[b_1^\ast, b_2^\ast],
\\ 
\ell_{7,4} &=&
\frac{-2 \sqrt{ 2\pi \g} b_1^\ast b_2^{\ast 3}
\sqrt{\alpha_{13}[b_1^\ast, b_2^\ast] + \alpha_{14}[b_1^\ast, b_2^\ast] E_1 + \alpha_{15}[b_1^\ast,b_2^\ast] F_1}}{( b_1^{\ast 2} b_2^{\ast 4} - 1 )^{5/2} ( b_1^{\ast 4} b_2^{\ast 2} - 1 )^{1/4} \sqrt{b_1^{\ast 2} - b_2^{\ast 2}}},
\\
\ell_{8,4} &=& \gamma^{-}[b_1^\ast, b_2^\ast],
\eas
\noindent
where 
{\footnotesize 
\[
\begin{array}{rcl}
{\alpha_1}^{\pm}[b_1^\ast, b_2^\ast] &=&
 \sqrt{b_1^{\ast 2} b_2^\ast + 1} \left( b_2^{\ast 2} - b_1^{\ast 4} b_2^\ast ( 2 - 3 b_2^{\ast 3} ) + b_1^{\ast 2} ( 1 - 2 b_2^{\ast 3} - b_2^{\ast 6} )\right) 
  \sqrt{\beta[b_1^\ast, b_2^\ast]} \\[0.8ex] &&
\mp \,
\sqrt{b_1^{\ast 2} b_2^\ast - 1} \left( b_2^{\ast 2} + b_1^{\ast 4} b_2^\ast ( 2 + 3\, b_2^{\ast 3} ) + b_1^{\ast 2} ( 1 + 2 b_2^{\ast 3} - b_2^{\ast 6}) \right)  \sqrt{-\beta[-b_1^\ast,-b_2^\ast]},
\end{array}
\]}
\noindent with 
{\footnotesize
\bas
\beta[b_1^\ast, b_2^\ast] &=& b_2^\ast ( b_2^{\ast 2} - b_1^{\ast 2} ) ( 2 b_1^{\ast 2} b_2^\ast - 1 ) \sqrt{b_1^{\ast 4} b_2^{\ast 2} - 1}
\\ &&
 + \, b_1^\ast ( b_1^{\ast 2} b_2^\ast + 1 ) \left( b_2^{\ast 2} + b_1^{\ast 2} \big( 1 + b_2^{\ast 3} ( b_1^{\ast 2} b_2^\ast - 3 ) \big) \right) E_1
\\ &&
 + \, b_1^\ast ( b_1^{\ast 2} b_2^{\ast 4} - 1 ) \left( b_2^{\ast 2} - b_1^{\ast 2} ( b_2^{\ast 3} - 1 ) \right) F_1
\eas}
\noindent 
and 
{\footnotesize
\bas
\alpha_2[b_1^\ast, b_2^\ast] &=&
b_2^\ast ( b_2^{\ast 2} - b_1^{\ast 2} ) \sqrt{b_1^{\ast 4} b_2^{\ast 2} - 1}
\\ &&
\times  
\left( b_2^{\ast 2} ( 15 b_1^{\ast 8} b_2^{\ast 4} + 9 b_1^{\ast 4} b_2^{\ast 2} - 8 ) + b_1^{\ast 2} ( 6 b_2^{\ast 6} + 9 ) - b_1^{\ast 6} b_2^{\ast 2} ( 14 b_2^{\ast 6} + 17 ) \right),
\\[1ex] 
\alpha_3[b_1^\ast, b_2^\ast] &=&
2 b_1^\ast \left( 4  b_2^{\ast 4} ( b_1^{\ast 12} b_2^{\ast 6} + 1 ) - 4 b_1^{\ast 4} ( 4 b_2^{\ast 6} + 1 )
+ 2 b_1^{\ast 6} b_2^{\ast 4} ( 7 b_2^{\ast 6} + 13 ) \right. \\&&
\left. \hspace*{0.6cm} - \, b_1^{\ast 2} b_2^{\ast 2} ( 3 b_2^{\ast 6} + 1 ) - 4 b_1^{\ast 8} b_2^{\ast 2} ( 4 b_2^{\ast 6} + 1 )
- b_1^{\ast 10} b_2^{\ast 6} ( 3 b_2^{\ast 6} + 1 ) \right),
\\
\alpha_4[b_1^\ast, b_2^\ast] &=&
b_1^\ast ( 1 - b_1^{\ast 2} b_2^{\ast 4} )
\left(
7 b_1^{\ast 4} + b_1^{\ast 2} b_2^{\ast 2} ( b_1^{\ast 6} + 4 )
- b_2^{\ast 4} ( 19 b_1^{\ast 6} + 9 ) \right. 
\\ &&
\left.
\hspace*{17.9ex} + \, b_1^{\ast 2} b_2^{\ast 8} ( 11 b_1^{\ast 6} + 8 ) - 8 b_1^{\ast 6} b_2^{\ast 10} + b_1^{\ast 4} b_2^{\ast 6} ( 6 - b_1^{\ast 6} ) \right),
\\[1ex]
\alpha_5[b_1^\ast, b_2^\ast] &=&
 b_2^\ast ( b_1^{\ast 2} - b_2^{\ast 2} ) \sqrt{b_1^{\ast 4} b_2^{\ast 2} - 1} \\ &&
\times \left( b_2^{\ast 2} ( 4 b_1^{\ast 8} b_2^{\ast 4} + 9 b_1^{\ast 4} b_2^{\ast 2}-5 )
+ b_1^{\ast 2} ( b_2^{\ast 6} + 4 ) - b_1^{\ast 6} b_2^{\ast 2} ( 5 b_2^{\ast 6} + 8 )
\right),
\\[1ex]
\alpha_6[b_1^\ast, b_2^\ast] &=&
b_1^\ast \left(- b_2^{\ast 4} ( 3 b_1^{\ast 12} b_2^{\ast 6} + 5 ) + b_1^{\ast 2} b_2^{\ast 2} ( b_2^{\ast 6} - 2 ) + b_1^{\ast 8} 
 b_2^{\ast 2} ( 4 b_2^{\ast 6} + 3 ) \right.
\\ && \hspace*{.5cm}
\left. + \, b_1^{\ast 10} b_2^{\ast 6} ( 5 b_2^{\ast 6} + 4 ) - 2 b_1^{\ast 6} b_2^{\ast 4} ( 7 b_2^{\ast 6} + 13 ) + b_1^{\ast 4} 
( 28 b_2^{\ast 6} + 5 ) \right),
\\[1ex]
\alpha_7[b_1^\ast, b_2^\ast] &=& 2 b_1^\ast ( b_1^{\ast 2} b_2^{\ast 4} - 1 )
\\ &&
\times \, \left(
b_2^{\ast 4} ( b_1^{\ast 8} b_2^{\ast 4} + 2 b_1^{\ast 2} b_2^{\ast 4} - 3 ) - 2 b_1^{\ast 6} b_2^{\ast 4} 
( b_2^{\ast 6} + 2 ) + b_1^{\ast 4} ( 4 b_2^{\ast 6} + 2 )
\right),
\\[1ex]
\alpha_8[b_1^\ast, b_2^\ast] &=& b_2^\ast ( b_1^{\ast 2} - b_2^{\ast 2} ) 
\sqrt{b_1^{\ast 4} b_2^{\ast 2} - 1}\, 
\\ &&
\times \, \left( b_2^{\ast 2} ( 2 b_1^{\ast 8} b_2^{\ast 4} - 9 b_1^{\ast 4} b_2^{\ast 2} + 3 ) - 3 b_1^{\ast 6} b_2^{\ast 2} ( b_2^{\ast 6} - 2 ) + b_1^{\ast 2} (5 b_2^{\ast 6} - 4 ) \right),
\\[1ex]
\alpha_9[b_1^\ast, b_2^\ast] &=&
b_1^\ast \left( -b_1^{\ast 2} ( b_1^{\ast 6} b_2^{\ast 2} + 3 b_1^{\ast 2} + 2 b_2^{\ast 2} )
+ b_2^{\ast 4} ( 22 b_1^{\ast 6} + 3 ) 
- 4 b_1^{\ast 4} b_2^{\ast 6} ( 2 b_1^{\ast 6} + 5 ) \right.
\\ &&
\hspace*{0.4cm} \left. + \, b_1^{\ast 2} b_2^{\ast 8} ( 4 b_1^{\ast 6} + 5 ) + b_1^{\ast 6} b_2^{\ast 10}
( b_1^{\ast 6} - 2 ) + b_1^{\ast 10} b_2^{\ast 12} \right),
\\[1ex]
\alpha_{10}[b_1^\ast, b_2^\ast] &=& 2 b_1^\ast ( 1 - b_1^{\ast 2} b_2^{\ast 4} )
\left( b_1^{\ast 2} ( b_1^{\ast 2} + 2 b_2^{\ast 2} ) - 2 b_2^{\ast 4} ( 2 b_1^{\ast 6} + 1 ) + b_1^{\ast 4} b_2^{\ast 6} ( b_1^{\ast 4} b_2^{\ast 2} + 2 )
\right),\\[1ex]
\gamma^{\pm}[b_1^\ast, b_2^\ast] &=&
\frac{-\sqrt{\pi\,\g} \, b_1^{\ast 2} b_2^{\ast 2}}{( b_1^{\ast 4} b_2^{\ast 2} - 1 )^{3/4} 
( b_1^{\ast 2} b_2^{\ast 4} - 1 )^{5/2} \sqrt{b_1^{\ast 2} - b_2^{\ast 2}}} \\
&&
\times \, \left( \sqrt{b_1^{\ast 2} b_2^{\ast} + 1} \, \alpha_{11}[b_1^\ast, b_2^\ast]
\sqrt{\alpha_{12}[b_1^\ast, b_2^\ast]} \right. \\[-1.5ex]
&& \left. \hspace*{0.5cm} \pm \, \sqrt{b_1^{\ast 2} b_2^{\ast} - 1}\, \alpha_{11}[b_1^\ast,-b_2^\ast] \sqrt{-\alpha_{12}[-b_1^\ast,-b_2^\ast]} \right), 
\\
\alpha_{11}[b_1^\ast, b_2^\ast] &=& -2 + b_2^{\ast 3} ( b_1^{\ast 4} b_2^{\ast 2} + 3 ) - b_1^{\ast 2} b_2^\ast ( b_2^{\ast 6} + 2 b_2^{\ast 3} - 1),
\\
\alpha_{12}[b_1^\ast, b_2^\ast] &=&
b_2^\ast ( b_2^{\ast 2}- b_1^{\ast 2} ) ( 2 b_1^{\ast 2} b_2^\ast - 1 ) \sqrt{ b_1^{\ast 4} b_2^{\ast 2}-1}
\\ &&
+ \, b_1^\ast ( b_1^{\ast 2} b_2^\ast + 1 ) \left( b_2^{\ast 2} + b_1^{\ast 2} \big( 1 + b_2^{\ast 3} ( b_1^{\ast 2} b_2^\ast - 3 ) \big) \right) E_1
\\ &&
+ \, b_1^\ast ( b_1^{\ast 2} b_2^{\ast 4} - 1 ) \left( b_2^{\ast 2} - b_1^{\ast 2} ( b_2^{\ast 3} - 1 ) \right) F_1,
\\
\alpha_{13}[b_1^\ast, b_2^\ast] &=&
b_2^\ast ( b_2^{\ast 2} - b_1^{\ast 2} ) \sqrt{b_1^{\ast 4} b_2^{\ast 2} - 1} ( 2 b_1^{\ast 8} b_2^{\ast 4} + 9 b_1^{\ast 4} b_2^{\ast 2} + 1 ),
\\
\alpha_{14}[b_1^\ast, b_2^\ast] &=&
b_1^\ast ( b_1^{\ast 4} b_2^{\ast 2} - 1 ) \, ( b_1^{\ast 8} b_2^{\ast 6} + b_1^{\ast 6} b_2^{\ast 2} + 8 b_1^{\ast 4} b_2^{\ast 4} + b_1^{\ast 2} + b_2^{\ast 2} ),
\\
\alpha_{15}[b_1^\ast, b_2^\ast] &=&
b_1^\ast ( 1 - b_1^{\ast 2} b_2^{\ast 4} ) ( b_1^{\ast 8} b_2^{\ast 6} + 3 b_1^{\ast 6} b_2^{\ast 2} + 6 b_1^{\ast 4} b_2^{\ast 4} + b_1^{\ast 2} + b_2^{\ast 2} ).
\eas}

\item [(ii)] The coefficients $\ell_{i,j}$ for the counter-parallel and irrotational cases of $S_2$, as well as for $S_3$-ellipsoids admit analogous expressions and are provided in the {\sc Mathematica} file.
\\
 	
\item [(iii)] The eigenvalues of the linearisation matrix ${\mathcal L}$ are $\pm \imath \omega_i$, $i = 1, \ldots, 4$, with
\begin{equation}
\begin{array}{lcl}
\omega_1 & = & \displaystyle
\mbox{$\frac{-\imath}{\sqrt{2}}$} \sqrt{\ell_{1,5} {\ell_{5,1}} + 2\, \ell_{1,6} {\ell_{5,2}} + \ell_{2,6} \ell_{6,2} - \sqrt{{\mathcal P}_1}},
\\[1.5ex]
{\omega_2} &=& \displaystyle
\mbox{$\frac{-\imath}{\sqrt{2}}$} \sqrt{\ell_{1,5} {\ell_{5,1}} + 2\, \ell_{1,6} \ell_{5,2} + \ell_{2,6} \ell_{6,2}
	+ \sqrt{{\mathcal{P}_1}}},
\\[1.5ex]
{\omega_3} &=& \displaystyle
\mbox{$\frac{-\imath}{\sqrt{2}}$} \sqrt{{\ell_{4,8}} \ell_{8,4} + 2\, \ell_{7,4} \ell_{3,8} + \ell_{3,7} \ell_{7,3}
	- \sqrt{{\mathcal{P}_2}}},
\\[.51ex]
{\omega_4} &=& \displaystyle
\mbox{$\frac{-\imath}{\sqrt{2}}$} \sqrt{{\ell_{4,8}} \ell_{8,4} + 2\, \ell_{7,4} \ell_{3,8} + \ell_{3,7} \ell_{7,3}
	+ \sqrt{{\mathcal{P}_2}}},
\end{array}
\label{omegas}
\end{equation}
where 
\bas
\mathcal{P}_1 &=& \ell_{1,5}^2 \ell_{5,1}^2 + 4\,\ell_{1,5} \ell_{5,2} \, 
 ( \ell_{1,6} \ell_{5,1} + \ell_{2,6} \ell_{5,2} ) - 2 \, \ell_{1,5} \ell_{2,6} \ell_{5,1} \ell_{6,2}
\\&& + \, \ell_{6,2} \left( 4\,\ell_{1,6} ( \ell_{1,6} \ell_{5,1} + {\ell_{2,6}} \ell_{5,2} ) +
\ell_{2,6}^2 \ell_{6,2} \right),
\\[1.2ex]
\mathcal{P}_2 &=& \ell_{3,7}^2 \ell_{7,3}^2 + 4\,\ell_{3,7} \ell_{7,4} \, 
 ( \ell_{3,8} \ell_{7,3} + \ell_{4,8} \ell_{7,4} ) - 2 \, \ell_{3,7} \ell_{4,8} \ell_{7,3} \ell_{8,4}
\\&& + \, \ell_{8,4} \left( 4\,\ell_{3,8} ( \ell_{3,8} \ell_{7,3} + \ell_{4,8} \ell_{7,4} ) +
\ell_{4,8}^2 \ell_{8,4} \right).
\eas

\item [(iv)] The eigenvalues of the linearisation matrix ${\mathcal L}_I$ for the irrotational case are $\pm \imath \omega_i$, $i = 1, \ldots, 3$, with
\begin{equation}
\begin{array}{lcl}
\omega_1& = & \displaystyle
\mbox{$\frac{-\imath}{\sqrt{2}}$} \sqrt{\ell_{1,4} {\ell_{4,1}} + 2\, \ell_{1,5}{\ell_{4,2}} + \ell_{2,5} \ell_{5,2}
	- \sqrt{\bar{{\mathcal{P}}_1}}},
\\[1.5ex]
{\omega_2} &=& \displaystyle
\mbox{$\frac{-\imath}{\sqrt{2}}$} \sqrt{\ell_{1,4} {\ell_{4,1}} + 2\, \ell_{1,5} \ell_{4,2} + \ell_{2,5} \ell_{5,2}
	+ \sqrt{\bar{{\mathcal{P}}_1}}},
\\[1.5ex]
{\omega_3} &=& \displaystyle \imath \sqrt{\ell_{3,6} \ell_{6,3}},
\end{array}
\label{omegasI}
\end{equation}
with 
\[
\begin{array}{rcl} \bar{\mathcal P}_1 &=& 
\ell_{1,4}^2 \ell_{4,1}^2 + 4\,\ell_{1,4} \ell_{4,2} \, 
 ( \ell_{1,5} \ell_{4,1} + \ell_{2,5} \ell_{4,2} ) - 2 \, \ell_{1,4} \ell_{2,5} \ell_{4,1} \ell_{5,2}
\\[1ex] && + \,\ell_{5,2} \left( 4\,\ell_{1,5} ( \ell_{1,5} \ell_{4,1} + \ell_{2,5} \ell_{4,2} ) +
\ell_{2,5}^2 \ell_{5,2} \right).
\\ \end{array} 
\]
\end{itemize}
\end{lem}
\begin{proof}
The calculations are straightforward but lengthy and they involve simplifications of
the intermediate formulae. They are given in the {\sc Mathematica} file.
\end{proof}

\section{Linear normal form}
\label{Linearization}

We describe Markeev's procedure \cite{Markeev} to determine the linear normal form corresponding to a linearisation matrix ${\mathcal L}$ for an elliptic point, that is, an equilibrium of a Hamiltonian system with $n$ degrees of freedom such that ${\mathcal L}$ is diagonalisable with eigenvalues $\pm \imath \omega_i$, such that $\omega_i > 0$ for all $i$. 
\\

\begin{remark}
Notice that this corresponds to a linearly stable equilibrium, which is slightly stronger than 
spectrally stable equilibrium, as we require ${\mathcal L}$ to be diagonalisable, avoiding possible nilpotent components in its decomposition that could lead to linear instability.
\\
\end{remark}

We use the eigenvalues and eigenvectors of ${\mathcal L}$ to construct a symplectic linear transformation, bringing ${\mathcal L}$ to diagonal form. We proceed as follows: (i) Due to the fact that for an elliptic equilibrium there are always $n$ pairs of complex eigenvectors, $v_i, \bar v_i$, $i = 1, \ldots, n$, we construct $n$ pairs of real eigenvectors, say $r_i = ( v_i - \bar v_i )/( 2\imath ), s_i = ( v_i + \bar v_i )/2$, $i = 1, \ldots, n$. (ii) We build the vectors $t_i = {\mathcal J}_{2n} \cdot s_i$, $i = 1, \ldots, n$, with
${\mathcal J}_{2n}$ the $( 2n \times 2n )$-skew symmetric matrix. Vectors $t_i, r_i$ are arranged so that $\bar{n}_i = r_i \cdot t_i$, $i = 1, \ldots, n$ are positive numbers. (iii) We define $k_i = 1/\sqrt{\bar{n}_i}$, for $i = 1, \ldots, n$.  
\\

Setting $n = 4$, the (real) matrix of the linear transformation is given by the symplectic matrix
\begin{equation}
{\mathcal T} = ( -k_1 s_1, \, -k_2 s_2, \, -k_3 s_3, \, -k_4 s_4, \, k_1 r_1, \, k_2 r_2, \, k_3 r_3, \, k_4 r_4 )^T.
\\
\label{eigenvectors1}
\end{equation}

The column vectors of ${\mathcal T}$ form a basis of $\mathbb{R}^8$ and the weights $k_i$ are strategically chosen so that the linear transformation is symplectic. For the irrotational $S_2$-ellipsoids we proceed similarly with $n = 3$, and the matrix ${\mathcal T}_I$ of the linear transformation is $(6 \times 6)$-dimensional.\\

Thus, the Hamiltonian function $H_2( u ) = -\frac{1}{2}\, u^T \cdot {\mathcal J} {\mathcal L} \cdot u$ corresponding to the linearised Hamiltonian system associated to $H$ given in (\ref{Hamiltonian}), can be brought to linear normal form. Indeed, defining
$z = ( x_1, x_2, x_3, x_4,$ $ y_1, y_2, y_3, y_4 )$, the change $u = {\mathcal T} \cdot z$ transforms the quadratic Hamiltonian to
$H_2( z ) = \frac{1}{2} \, z^T \cdot {\mathcal S} \cdot z,$
where ${\mathcal S} = - {\mathcal T}^T {\mathcal J}_8 {\mathcal L} {\mathcal T}$. In the irrotational regime we have $z = ( x_1, x_2, x_3, y_1, y_2, y_3 )$, $H_2(z) = \frac{1}{2} \, z^T \cdot {\mathcal S}_I \cdot z$, with ${\mathcal S}_I = - {\mathcal T}_I^T {\mathcal J}_6 {\mathcal L}_I {\mathcal T}_I$.\\

\begin{lem} 
\label{Coeftlem}
In the co-parallel regime of $S_2$-ellipsoids the symplectic transformation that brings $H_2$ in items ii. and iii. of Theorem \ref{S2Theorem} to linear normal form is given by the matrix
\begin{equation}
{\mathcal T} = ( t_{i,j} ), \quad
{\rm for} \quad i,j = 1, \ldots, 8,
\\
\label{eigenvectors2}
\end{equation}
\noindent where
\[
\begin{array}{rcl}

t_{1,5} &=& -\frac{\left( \ell_{1,6} \ell_{5,2} + \ell_{2,6} \ell_{6,2} + \omega_1^2  \right) \sqrt{\omega_1}}{\sqrt{\ell_{2,6} \left( \ell_{5,2}^2 - \ell_{5,1} \ell_{6,2} \right) - \ell_{5,1} \omega_1^2} \sqrt{\omega_1^2 - \omega_2^2}}, \\[2.5ex]

t_{1,6} &=& \frac{\left( \ell_{1,6} \ell_{5,2} + \ell_{2,6} \ell_{6,2} + \omega_2^2\right) \sqrt{\omega_2}}{\sqrt{\ell_{2,6} \left( \ell_{5,1} \ell_{6,2} - \ell_{5,2}^2 \right) + \ell_{5,1} \omega_2^2} \sqrt{\omega_1^2 - \omega_2^2}}, \\[2.5ex]

t_{2,5} &=& \frac{\left( \ell_{1,6} \ell_{5,1} + \ell_{2,6} \ell_{5,2} \right)  
 \sqrt{\omega_1}}{\sqrt{\ell_{2,6} \left( \ell_{5,2}^2 - \ell_{5,1} \ell_{6,2} \right)  - \ell_{5,1} \omega_1^2} \sqrt{\omega_1^2 - \omega_2^2}}, \\[2.5ex]

t_{2,6} &=&
 -\frac{\left( \ell_{1,6} \ell_{5,1} + \ell_{2,6} \ell_{5,2} \right) \sqrt{\omega_2}}{\sqrt{\ell_{2,6} \left( \ell_{5,1} \ell_{6,2} - \ell_{5,2}^2 \right) + \ell_{5,1} \omega_2^2} \sqrt{\omega_1^2 - \omega_2^2}}, \\[2.5ex]

t_{3,4} &=&
\frac{ \left( \ell_{3,8} \ell_{7,4} + \ell_{4,8} \ell_{8,4} + \omega_4^2 \right) \sqrt{\left( \ell_{3,7} \ell_{4,8} - \ell_{3,8}^2 \right) \ell_{7,3} + \ell_{4,8} \omega_4^2}}{\left( \ell_{3,8} \ell_{7,3} + \ell_{4,8} \ell_{7,4} \right) \sqrt{\omega_3^2 - \omega_4^2}}, \\[2.5ex]

t_{3,7} &=& -\frac{\left( \ell_{3,8} \ell_{7,4} + \ell_{4,8} \ell_{8,4} + \omega_3^2 \right) \sqrt{\omega_3}}{\sqrt{\ell_{4,8} \left( \ell_{7,4}^2 - \ell_{7,3} \ell_{8,4}  \right) - \ell_{7,3} \omega_3^2} \sqrt{\omega_3^2 - \omega_4^2}}, \\[2.5ex]

t_{4,4} &=&
-\frac{ \sqrt{\ell_{7,3} \left( \ell_{3,7} \ell_{4,8} - \ell_{3,8}^2 \right) + \ell_{4,8} \omega_4^2}}{\sqrt{\omega_3^2 - \omega_4^2}}, \\[2.5ex]

t_{4,8} &=&
\frac{\left( \ell_{3,8} \ell_{7,3} + \ell_{4,8} \ell_{7,4} \right) \sqrt{\omega_3}}{\sqrt{\ell_{4,8} \left( \ell_{7,4}^2 - \ell_{7,3} \ell_{8,4} \right) - \ell_{7,3}\omega_3^2} \sqrt{\omega_3^2 - \omega_4^2}}, \\[2.5ex]

t_{5,1} &=&
\frac{\sqrt{\ell_{2,6} \left( \ell_{5,2}^2 - \ell_{5,1} \ell_{6,2} \right) - \ell_{5,1}\omega_1^2}}{\sqrt{\omega_1 ( \omega_1^2 - \omega_2^2 )}}, \\[2.5ex]

\end{array}\]

\[
\begin{array}{rcl}

t_{5,2} &=&
\frac{\sqrt{\ell_{2,6} \left( \ell_{5,1} \ell_{6,2} - \ell_{5,2}^2 \right) + \ell_{5,1} \omega_2^2}}{\sqrt{\omega_2 (\omega_1^2 - \omega_2^2 )}}, \\[2.5ex]

t_{6,1} &=& 
\frac{\ell_{1,6} \left( \ell_{5,1} \ell_{6,2} - \ell_{5,2}^2 \right) - \ell_{5,2}\omega_1^2}{\sqrt{\ell_{2,6} \left( \ell_{5,2}^2 - \ell_{5,1} \ell_{6,2} \right) - \ell_{5,1} \omega_1^2} \sqrt{\omega_1 ( \omega_1^2 - \omega_2^2 )}}, \\[2.5ex]

t_{6,2} &=&
\frac{\ell_{1,6} \left( \ell_{5,2}^2 - \ell_{5,1} \ell_{6,2} \right) + \ell_{5,2}\omega_2^2}{\sqrt{\ell_{2,6} \left( \ell_{5,1} \ell_{6,2} - \ell_{5,2}^2 \right) + \ell_{5,1} \omega_2^2} \sqrt{\omega_2 ( \omega_1^2 - \omega_2^2)}}, \\[2.5ex] 

t_{7,3} &=&
\frac{\sqrt{\ell_{4,8} \left( \ell_{7,4}^2 - \ell_{7,3} \ell_{8,4} \right) - \ell_{7,3}\omega_3^2}}{\sqrt{\omega_3 ( \omega_3^2 - \omega_4^2 )}}, \\[2.5ex]

t_{7,8} &=&
\frac{\sqrt{\ell_{7,3} \left( \ell_{3,7} \ell_{4,8} - \ell_{3,8}^2 \right) + \ell_{4,8}\omega_3^2}}{\sqrt{\ell_{3,8}^2 - \ell_{3,7} \ell_{4,8}} \sqrt{\omega_3^2 - \omega_4^2}}, \\[2.5ex]

t_{8,3} &=&
\frac{\ell_{3,8} \left( \ell_{7,3} \ell_{8,4} - \ell_{7,4}^2 \right) - \ell_{7,4}\omega_3^2}{\sqrt{\ell_{4,8} \left( \ell_{7,4}^2 - \ell_{7,3} \ell_{8,4} \right)
- \ell_{7,3} \omega_3^2} \sqrt{\omega_3 (\omega_3^2 - \omega_4^2 )}}, \\[2.5ex]

t_{8,8} &=&
\frac{\sqrt{\ell_{8,4} \left( \ell_{3,7} \ell_{4,8} - \ell_{3,8}^2 \right) + \ell_{3,7}\omega_3^2}}{\sqrt{\ell_{3,8}^2 - \ell_{3,7} \ell_{4,8}} \sqrt{\omega_3^2 - \omega_4^2}}, 
\end{array}
\]
and the rest of the $t_{i,j}$ are zero.
\end{lem}

\begin{proof}
We apply the procedure due to Markev as described in this appendix, but making the extra change $x_4 \rightarrow \sqrt{\omega_4} x_4$, $y_4 \rightarrow y_4/\sqrt{\omega_4}$ in order to accommodate the analysis including the pitchfork bifurcation line. After some simplifications on the resulting expressions that we have made with {\sc Mathematica} we end up with the entries $t_{i,j}$.
\end{proof}

The coefficients of the transformation matrix for the case of counter-parallel $S_2$-ellipsoids, as well as irrotational ones, are analogous and can be found in the {\sc Mathematica} file.
\\

\end{appendices}

\bibliography{Riemann Ellipsoids} 

\end{document}